\documentclass[11pt,letterpaper]{article}

\usepackage{_style}

\usepackage{quoting}
\quotingsetup{vskip=0pt}

\usepackage[compact]{titlesec}
\usepackage{titletoc}
\usepackage[compact]{titlesec}
\titlespacing{\section}{0pt}{1.5ex}{0ex}
\titlespacing{\subsection}{0pt}{1.5ex}{0ex}
\titlespacing{\subsubsection}{0pt}{1ex}{0ex}
\titlespacing{\paragraph}{0pt}{1.5ex}{1ex}
\allowdisplaybreaks

\usepackage{hyperref}

\usepackage{url}
\usepackage{bbm}

\usepackage{svg}
\usepackage{tcolorbox}
\tcbuselibrary{minted,breakable,xparse,skins}

\definecolor{bg}{gray}{0.95}
\definecolor{softgray}{HTML}{fff3d6}
\definecolor{softblue}{HTML}{c1cfa1}
\definecolor{softred}{HTML}{ffe3e3}

\DeclareTCBListing{mintedbox}{O{}m!O{}}{%
  breakable=true,
  listing engine=minted,
  listing only,
  minted language=#2,
  minted style=default,
  minted options={%
    gobble=0,
    breaklines=true,
    breakafter=,,
    fontsize=\scriptsize,
    numbersep=8pt,
    #1},
  boxsep=0pt,
  left skip=0pt,
  right skip=0pt,
  left=25pt,
  right=0pt,
  top=3pt,
  bottom=3pt,
  arc=5pt,
  leftrule=0pt,
  rightrule=0pt,
  bottomrule=2pt,
  toprule=2pt,
  colback=bg,
  colframe=orange!70,
  enhanced,
  overlay={%
    \begin{tcbclipinterior}
    \fill[orange!20!white] (frame.south west) rectangle ([xshift=20pt]frame.north west);
    \end{tcbclipinterior}},
  #3}

\newcommand{\pbar}{\overline{p}}
\newcommand{\ks}{k^{\star}}
\usepackage{thm-restate}

\usepackage{pgfplots}
\pgfplotsset{compat=1.18}
\usepackage[dvipsnames]{xcolor}

  \usepackage{nth}
  \usepackage{intcalc}

  \newcommand{\cAAAI}[1]{AAAI\ Conference\ on\ Artificial (AAAI)}

\newif\ifarxiv
\arxivtrue

\begin{document}
\author[1]{Feyza Duman Keles}
\author[1]{Lisa Hellerstein}
\author[2]{Kunal Marwaha}
\author[1]{Christopher Musco}
\author[3]{Xinchen Yang}
\affil[1]{New York University\thanks{\texttt{\{fd2135,lisa.hellerstein,cmusco\}@nyu.edu}}}
\affil[2]{University of Chicago\thanks{\texttt{kmarw@uchicago.edu}}}
\affil[3]{University of Maryland\thanks{\texttt{xcyang@cs.umd.edu}}}

\title{An Exact Algorithm for the Unanimous Vote Problem}
\date{\vspace{-2em}}
\maketitle
\pagenumbering{arabic}
\begin{abstract}
Consider $n$ independent, biased coins, each with a known probability of heads.  
Presented with an ordering of these coins, flip (i.e., toss) each coin once, in that order, until we have observed both a \emph{head} and a \emph{tail}, or flipped all coins.
The Unanimous Vote problem asks us to find the ordering that minimizes the expected number of flips. 
Gkenosis et al.~\cite{gkenosis2018stochastic} 
gave a polynomial-time $\phi$-approximation algorithm for this problem, where $\phi \approx 1.618$ is the golden ratio.  
They left open whether the problem was $\NP$-hard.  
We answer this question by giving an exact algorithm that runs in time $O(n \log n)$.  The Unanimous Vote problem is an instance of the more general Stochastic Boolean Function Evaluation problem: it thus becomes one of the only such problems known to be solvable in polynomial time.
Our proof uses simple interchange arguments to show that the optimal ordering must be close to the ordering produced by a natural greedy algorithm.
Beyond our main result, we compare the optimal ordering with the best adaptive strategy, proving a tight adaptivity gap  of $1.2\pm o(1)$ for the Unanimous Vote problem.

\end{abstract}



\thispagestyle{empty}
\clearpage
\newpage

\section{Introduction}
\label{sec:introduction}
\setcounter{page}{1}
Suppose you have $n$ independent biased coins with probability of heads $p_1, \ldots, p_n \in [0,1]$.   
You arrange the coins in some order, and then flip (that is, toss) each coin, in that order, until either you have observed 
a \emph{head} or you have flipped all coins.  In what order should you flip the coins to minimize your expected number of flips?  The answer is obvious: flip the coins in decreasing order of their $p_i$ values.

Next, consider a slight variation of this problem, where you flip coins until you have observed \emph{both} a head and a tail, or until you have flipped all coins.  Again, in what order should you flip the coins to minimize your expected number of flips?  Now the answer is not obvious at all.  This problem, introduced by Gkenosis et al.~\cite{gkenosis2018stochastic}, is the one we consider in this paper.
Viewing the coin flips as yes/no votes of $n$ stochastic voters, Gkenosis et al.\ framed the problem as 
seeking to determine whether the $n$ votes are unanimously yes, unanimously no, or not unanimous. We call it the \emph{Unanimous Vote problem}.

Gkenosis et al.\ gave a simple \emph{2-approximation} algorithm for the problem, which produces an ordering whose expected number of coin flips is at most \emph{twice} the optimal \cite{gkenosis2018stochastic}.
They also gave a more involved $\phi$-approximation algorithm for the problem, where $\phi \approx 1.618$ is the golden ratio. 
They left as an open question whether the Unanimous Vote problem is $\NP$-hard.
Indeed, this problem is a \emph{non-adaptive} Stochastic Boolean Function Evaluation (SBFE) problem (see e.g., ~\cite{IbarakiKameda:1984,KrishnamurthyBoralZaniolo:1986,HellersteinKletenikLiu:2022,grammel2022algorithms,GhugeGuptaNagarajan:2024,HarrisNagarajanTan:2025,Unluyurt:2025,nielsen2025nonadaptive}).  Very few of these problems are known to be solvable exactly in polynomial time.

Our main contribution is a simple, exact, $O(n \log n)$-time algorithm for the Unanimous Vote problem, which resolves the question of Gkenosis et al. 
We obtain our result by showing that the structure of an optimal ordering must be very similar to a natural \emph{greedy} solution to the problem. 
To produce the optimal ordering, our algorithm constructs this greedy solution, and rearranges the positions of at most two coins.

In addition to our main result, we prove that the greedy solution itself is near optimal, flipping at most one more coin in expectation than the optimal ordering. We also prove a result on \emph{adaptive strategies}, which can choose the next coin based on the outcomes of previous flips.
For the Unanimous Vote problem, there is a simple polynomial time optimal adaptive strategy \cite{gkenosis2018stochastic}. The algorithm presented in this paper, in contrast, is \emph{non-adaptive}.
In general, the non-adaptive ordering for an SBFE problem requires at least as many flips as the best adaptive strategy, and there has been interest in quantifying the benefit of adaptivity \cite{HellersteinKleteniketal:2022,grammel2022algorithms}. 
To that end,
we prove that the  \emph{adaptivity gap}, i.e., the worst possible ratio of the expected number of flips of the optimal non-adaptive ordering to the expected number of flips of the optimal adaptive 
strategy, is $1.2 \pm o(1)$. 

\begin{figure}[h!]
    \centering
    \vspace{.5em}
\begin{center}
\begin{tikzpicture}[scale=0.6, every node/.style={transform shape}]

\def\circleRadius{0.6}
\def\spacing{2.3}
\def\rowgap{2.5}

\foreach \i/\dir/\len in {
    0/up/2.0,
    1/down/2.0,
    2/up/0.8,
    3/down/1.8
} {

\node at (-\rowgap,0) {\Huge cost};
\draw[thick] (-1.4,1) -- (-1.4,-1); 
\draw[thick] (-1.4,1) -- (-1.1,1); 
\draw[thick] (-1.4,-1) -- (-1.1,-1); 

\draw[thick] (4*\spacing+1.4,1) -- (4*\spacing+1.4,-1); 
\draw[thick] (4*\spacing+1.4,1) -- (4*\spacing+1.1,1); 
\draw[thick] (4*\spacing+1.4,-1) -- (4*\spacing+1.1,-1); 

\node at (4*\spacing+2.7,0) {\Huge $=\,$2.87};

\node at (-\rowgap,0-\rowgap) {\Huge cost};
\draw[thick] (-1.4,1-\rowgap) -- (-1.4,-1-\rowgap); 
\draw[thick] (-1.4,1-\rowgap) -- (-1.1,1-\rowgap); 
\draw[thick] (-1.4,-1-\rowgap) -- (-1.1,-1-\rowgap); 

\draw[thick] (4*\spacing+1.4,1-\rowgap) -- (4*\spacing+1.4,-1-\rowgap); 
\draw[thick] (4*\spacing+1.4,1-\rowgap) -- (4*\spacing+1.1,1-\rowgap); 
\draw[thick] (4*\spacing+1.4,-1-\rowgap) -- (4*\spacing+1.1,-1-\rowgap); 

\node at (4*\spacing+2.7,-\rowgap) {\Huge $=\,$2.75};

\draw[fill=Goldenrod!90!orange, thick] (0*\spacing,0) circle (1cm);
\draw[fill=Goldenrod!60!white] (0*\spacing,0) circle (0.8cm);
\node at (0*\spacing,0) {\Huge 0.9};
\fill[white, opacity=0.3] (-0.4,0.4) circle (0.2cm);

\draw[fill=Maroon!50!white, thick] (1*\spacing,0) circle (1cm);
\draw[fill=Maroon!20!white] (1*\spacing,0) circle (0.8cm);
\node at (1*\spacing,0) {\Huge 0.4};
\fill[white, opacity=0.2] (1*\spacing-0.4,0.4) circle (0.2cm);

\draw[fill=gray!60!white, thick] (2*\spacing,0) circle (1cm);
\draw[fill=gray!30!white] (2*\spacing,0) circle (0.8cm);
\node at (2*\spacing,0) {\Huge 0.8};
\fill[white, opacity=0.2] (2*\spacing-0.4,0.4) circle (0.2cm);

\draw[fill=Fuchsia!50!white, thick] (3*\spacing,0) circle (1cm);
\draw[fill=Fuchsia!20!white] (3*\spacing,0) circle (0.8cm);
\node at (3*\spacing,0) {\Huge 0.5};
\fill[white, opacity=0.2] (3*\spacing-0.4,0.4) circle (0.2cm);

\draw[fill=olive!90!green, thick] (4*\spacing,0) circle (1cm);
\draw[fill=olive!50!white] (4*\spacing,0) circle (0.8cm);
\node at (4*\spacing,0) {\Huge 0.6};
\fill[white, opacity=0.15] (4*\spacing-0.4,0.4) circle (0.2cm);


\draw[fill=Goldenrod!90!orange, thick] (0*\spacing, -\rowgap) circle (1cm);
\draw[fill=Goldenrod!60!white] (0*\spacing, -\rowgap) circle (0.8cm);
\node at (0*\spacing, -\rowgap) {\Huge 0.9};
\fill[white, opacity=0.3] (0*\spacing-0.4,0.4-\rowgap) circle (0.2cm);

\draw[fill=Maroon!50!white, thick] (1*\spacing, -\rowgap) circle (1cm);
\draw[fill=Maroon!20!white] (1*\spacing, -\rowgap) circle (0.8cm);
\node at (1*\spacing, -\rowgap) {\Huge 0.4};
\fill[white, opacity=0.2] (1*\spacing-0.4,0.4-\rowgap) circle (0.2cm);

\draw[fill=Fuchsia!50!white, thick] (2*\spacing,-\rowgap) circle (1cm);
\draw[fill=Fuchsia!20!white] (2*\spacing,-\rowgap) circle (0.8cm);
\node at (2*\spacing,-\rowgap) {\Huge 0.5};
\fill[white, opacity=0.2] (2*\spacing-0.4,0.4-\rowgap) circle (0.2cm);

\draw[fill=olive!90!green, thick] (3*\spacing, -\rowgap) circle (1cm);
\draw[fill=olive!50!white] (3*\spacing, -\rowgap) circle (0.8cm);
\node at (3*\spacing, -\rowgap) {\Huge 0.6};
\fill[white, opacity=0.15] (3*\spacing-0.4,0.4-\rowgap) circle (0.2cm);

\draw[fill=gray!60!white, thick] (4*\spacing, -\rowgap) circle (1cm);
\draw[fill=gray!30!white] (4*\spacing, -\rowgap) circle (0.8cm);
\node at (4*\spacing, -\rowgap) {\Huge 0.8};
\fill[white, opacity=0.2] (4*\spacing-0.4,0.4-\rowgap) circle (0.2cm);

}
\end{tikzpicture}
\vspace{-.5em}
\end{center}
    \caption{\small Two example orderings for an instance of the Unanimous Vote problem. The second ordering has a smaller number of expected flips (cost). One intuitive strategy is to alternate between the highest and lowest bias coins, as in the first sequence. As we can see, this is not always optimal. The second sequence is the one returned by a natural greedy algorithm described below. For this instance, it happens to be optimal.}
    \label{fig:problem_example}
\end{figure}

\subsection{Overview of our algorithm}

The Unanimous Vote problem asks for a fixed ordering of the coins that minimizes the expected number of flips required until both heads and tails are seen, or all coins have been flipped. 
\Cref{fig:problem_example} compares two orderings on a small example instance. In \Cref{sec:structural,sec:optimal_algo} we prove our main result:
\begin{theorem}\label{thm:mainresult_intro}
    There is an algorithm (\Cref{alg:fastermodifiedgreedy}) that solves the Unanimous Vote problem (i.e., computes the minimum cost ordering) in time $O(n \log n)$.
\end{theorem}
To understand our approach to \Cref{thm:mainresult_intro}, it is helpful to define \emph{biased blocks} of an ordering, a concept that we introduce, and which is central to our analysis. To define biased blocks, 
fix an ordering, and consider some position $k$.
Suppose we flip coins according to that ordering, terminating as soon as we observe a head and a tail, or have flipped all coins.  We will only flip the $k^\text{th}$ coin if the first $k-1$ coins all come up heads, or all come up tails.  
Using the convention that $\text{heads}=1$, and $\text{tails}=0$, we call the
$k^\text{th}$ position \emph{1-biased} if the probability that the first $k-1$ coins all come up heads is greater than the probability that they all come up tails, \emph{0-biased} if the opposite relationship holds, and \emph{unbiased} if the probabilities are equal.

The $n$ positions in the ordering can be partitioned into ``blocks'', according to their biases.  Each block is a maximal set of contiguous positions of the same bias.  See~\Cref{fig:biasedblocksunsorted} for an example.

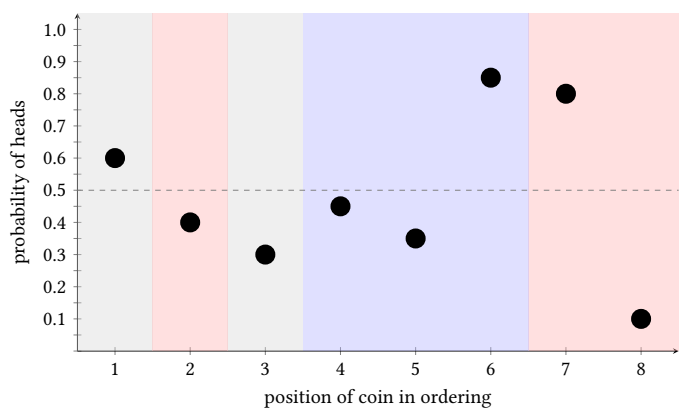
\begin{figure}[h!]
    \centering    
\begin{tikzpicture}[scale=.7]
\begin{axis}[
    scaled y ticks=false,
    width=13cm,
    height=8cm,
    ymin=0, ymax=1.05,
    ytick={0, 0.05, 0.10, 0.15, 0.20, 0.25, 0.30, 0.35, 0.40, 0.45, 0.50, 0.55, 0.60, 0.65, 0.70, 0.75, 0.80, 0.85, 0.90, 0.95, 1.0},
    yticklabels={ , ,0.1 , ,0.2 , ,0.3 , , 0.4, , 0.5, ,0.6 , ,0.7 , , 0.8, ,0.9 , ,1.0 },
    xtick={1,2,3,4,5,6,7,8},
    xticklabels={1, 2, 3, 4, 5, 6, 7, 8},
    enlargelimits=0.1,
    axis lines=left,
    ylabel={probability of heads},
    xlabel={position of coin in ordering},
    scatter/classes={
        a={mark=o,draw=Gray,fill=Gray,mark size=2.7pt,line width=5.4pt},
        c={mark=o,draw=MidnightBlue,fill=MidnightBlue,mark size=2.7pt,line width=5.4pt},
        b={mark=o,draw=BrickRed,fill=BrickRed,mark size=2.7pt,line width=5.4pt},
        d={mark=none},
        e={mark=o,draw=black,fill=black,mark size=2.7pt,line width=5.4pt}
    }
]
\addplot[
    only marks,
    scatter,
    scatter src=explicit symbolic
]
coordinates {
    (0.5,0) [d]
    (1,0.6) [e]
    (2,0.4) [e]
    (3,0.3) [e]
    (4,0.45) [e]
    (5,0.35) [e]
    (6,0.85) [e]
    (7,0.8) [e]
    (8,0.1) [e]
    (8.5,0) [d]
};
\addplot [
    draw=none,
    fill=gray!30,
    opacity=0.4
] coordinates {
    (0.5,0) (0.5,1.05) (1.5,1.05) (1.5,0)
};
\addplot [
    draw=none,
    fill=gray!30,    
    opacity=0.4
] coordinates {
    (2.5,0) (2.5,1.05) (3.5,1.05) (3.5,0)
};
\addplot [
    draw=none,
    fill=red!30,
    opacity=0.4
] coordinates {
    (1.5,0) (1.5,1.05) (2.5,1.05) (2.5,0)
};
\addplot [
    draw=none,
    fill=red!30,
    opacity=0.4
] coordinates {
    (6.5,0) (6.5,1.05) (8.5,1.05) (8.5,0)
};
\addplot [
    draw=none,
    fill=blue!30,
    opacity=0.4
] coordinates {
    (3.5,0) (3.5,1.05) (6.5,1.05) (6.5,0)
};
\addplot [
    domain=0.5:8.5,
    samples=2,
    thin,
    dashed,
    gray
] {0.5};
\end{axis}
\end{tikzpicture}
\vspace{-.5em}
    \caption{\small Plot of the probability of heads of a coin versus its position in an ordering. 
    Unbiased blocks are in gray, $0$-biased blocks are in blue, and $1$-biased blocks are in red. For example, the fourth coin in the ordering has probability of heads equal to 0.45 and is in a 0-biased block.
    }
    \label{fig:biasedblocksunsorted}
\end{figure}

At a high level, our algorithm works by executing the following three steps:
\begin{enumerate} 
\item Generate an initial ordering of the coins using a greedy rule.
\item Generate $O(n)$ alternative orderings by making small modifications to the initial greedy ordering.
\item Calculate the expected number of flips for each generated ordering and return the best one. 
\end{enumerate}

Our initial ordering (formalized in \Cref{alg:greedy}) is constructed greedily from position $1$ through position $n$. The coin in the $k^\text{th}$ position is chosen from the remaining coins according to the following rule:
\begin{quote}
    \textit{Maximize the probability of terminating on the $k^\text{th}$ flip, assuming that termination has not yet occurred.}
\end{quote}
We show in \Cref{sec:greedy_prelim} that if position $k$ is 1-biased, the greedy rule chooses the remaining coin with smallest $p_i$ value.
Symmetrically, if position $k$ is 0-biased, then the greedy rule chooses the coin with largest $p_i$ value. If the position is unbiased (for example, when $k=1$), \emph{all} coins satisfy the greedy rule.
As a convention, our algorithm chooses the remaining coin with the largest $p_i$ value when position $k$ is unbiased. 

In general, the greedy ordering is not optimal; see \Cref{sub:greedyvsoptimal} for an example.
Nonetheless, we prove that, surprisingly, the optimal ordering can be produced by following the greedy rule \emph{except at one special position}: the last position in the second-to-last block.
This is the \emph{only} position where it may be better to choose a different coin.
With this key fact in place, it is easy to obtain \Cref{thm:mainresult_intro}:
after sorting $p_1, \ldots, p_n$ in $O(n \log n)$ time, our algorithm generates the greedy ordering and all alternatives that result from making a non-greedy choice at that position. There are at most $n-1$ such alternatives.
We  can evaluate the expected cost of all of these orderings in linear time and then choose the one with minimum cost. 

The key challenge in our work is proving that the optimal ordering is only ``non-greedy'' at one particular location. 
To do so, we prove in \Cref{sec:structural,sec:optimal_algo}  that any optimal ordering must obey a set of ``monotonicity'' properties that hold for the greedy ordering. For example, the first such property we prove is that coins within a single $1$-biased (resp. $0$-biased) block have increasing (resp. decreasing) probability of heads. 
Each property is proven using elementary \emph{interchange} arguments, which consider the effect on the expected number of flips if you interchange (swap) a certain pair of coins in an ordering. 

\subsection{Additional results}
Beyond our main result, we prove two additional results on the Unanimous Vote problem. First, we show in \Cref{sub:greedyvsoptimal} that, 
even though the greedy ordering is not optimal, it is always \emph{close} to optimal:
\begin{claim}[Greedy is near-optimal]
    \label{claim:greedy_vs_optimal_intro}
    For any instance of the Unanimous Vote problem, the greedy ordering (\Cref{alg:greedy}) flips at most $1$ more coin in expectation than the optimal ordering. There is a family of instances where this claim is asymptotically tight.
\end{claim}

Finally, in \Cref{sec:adaptivity_gap} we prove a new bound on the adaptivity gap of the Unanimous Vote problem:
\ifarxiv
\begin{theorem}[Bound on adaptivity gap]
\label{thm:adaptivitygap_intro}
    The adaptivity gap of the Unanimous Vote problem is at least $1.2 - o(1)$ and at most $1.2 + o(1)$.
\end{theorem}
The family of instances that give the lower bound is simple to describe and analyze: one coin has $p_i = 0$, and the remaining coins all have $p_i = \frac{1}{2}$. 
The proof of the upper bound is much more involved, analyzing a sequence of specific non-adaptive orderings and showing that at least one ordering always matches the optimal adaptive strategy up to a $1.2 + o(1)$ factor. To illustrate the technique, we first prove a $1.5$ upper bound, and use \Cref{apx:adaptivity_gap_upperbound} to show the $1.2 + o(1)$ upper bound.
\fi

\subsection{Related work}
\label{sec:additional_related}

The Unanimous Vote problem is a Stochastic Boolean Function Evaluation (SBFE) problem.
For a comprehensive survey on SBFE problems, which are also referred to as ``sequential testing'' problems in the Operations Research literature, we refer the reader to the surveys of {\"U}nl{\"u}yurt ~\cite{Unluyurt:2004,Unluyurt:2025}.

In \emph{unit-cost} SBFE problems, the goal is to \emph{exactly evaluate} a given Boolean function $f:\{0,1\}^n \to \{0,1\}$, while observing as few of the inputs bits, $x_1,\ldots, x_n$, as possible.  Each input bit $x_i$ is assumed to be drawn from an independent Bernoulli distribution with given parameter $p_i$.   The problem is to determine the optimal order in which to observe the bits, so as to minimize the expected number of observations.\footnote{In arbitrary-cost SBFE problems, there is a cost $c_i$ associated with observing a bit $x_i$, and the expected total cost of the observations must be minimized.
In what follows, any previous result said to hold for the ``SBFE problem'' rather than for the ``unit-cost SBFE problem'' should be understood to apply to the arbitrary cost version as well.}

SBFE problems are studied in the \emph{adaptive setting}, 
which allows adaptive strategies, and 
in the \emph{non-adaptive setting}, where the observation order must be fixed in advance. One motivation for non-adaptive strategies is that they can always be represented compactly as a permutation and deployed efficiently (we simply read the ID of the next bit to reveal). Adaptive strategies typically require more computation at each step. 
The Unanimous Vote problem is equivalent to the unit-cost, non-adaptive SBFE problem when $f$ is the not-all-equal function, i.e., $f$ evaluates to $1$ iff $x_i \neq x_j$ for some $i,j$. Our work shows that this SBFE problem can be solved in polynomial time.

For some other unit-cost, non-adaptive SBFE problems, designing a polynomial-time algorithm is trivial.
For example, for the Boolean OR function, the problem is equivalent to the coin-flipping problem mentioned at the start of the introduction, where you flip coins until you see the first occurrence of heads. As another example, consider the parity function: evaluating this function requires observing all bits, so the expected number of observations is the same no matter which ordering you use.  
Our work appears to be the first to give a polynomial-time algorithm solving a  unit-cost, non-adaptive SBFE problem that is \emph{not} trivial.

There is, however, interesting prior work on {\em approximately} solving non-adaptive SBFE problems.  For example, constant-factor approximation algorithms have been studied for a problem called ``Stochastic Score Classification’'~\cite{GhugeGuptaNagarajan:2024,plank2024simple,Liu:2022,grammel2022algorithms}. These methods yield constant factor approximation algorithms for the non-adaptive SBFE problems for any symmetric Boolean function and for any linear threshold function. Recall that a Boolean function is symmetric if its output depends only on the number of 1's in the input, so these results cover the Unanimous Vote problem (the not-all-equal function is symmetric). Before our work, the best previous polynomial time algorithm for the Unanimous Vote problem achieved an approximation factor of $\phi \approx 1.618$ using a method specialized to that problem \cite{gkenosis2018stochastic}.

Specialized approximation algorithms have also been developed for 
the \emph{$k$-of-$n$ function}, the symmetric Boolean function whose output is 1 iff at least $k$ of its inputs are 1.  
Recent work gives a PTAS for the unit-cost, non-adaptive SBFE problem for the $k$-of-$n$ function~\cite{nielsen2025nonadaptive}.
There is also a very simple 1.5 approximation algorithm for this problem~\cite{grammel2022algorithms}.  
Beyond symmetric functions, an 8-approximation algorithm is known for the unit-cost, non-adaptive SBFE problem for Boolean read-once formulas~\cite{HappachHellersteinetal:2022}. 

\paragraph{Adaptive Methods.}
In the \emph{adaptive} setting, the Unanimous Vote problem can be easily solved by noting the following: once you have flipped the first coin and observed its outcome, it will be optimal to flip the remaining coins in either increasing or decreasing probability of heads, depending on whether your first flip is heads or tails~\cite{gkenosis2018stochastic}.  Thus the only difficulty is to determine the first coin. To do so, we can just compute the expected cost associated with each of the $n$ choices of the first coin, and choose the best one.  

Other adaptive SBFE problems are more interesting. There is an elegant polynomial time algorithm that solves the 
adaptive SBFE problem for the $k$-of-$n$ function~\cite{Ben-Dov:1981,changetal:1990,SalloumBreuer:1997}, and also works for the exactly-$k$ function, whose output is 1 iff exactly $k$ of the inputs are $1$~\cite{GkenosisGrammelHellerstein:2022,AcharyaJafarpourOrlitsky:2011}. 
The adaptive SBFE problem has also been studied for functions represented by read-once CNF (dually, DNF) formulas~\cite{Boros2000} and for linear threshold functions and symmetric functions~\cite{DeshpandeHellersteinKletenik:2016,GhugeGuptaNagarajan:2024,AcharyaJafarpourOrlitsky:2011,DasJafarpourOrlitsky:2012,KowshikKumar:2013}. 

There has also been interest in the gap between the best adaptive and non-adaptive solutions to SBFE problems.
The unit-cost SBFE problem for arbitrary symmetric functions has an adaptivity gap that is between 1.5 and 2~\cite{grammel2022algorithms,plank2024simple}. 
The unit-cost SBFE problem for the $k$-of-$n$ function has an adaptivity gap of  1.5~\cite{grammel2022algorithms,plank2024simple}. 
For some other Boolean functions, the problem has non-constant gaps~\cite{HellersteinKleteniketal:2022}.

\paragraph{Hardness.}
Few $\NP$-hardness results are known for unit-cost SBFE problems.  
The unit-cost SBFE problem for functions represented by arbitrary CNF formulas is trivially $\NP$-hard, in both the adaptive and non-adaptive settings: if the formula is not satisfiable, the optimal evaluation strategy does not observe any bits at all.  By a simple reduction from vertex cover, $\NP$-hardness also holds even if the CNF formula has no negations and is restricted to have 2 literals per clause~\cite{AllenHellersteinetal:2017}. Analogous results hold for DNF formulas by duality.  
We note that $\NP$-hardness of the unit-cost SBFE problem for linear-threshold functions is still an open question.  
$\NP$-hardness was shown for arbitrary costs in~\cite{heuristicLeastCostCox}, but contrary to what was stated in~\cite{DeshpandeHellersteinKletenik:2016},
that result did not apply to the unit-cost case.


\paragraph{Related Problems.}
Finally, we note that there is a large body of work on Boolean function evaluation problems in both adversarial and probabilistic models with other assumptions or goals; see for example \cite{KrishnamurthyBoralZaniolo:1986,IbarakiKameda:1984,charikarPricedInfo00,kaplanMansour-Stoc05,cicalese2011competitive,blanc2021query,Unluyurt:2025,HarrisNagarajanTan:2025}.
There has also been interest in functions over larger alphabets, like voting problems with more than two choices~\cite{HellersteinLiuSchewior:2024,BradacSinglaZuzic:2019}. 
More generally, stochastic probing problems (where each input is $1$ with some known probability) are studied beyond function evaluation. For example, there is interest in solving encoded optimization problems, where input bits can be queried sequentially; see e.g. \cite{GuptaNagarajan:2013,GuptaNagarajanSingla:2017,Singla:2018,PattonRussoSingla:2023,SegevSingla:2021}.

\section{Preliminaries}
\label{sec:preliminaries}
In this section we review notation and terminology used throughout the paper. We also formalize a greedy algorithm for the Unanimous Vote problem that plays a central role in our main result, \Cref{thm:mainresult_intro}.
\subsection{Notation and terminology}
\label{sec:notation}

Throughout, we use the term ``increasing'' (resp. ``decreasing'') as a synonym for ``non-decreasing'' (resp. ``non-increasing''). 
We use ``ordering'' to refer to a permutation $a$ on $n$ elements, with order $a(1), a(2), \dots, a(n)$.
  The outcome ``heads'' is synonymous with 1, and ``tails'' with 0. The input to our problem is a set of $n$ coins with biases $p_1, \ldots, p_n \in [0,1]$. Without loss of generality, we assume $p_1 \le \dots \le p_n$. The first step in all of our algorithms is to sort the probabilities in $O(n \log n)$ time.  We denote $\pbar_i \defeq 1 - p_i$. 

The \emph{cost} of an ordering is the expected number of flipped coins to determine whether or not a vote is unanimous. The cost can be written as the expectation of a sum of $\{0,1\}$ indicator random variables $\{X_j\}_{j \ge 1}$, where $X_j=1$ if we reach the $j^{\text{th}}$ coin in the ordering without terminating (i.e., we flip the $j^\text{th}$ coin). We never terminate before the second coin.
Moreover, for all $j\geq 3$, we only fail to terminate if the coins prior to the $j^{\text{th}}$ coin all come up heads or all come up tails. Thus,
    \begin{align*}
        cost(a) = \sum_{j \ge 1} \E[X_j] 
        = 2 + \sum_{j \ge 3}\left( \prod_{i=1}^{j-1} p_{a(i)} + \prod_{i=1}^{j-1} \pbar_{a(i)}\right) \,.
    \end{align*}
Since these quantities will be important later, we introduce the notation $z_j^1(a) \defeq \prod_{i=1}^{j-1} p_{a(i)}$ to be the chance that the first $j-1$ coins are \emph{heads} in ordering $a$, and  $z_j^0(a) \defeq \prod_{i=1}^{j-1} \pbar_{a(i)}$ to be the chance that the first $j-1$ coins are \emph{tails} in ordering $a$. As above, for $j \geq 3$, $\E[X_j] = z_j^1(a) + z_j^0(a)$.
It follows that
    \begin{align}
    \label{eq:cost_express}
        cost(a) 
        = 2 + \sum_{j \ge 3} z_j^0(a) + z_j^1(a) = 1 + \sum_{j \ge 2} z_j^0(a) + z_j^1(a)\,. 
    \end{align}
We call an ordering, $a$, \emph{$0$-biased} at position $j$ if $z_j^0(a) > z_j^1(a)$; i.e., it is more likely that the first $j-1$ coins flipped are all tails (0) than all heads (1). If $z_j^0(a) < z_j^1(a)$, we call the ordering \emph{$1$-biased} at position $j$. If $z_j^0(a) = z_j^1(a)$, we call the ordering \emph{unbiased} at position $j$.

Using this language, we can uniquely partition the sequence $[a(1), \dots, a(n)]$ into contiguous \emph{blocks}
$[B_1, \dots, B_q]$ so that for each $j$, $a(j)$ and $a(j+1)$ belong to the same block if and only if $a$ has the same bias type (0-biased, 1-biased, or unbiased) at positions $j$ and $j+1$.
See \Cref{fig:biasedblocksunsorted} for an example partitioning of an ordering into its \emph{biased blocks}.

\subsection{The greedy algorithm}
\label{sec:greedy_prelim}
A natural greedy algorithm for the Unanimous Vote Problem is to always choose the coin that, if flipped, has the highest probability of terminating the algorithm. This probability can be expressed as follows:
\begin{fact}[Probability of terminating the sequence]
\label{fact:prob_termination}
    Fix an ordering $a$ and position $x > 1$. Assuming we flip $x-1$ coins without terminating, the probability of terminating after flipping the coin $a(x)$ is 
    \begin{align*}
       \frac{\pbar_{a(x)} z_x^1(a) + p_{a(x)} z_x^0(a)}{z_{x}^1(a) + z_{x}^0(a)} = 1 -  \frac{z_{x+1}^1(a) + z_{x+1}^0(a)}{z_{x}^1(a) + z_{x}^0(a)}\,.
    \end{align*}
\end{fact}
The value of $p_{a(x)}$ that maximizes the expression in  \Cref{fact:prob_termination} depends on the \emph{bias type} of $a$ at position $x$. If it is $0$-biased, the expression is increasing with $p_{a(x)}$; if it is $1$-biased, the expression is \emph{decreasing} with $p_{a(x)}$; if $a$ is unbiased at position $x$, then the expression always equals $\frac{1}{2}$.
So, in a $0$-biased block, the choice that maximizes the probability of termination is to select the remaining coin with the \emph{largest} heads probability. In a $1$-biased block, it is to select the coin with the \emph{smallest} heads probability.  For unbiased blocks, selecting any coin terminates the sequence with probability exactly $\frac{1}{2}$, so the greedy choice is not unique. By convention, we always choose the remaining coin with \emph{largest} heads probability, although our analysis would also works if we chose the coin with smallest heads probability. We may now construct a ``greedy'' algorithm:

\begin{tabular}{|p{6.5in}}{\underline{\textsc{Greedy Algorithm (\Cref{alg:greedy})}}}
\alglabel{alg:greedy}

Assume coins are in increasing order; i.e. $p_1 \le \dots \le p_{n}$.

Start with an empty ordering. Repeat the following rule until all coins are chosen:

\quad If the ordering is $1$-biased, choose the remaining coin with smallest probability of heads.

\quad Otherwise, choose the remaining coin with largest probability of heads.
\end{tabular}
\vspace{.25em}

We assume the algorithm can break ties arbitrarily (i.e. if two coins have the same probability of heads). See \Cref{fig:biasedblockssorted} for an example output of the algorithm.

\begin{figure}[ht]
\vspace{.5em}
    \centering
\begin{center}
\begin{tikzpicture}[scale=.7]
\begin{axis}[
    scaled y ticks=false,
    width=13cm,
    height=8cm,
    ymin=0, ymax=1.05,
    ytick={0, 0.05, 0.10, 0.15, 0.20, 0.25, 0.30, 0.35, 0.40, 0.45, 0.50, 0.55, 0.60, 0.65, 0.70, 0.75, 0.80, 0.85, 0.90, 0.95, 1.0},
    yticklabels={ , ,0.1 , ,0.2 , ,0.3 , , 0.4, , 0.5, ,0.6 , ,0.7 , , 0.8, ,0.9 , ,1.0 },
    xtick={1,2,3,4,5,6,7,8},
    xticklabels={$a(1)$, $a(2)$, $a(3)$, $a(4)$, $a(5)$, $a(6)$, $a(7)$, $a(8)$},
    enlargelimits=0.1,
    axis lines=left,
    ylabel={probability of heads},
    xlabel={position of coin in ordering},
    scatter/classes={
        a={mark=o,draw=Gray,fill=Gray,mark size=2.7pt,line width=5.4pt},
        c={mark=o,draw=MidnightBlue,fill=MidnightBlue,mark size=2.7pt,line width=5.4pt},
        b={mark=o,draw=BrickRed,fill=BrickRed,mark size=2.7pt,line width=5.4pt},
        d={mark=none},
        e={mark=o,draw=black,fill=black,mark size=2.7pt,line width=5.4pt}
    }
]
\addplot[
    only marks,
    scatter,
    scatter src=explicit symbolic
]
coordinates {
    (0.5,0) [d]
    (1,0.10) [e]
    (2,0.85) [e]
    (3,0.80) [e]
    (4,0.30) [e]
    (5,0.35) [e]
    (6,0.60) [e]
    (7,0.45) [e]
    (8,0.40) [e]
    (8.5,0) [d]
};
\addplot [
    draw=none,
    fill=red!30,  
    opacity=0.4
] coordinates {
    (3.5,0) (3.5,1.05) (5.5,1.05) (5.5,0)
};
\addplot [
    draw=none,
    fill=gray!30,
    opacity=0.4
] coordinates {
    (0.5,0) (0.5,1.05) (1.5,1.05) (1.5,0)
};
\addplot [
    draw=none,
    fill=blue!30,
    opacity=0.4
] coordinates {
    (5.5,0) (5.5,1.05) (8.5,1.05) (8.5,0)
};
\addplot [
    draw=none,
    fill=blue!30,
    opacity=0.4
] coordinates {
    (1.5,0) (1.5,1.05) (3.5,1.05) (3.5,0)
};
\addplot [
    domain=0.5:8.5,
    samples=2,
    thin,
    dashed,
    gray
] {0.5};
\end{axis}
\end{tikzpicture}
\end{center}
\vspace{-1em}
    \caption{\small Plot of the probability of heads of a coin versus its position in an example greedy ordering. This is the same instance as in \Cref{fig:biasedblocksunsorted}. Unbiased blocks are in gray, $0$-biased blocks are in blue, and $1$-biased blocks are in red.}
    \label{fig:biasedblockssorted}
\end{figure}
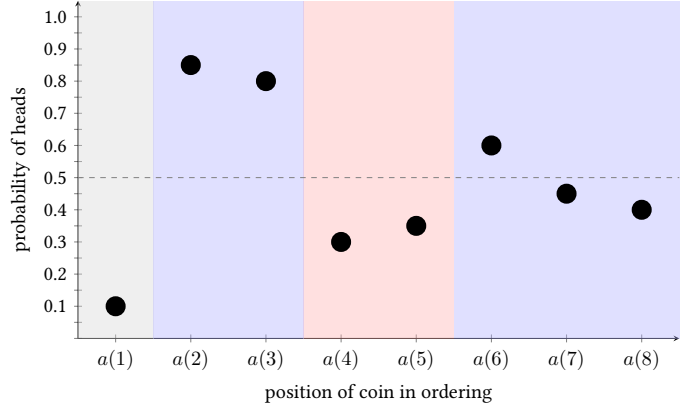

\section{Structural results}
\label{sec:structural}
In this section, we show that any optimal ordering  for the Unanimous Vote problem must share several natural characteristics of the \emph{greedy} ordering (\Cref{alg:greedy}).

\subsection{Monotonicity within a block}
Consider the coins in a single $0$-biased biased block of the greedy solution. Since, in a $0$-biased block,  \Cref{alg:greedy} chooses the next coin with largest probability of heads, the coins in this block will have \emph{decreasing} probability of heads. Similarly, the coins in any $1$-biased block of the greedy solution will be sorted so that the probability of heads is \emph{increasing} in the block. Our first structural result is that the same fact holds for \emph{any optimal ordering}.

The proof proceeds via an intuitive ``swap'' argument: if we have not terminated by the time we reach a $0$-biased block, this is most likely because we have observed all tails ($0$s) up until that point. This remains true even if we swap the positions of coins in the block. Accordingly, we can increase our probability of termination by reordering the coins to flip a heads as quickly as possible, i.e., by moving coins with larger probability of heads to the front of the block. 
Formally, we require the following expression for how the cost of an ordering changes if adjacent coins are swapped:

\begin{claim}[Cost of swapping adjacent positions]
\label{claim:swap_order_by_one}
    Choose an ordering $a$ and position $x \in [1,n-1]$. Let $b$ be the ordering starting from $a$ but swapping position $x$ and $x+1$. Then
    \begin{align*}
        cost(b) - cost(a) =  \left( z_x^1(a) - z_x^0(a) \right) \left( p_{a(x+1)} - p_{a(x)} \right)\,.
    \end{align*}
\end{claim}

\begin{proof}
    The orderings $a$ and $b$ are identical except at positions $x$ and $x+1$. So we have that $z_j^0(a) = z_j^0(b)$ and $z_j^1(a) = z_j^1(b)$ for all $j\in [1,\dots,n]\setminus\{x+1\}$. Using the cost expression \eqref{eq:cost_express}, we then have:
    \begin{align*}
        cost(b) - cost(a) 
        &= 
        \left(
        1 + \sum_{j \ge 2} \left( z_j^1(b) + z_j^0(b) \right)
        \right)
        -
        \left(
        1 + \sum_{j \ge 2} \left( z_j^1(a) + z_j^0(a) \right)
        \right)
        \\
        &=
        z_{x+1}^1(b) - z_{x+1}^1(a) + z_{x+1}^0(b) - z_{x+1}^0(a)
        \\
        &= z_x^1(a) \left( p_{a(x+1)} - p_{a(x)} \right) + z_x^0(a) \left( \pbar_{a(x+1)} - \pbar_{a(x)} \right)
        \\
        &= \left( z_x^1(a) - z_x^0(a) \right) \left( p_{a(x+1)} - p_{a(x)} \right)\,. \qedhere
    \end{align*}
\end{proof}

\begin{corollary}[Optimal ordering is sorted within each block]
\label{cor:sorted_within_block}
    Consider any optimal ordering $a$ with block sequence $[B_1, \dots, B_q]$. Within each $0$-biased block, the coins have \emph{decreasing} probability of heads. Within each $1$-biased block, they have \emph{increasing} probability of heads.
\end{corollary}
\begin{proof}
Since $a$ is optimal, we have $cost(b) - cost(a) \ge 0$ for all orderings $b$.
Suppose $x$ is a position within a  $0$-biased block, i.e. $z_x^0(a) > z_x^1(a)$. By \Cref{claim:swap_order_by_one}, we must have $p_{a(x)} \ge p_{a(x+1)}$; i.e. the next coin must have an equal or smaller probability of heads. The argument is symmetric for $1$-biased blocks. 
\end{proof}

\subsection{Monotonicity across non-final blocks}

The greedy solution obeys a stronger property than  monotonicity within each block. \emph{Every} coin in a $1$-biased block has a smaller (or equal) probability of heads than all future coins. Similarly, \emph{every} coin in a $0$-biased or unbiased block (because of our convention) has a larger or equal probability of heads than all future coins.
We show there is an optimal solution that also obeys this ``across block'' monotonicity, \emph{at least among all coins in non-final blocks}.
To prove this statement, we first require the following basic fact:

\begin{claim}[Optimal ordering uses coins on one side of $\frac{1}{2}$ in non-final blocks]
\label{claim:nonfinal_probatleasthalf}
        Consider any optimal ordering $a$ with block sequence $[B_1, \dots, B_q]$. For all $1 \leq k < q$, if $B_k$ is $0$-biased, then it only contains coins with probability of heads $> \frac{1}{2}$. If $B_k$ is  $1$-biased, then it only contains coins with heads probability of  $< \frac{1}{2}$.
\end{claim}
\begin{proof}
    Fix a block $B_k = [a(i), \dots, a(j)]$ with $k \ne q$. If $B_k$ is $0$-biased, then in the last position $z_j^0(a) > z_j^1(a)$ but $z_{j+1}^0(a) \le z_{j+1}^1(a)$. Since $z_{j+1}^1(a) = p_{a(j)} z_j^1(a)$ and $z_{j+1}^0(a) = \pbar_{a(j)} z_j^0(a)$, this implies $p_{a(j)} > \pbar_{a(j)}$; i.e. $p_{a(j)} > \frac{1}{2}$. By \Cref{cor:sorted_within_block}, all other coins in block $B_k$ have probability of heads at least $p_{a(j)}$, which is greater than $\frac{1}{2}$. (The argument is symmetric for $1$-biased blocks.)
\end{proof}

\Cref{claim:nonfinal_probatleasthalf} only considers $0$-biased and $1$-biased blocks. However, it is possible that the optimal ordering has unbiased blocks. This makes the optimal ordering non-unique: by \Cref{claim:swap_order_by_one}, swapping a coin in an unbiased block with the following coin in the ordering does not change the cost. 
Nonetheless, we show there is an optimal ordering where all coins in non-final blocks have probability \emph{not} equal to $\frac{1}{2}$: 

\begin{claim}
    \label{claim:unbiasedblockform}
    There is an optimal ordering $a$ with block sequence $[B_1, \dots, B_q]$ where every unbiased block that is not $B_q$ has just one coin, and that coin has probability of heads \emph{not equal} to $\frac{1}{2}$. Moreover, if $B_q$ is unbiased and contains more than one coin, all coins in the block have probability of heads equal to $\frac{1}{2}$.
\end{claim}
\begin{proof}
We begin with an optimal ordering $a$ with block sequence $[B_1, \dots, B_q]$, and argue that we can swap adjacent coins, without changing the cost of the ordering, in such a way that the properties above are eventually satisfied.
Note that the last coin of any non-final unbiased block \emph{must} have probability of heads not equal to $\frac{1}{2}$, since $a$ is unbiased at $j$ and \emph{not} unbiased at position $j+1$.

Let $k$ be the first block $B_k = [a(i), \dots, a(j)]$ in $a$ such that $B_k$ is unbiased,  has more than one coin, and $k \ne q$. For any $i \le x < j$, $a$ is unbiased at positions $x$ and $x+1$; i.e., both $z_x^1(a) = z_x^0(a)$ and $p_{a(x)} z_x^1(a) = z_{x+1}^1(a) = z_{x+1}^0(a) = \pbar_{a(x)} z_x^0(a)$.
    So the probability $p_{a(x)}$ \emph{must} be equal to $\frac{1}{2}$.

By \Cref{claim:swap_order_by_one}, for any $i \le x \le j$, we can swap position $x$ and $x+1$ without changing the cost. Suppose we do this with position $x = j-1$. Since the new coin at position $j-1$ has probability of heads \emph{not} equal to $\frac{1}{2}$, the coin \emph{initially} at position $j-1$ moves to the block $B_{k+1}$. We may now swap position $j-2$ and $j-1$, moving the coin \emph{initially} at position $j-2$ to the block $B_{k+1}$. We continue this process until we swap position $i$ and $i+1$. At this point, the block $B_k$ has one coin in it. Moreover, since $B_k$ is unbiased, $B_{k+1}$ must be $0$-biased or $1$-biased. Since the coins previously in block $B_k$ have probability of heads equal to $\frac{1}{2}$, by \Cref{claim:nonfinal_probatleasthalf}, $B_{k+1}$ must be the final block. So, by making these swaps, we satisfy the first part of \Cref{claim:unbiasedblockform}.

Suppose at this point, the final block $B_q = [a(x), \dots, a(n)]$ is unbiased and has more than one coin. By the same argument as before, every coin except possibly the last coin in $B_q$ has probability of heads equal to $\frac{1}{2}$. If $p_{a(n)} = \frac{1}{2}$, we are done. Otherwise, we may repeat the process from before until $B_q$ has one coin in it. In this case, a new block $B_{q+1}$ is created which is either $0$-biased or $1$-biased. 
\end{proof}

By \Cref{claim:nonfinal_probatleasthalf} and \Cref{claim:unbiasedblockform}, there is an optimal ordering where the coins in non-final blocks can be partitioned into $S_+$ and $S_-$, such that all coins in $S_+$ have probability  of heads greater than $\frac{1}{2}$, and all coins in $S_-$ have probability of heads \emph{less} than $\frac{1}{2}$.  We show that all coins in $S_+$ are sorted in decreasing probability of heads, and all coins in $S_-$ are sorted in \emph{increasing} probability of heads. It is in this sense we show that there is an optimal ordering with monotonicity across non-final blocks.
To show this, we give an expression for the change in cost of an ordering when coins are swapped across blocks:
\begin{claim}[Cost of swapping across a block]
\label{claim:multiblock_monotonic}
    Choose an ordering $a$.
    Fix positions $s,t \in [n]$ where $s < t$.
    Let $b$ be the ordering obtained from swapping positions $s$ and $t$ in $a$. Suppose $cost(a) \le cost(b)$. Then:
    \begin{itemize}
        \item If $a$ is $0$-biased at $s$ and $p_{a(x)} \le \frac{1}{2}$ for all $s < x < t$, then $p_{a(s)} \ge p_{a(t)}$.
        \item If $a$ is $1$-biased at $s$ and $p_{a(x)} \ge \frac{1}{2}$ for all $s < x < t$, then $p_{a(s)} \le p_{a(t)}$.
        \item If $a$ is unbiased at $s$ and $p_{a(x)} \le \frac{1}{2}$ for all $s < x < t$ and $p_{a(y)} < \frac{1}{2}$ for some $s < y < t$,   $p_{a(s)} \ge p_{a(t)}$.
        \item If $a$ is unbiased at $s$ and $p_{a(x)} \ge \frac{1}{2}$ for all $s < x < t$ and $p_{a(y)} > \frac{1}{2}$ for some $s < y < t$,  $p_{a(s)} \le p_{a(t)}$.
    \end{itemize}
\end{claim}
\begin{proof}
    By construction, $z_j^0(a) = z_j^0(b)$ and  $z_j^1(a) = z_j^1(b)$ for all $j \leq s$ and $j \geq t+1$. The cost difference is
    \begin{align*}
         0 &\le cost(b) - cost(a) 
        = 
        \left(
        1 + \sum_{j \ge 2} \left( z_j^1(b) + z_j^0(b) \right)
        \right)
        -
        \left(
        1 + \sum_{j \ge 2} \left( z_j^1(a) + z_j^0(a) \right)
        \right)
        \\
        &=
        \sum_{j = s+1}^{t}
        z_{j}^1(b) - z_{j}^1(a) + z_{j}^0(b) - z_{j}^0(a)
        \\
        &= 
        \sum_{j = s+1}^t
        \left(
        z_s^1(a)
        \cdot 
        (
        p_{a(t)} - p_{a(s)}
        )
       \left( \prod_{x = s+1}^{j-1}
        p_{a(x)}\right)
        +
           z_s^0(a)
        \cdot 
        (
        \pbar_{a(t)} - \pbar_{a(s)}
        )
        \left(\prod_{x = s+1}^{j-1}
        \pbar_{a(x)} \right) 
        \right)
        \\
        &= (p_{a(t)} - p_{a(s)})
        \left(
        z_s^1(a) \cdot 
         \left(\sum_{j=s+1}^{t} \prod_{x = s+1}^{j-1} p_{a(x)} \right) 
        -
        z_s^0(a) \cdot \left(\sum_{j=s+1}^{t} \prod_{x = s+1}^{j-1} \pbar_{a(x)} \right) 
        \right)\,.
    \end{align*}
Suppose $p_{a(x)} \le \frac{1}{2}$ for all $s < x < t$. Then for each $x$, $p_{a(x)} \le \pbar_{a(x)}$, and so the first summation of products is at most the second summation of products. We can conclude then:
\begin{itemize}
\item If $a$ is $0$-biased at $s$, then $z_s^0(a) > z_s^1(a)$, and so the second term is strictly less than $0$. Since the difference in cost must be non-negative, we have $p_{a(s)} \ge p_{a(t)}$. 
\item If instead $a$ is unbiased at $s$, then  $z_s^1(a) = z_s^0(a)$, and so the second term is at most $0$. 
Since the difference in cost must be non-negative, it is possible that $p_{a(s)} < p_{a(t)}$ only if the second term is equal to $0$. This can only occur if $p_{a(x)} = \frac{1}{2}$ for all $s < x < t$. 
\end{itemize}
The parts of the claim when $p_{a(x)} \ge \frac{1}{2}$ for all $s < x < t$ follow from a symmetric argument.
\end{proof}

\begin{corollary}[Optimal ordering is sorted across blocks]
\label{cor:sorted_across_block}
There exists an optimal ordering $a$ with block sequence $[B_1, \dots, B_q]$ where the following holds:
\begin{itemize}
    \item The coins in blocks $[B_1, \dots, B_{q-1}]$ can be partitioned into a set $S_+$ and a set $S_-$, where all coins in $S_+$ have probability of heads more than $\frac{1}{2}$, and all coins in $S_-$ have probability of heads less than $\frac{1}{2}$.
    \item The ordering $a$ restricted to the coins in $S_+$ are sorted in $a$ by decreasing probability of heads. 
    \item The ordering $a$ restricted to the coins in $S_-$ are sorted in $a$ by \emph{increasing} probability of heads.
\end{itemize}
\end{corollary}
\begin{proof}
Take the ordering guaranteed by \Cref{claim:unbiasedblockform}.
    By \Cref{claim:nonfinal_probatleasthalf}  we can partition the coins in blocks $[B_1, \dots, B_{q-1}]$ into $S_+$ and $S_-$. Note that no coins in $S_+$ or $S_-$ have probability of heads equal to $\frac{1}{2}$. The monotonicity of coins in $S_+$ and $S_-$ then follows by \Cref{claim:multiblock_monotonic}.
\end{proof}

\section{Constructing an exact algorithm}
\label{sec:optimal_algo}
In this section, we show how to generate an optimal ordering from small modifications to the greedy ordering. This relies on the structural results from \Cref{sec:structural}.

\subsection{Monotonicity almost everywhere}

In the greedy ordering, ``monotonicity across blocks'' holds \emph{even for the final block}. A coin in a $0$-biased block has probability of heads at least that of every following coin. Similarly, a coin in a $1$-biased block has probability of heads at most that of every following coin. Surprisingly, it turns out that the optimal ordering guaranteed by \Cref{claim:unbiasedblockform} \emph{also} has this property, except possibly at one location.

\begin{claim}[There exists an optimal ordering that is almost greedy]
\label{claim:optimalordering_onespotfromgreedy}
    There is an optimal ordering $a$ with block sequence $[B_1, \dots, B_q]$ where the following statement holds at each position $x$ except for the position of the last coin of $B_{q-1}$:
    \begin{itemize}
        \item Suppose $a$ is $0$-biased at position $x$. Then for all positions $x' > x$, $p_{a(x)} \ge p_{a(x')}$ .
        \item Suppose $a$ is $1$-biased at position $x$. Then for all positions $x' > x$, $p_{a(x)} \le p_{a(x')}$.
        \item Suppose $a$ is unbiased at position $x$ and $p_{a(x)} > \frac{1}{2}$. Then for all positions $x' > x$, $p_{a(x)} \ge p_{a(x')}$. 
        \item Suppose $a$ is unbiased at position $x$ and $p_{a(x)} < \frac{1}{2}$. Then for all positions $x' > x$, $p_{a(x)} \le p_{a(x')}$.
    \end{itemize}
Furthermore, if the last block $B_q$ is unbiased, then the above statement holds at all positions.
\end{claim}
\begin{proof}
Consider an optimal ordering $a$ guaranteed by
\Cref{cor:sorted_across_block}. Suppose it has block sequence  $[B_1, \dots, B_q]$, and the coins in the non-final blocks are partitioned into $S_+$ and $S_-$. $S_+$ and $S_-$ are already sorted by probability of heads, and each coin in $S_+$ has probability of heads larger than every coin in $S_-$. So if a violation of the statement occurs, the violating position $x'$ must be in the final block.

Suppose the final block $B_q$ is unbiased. In this case we show that the statement holds at all positions:
\begin{itemize}
    \item Suppose $B_q$ contains more than one coin. Then all coins in $B_q$ have probability of heads $\frac{1}{2}$. The statement is vacuous for any position $x \in B_q$. Moreover, $\frac{1}{2}$ is smaller than the probability of heads of all coins in $S_+$ and larger than the probability of heads of all coins in $S_-$, so there are no violations for any $x \notin B_q$ and $x' \in B_q$. So the statement holds at all positions.
    \item Suppose $B_q$ contains just one coin ($x'$, at position $n$). Suppose the coin at position $n-1$ is in $S_+$. (A symmetric argument holds if it is in $S_-$). Since $B_q$ is unbiased, this coin must belong to a $0$-biased block; by \Cref{claim:multiblock_monotonic}, it has probability of heads at least $p_{a(n)}$.
    By \Cref{cor:sorted_across_block}, all coins in $S_+$ have probability of heads at least $p_{a(n)}$.
    Now consider the last coin in $S_-$, which is in position at most $n-2$. Note that the coin at position $n-1$ has probability of heads at least $\frac{1}{2}$. So then by \Cref{claim:multiblock_monotonic}, the last coin in $S_-$ has probability of heads at most $p_{a(n)}$. By \Cref{cor:sorted_across_block}, all coins in $S_-$ have probability of heads at most $p_{a(n)}$. So, the statement holds for any $x < n = x'$.
\end{itemize} 

Next, we consider when the final block $B_q$ is $1$-biased. (A symmetric argument holds when $B_q$ is $0$-biased.) First of all, by \Cref{cor:sorted_within_block}, the coins in $B_q$ are sorted in increasing probability of heads. Moreover, by \Cref{claim:multiblock_monotonic}, all coins in $S_-$ have probability of heads at most that of any coin in $B_q$. So if there is a violation of the statement, it occurs between a position $x$ in $S_+$ and $x'$ in $B_q$.

Suppose for the sake of contradiction there are \emph{two} positions $k,\ell$ in $S_+$ where the statement does not hold. 
Let $k < \ell$; then $p_{a(k)} \ge p_{a(\ell)}$.
So, there exists a position $x'$ in $B_q$ where $p_{a(x')} > p_{a(k)} \ge p_{a(\ell)}$. 
Since the coins in $S_+$ are sorted in \emph{decreasing} probability of heads, we can take $k,\ell$ to be the last two coins in $S_+$; so $p_{a(x)} < \frac{1}{2}$ for all $k < x < \ell$.
We choose $x'$ to be the earliest position where the violation with $k$ and $\ell$ occurs.

We show that swapping positions $k$ and $x'$ decreases the cost, which contradicts our claim of an optimal ordering. Let $b$ be the ordering starting from $a$ but swapping $k$ and $x'$. Then (as in the proof of \Cref{claim:multiblock_monotonic}),
\begin{align*}
    cost(b) - cost(a) &= 
    (p_{a(x')} - p_{a(k)})
        \left(
        z_k^1(a) \cdot 
         \left(\sum_{j=k+1}^{x'} \prod_{x = k+1}^{j-1} p_{a(x)} \right) 
        -
        z_k^0(a) \cdot \left(\sum_{j=k+1}^{x'} \prod_{x = k+1}^{j-1} \pbar_{a(x)} \right) 
        \right)\,.
\end{align*}
Since $p_{a(x')} > p_{a(k)}$, the sign of the difference in cost depends on the sign of the term in parentheses. We split this term into two parts: the first with summation terms where $j \in [k+1,\dots, \ell-1]$, and the second with summation terms where $j \in [\ell,\dots, x']$. Recall that at positions $k,\ell$, the ordering $a$ is \emph{not} $1$-biased; i.e.  $z^1_k(a) \le z^0_k(a)$ and $z^1_\ell(a) \le z^0_\ell(a)$.
\begin{itemize}
    \item The first part of the expression involves terms $p_{a(x)}$ for $k+1 \le x \le \ell - 2$. For all these terms, $p_{a(x)} < \frac{1}{2}$; so
\begin{align*}
        z_k^1(a) \cdot 
         \left(\sum_{j=k+1}^{\ell-1} \prod_{x = k+1}^{j-1} p_{a(x)} \right) 
        -
        z_k^0(a) \cdot \left(\sum_{j=k+1}^{\ell-1} \prod_{x = k+1}^{j-1} \pbar_{a(x)} \right) 
        \le 0\,.
    \end{align*}
    \item The second part of the expression also involves terms $p_{a(x)}$ for $\ell - 1 \le x \le x' - 1$. Since position $x'$ has the first coin where $x' > \ell$ such that $p_{a(x')} \ge p_{a(k)}$, we have $p_{a(x)} \le p_{a(k)}$ for all these terms.  We upper-bound the second part, setting $p \defeq p_{a(k)}$ for convenience:
 \begin{align*}
         &z_k^1(a) \cdot 
         \left(\sum_{j=\ell}^{x'} \prod_{x = k+1}^{j-1} p_{a(x)} \right) 
        -
        z_k^0(a) \cdot \left(\sum_{j=\ell}^{x'} \prod_{x = k+1}^{j-1} \pbar_{a(x)} \right) 
        \\
        &=
        z_k^1(a) \cdot \left( \prod_{x=k+1}^{\ell-1} p_{a(x)} \right) 
        \left( 1 + \sum_{j=\ell+1}^{x'} \prod_{x=\ell}^{j-1} p_{a(x)} \right)
        -
        z_k^0(a) \cdot \left( \prod_{x=k+1}^{\ell-1} \pbar_{a(x)} \right) 
        \left( 1 + \sum_{j=\ell+1}^{x'} \prod_{x=\ell}^{j-1} \pbar_{a(x)} \right)
        \\
          &\le 
        z_k^1(a) \cdot \left( \prod_{x=k+1}^{\ell-1} p_{a(x)} \right) 
        \left( 1 + \sum_{j=\ell+1}^{x'} \prod_{x=\ell}^{j-1} p \right)
        -
        z_k^0(a) \cdot \left( \prod_{x=k+1}^{\ell-1} \pbar_{a(x)} \right) 
        \left( 1 + \sum_{j=\ell+1}^{x'} \prod_{x=\ell}^{j-1} \pbar \right)
        \\
         &= 
        z_k^1(a) \cdot \left( \prod_{x=k+1}^{\ell-1} p_{a(x)} \right) 
        \cdot \frac{1-p^{x'-\ell+1}}{\pbar}
        -
        z_k^0(a) \cdot \left( \prod_{x=k+1}^{\ell-1} \pbar_{a(x)} \right) 
        \cdot \frac{1-\pbar^{x'-\ell+1}}{p}
        \\
        &= 
        z_{k}^1(a) \cdot \left( \prod_{x=k}^{\ell-1} p_{a(x)} \right) \cdot \frac{1-p^{x'-\ell+1}}{p \cdot \pbar}
        -
        z_{k}^0(a) \cdot \left( \prod_{x=k}^{\ell-1} \pbar_{a(x)} \right) \cdot \frac{1-\pbar^{x'-\ell+1}}{p \cdot \pbar}        \\
        &= \frac{1}{p \cdot \pbar} 
        \left( 
        z_{\ell}^1(a) \cdot (1-p^{x'-\ell+1})
        -
        z_{\ell}^0(a) \cdot (1-\pbar^{x'-\ell+1})
        \right)
        \,.
    \end{align*}
The inequality holds because $p_{a(x)} \le p_{a(k)}$ for all $\ell \le x < x'$. Since $p = p_{a(k)} > \frac{1}{2}$, we have $p^{t} > \pbar^t$ for any $t \ge 1$. Since $z_\ell^1(a) \le z_\ell^0(a)$, the first term is strictly less than the second term, so the expression is negative.
\end{itemize}
The sign of $cost(b) - cost(a)$ is the sign of the sum of these two terms. So $cost(b) - cost(a) < 0$. This contradicts our assumption that $a$ was an optimal ordering.

Finally, we remark on which block the last coin of $S_+$ is in. Since $B_q$ is $1$-biased, $B_{q-1}$ must be $0$-biased or unbiased. If it is $0$-biased, we are done. Otherwise, $B_{q-1}$ is unbiased, and so it contains one coin. If the coin in $B_{q-1}$ has probability of heads less than $\frac{1}{2}$, then $B_q$ is $0$-biased, which is a contradiction. So the coin in $B_{q-1}$ must have probability of heads larger than $\frac{1}{2}$, and so belongs to $S_+$. So the only coin that can violate the statement in the claim is the last coin in $S_+$, which is also the last coin in $B_{q-1}$.
\end{proof}

\subsection{Handling unbiased blocks}
 
\Cref{claim:optimalordering_onespotfromgreedy} implies there is an optimal ordering where at all positions \emph{except one}, the correct choice of coin is the remaining coin with largest or smallest probability of heads. This of course depends on the bias: the largest probability when $0$-biased, and the smallest probability when $1$-biased.

What happens for unbiased blocks? Since this optimal ordering is the same as any one guaranteed by \Cref{claim:unbiasedblockform}, any unbiased block that is not the last block contains only one coin.  
Let the block sequence of the ordering be $[B_1, \dots, B_q]$.
By \Cref{claim:optimalordering_onespotfromgreedy}, unless this block is $B_{q-1}$, it uses either the smallest or largest probability of heads. We argue that \Cref{claim:optimalordering_onespotfromgreedy} still holds if we assume unbiased blocks (that are not next-to-last) always use the coin with \emph{largest} remaining probability of heads. We do this to match our choice of convention in \Cref{alg:greedy}.
\begin{claim}
\label{claim:nearlygreedy_withunbiased}
    There is an optimal ordering $a$ with block sequence $[B_1, \dots, B_q]$ satisfying the properties of \Cref{claim:optimalordering_onespotfromgreedy}, \emph{and} each unbiased block $B_k$ (for $k \ne q-1$) uses the coins with largest probability of heads among coins in $[B_k, \dots, B_q]$.
\end{claim}
\begin{proof}
Start with an ordering $a$ generated by \Cref{claim:optimalordering_onespotfromgreedy}.
    We prove this in cases:
\begin{itemize}
    \item Suppose the unbiased block $B_i$ has $i < q-2$, and uses a coin with probability of heads less than $\frac{1}{2}$. Then the next block $B_{i+1}$ is $0$-biased, and its first coin has the largest probability of heads among all coins in $[B_{i+1}, \dots, B_q]$.
    By \Cref{claim:swap_order_by_one}, we may swap the coin in $B_i$ with the first coin in $B_{i+1}$ without changing the cost; now the coin in $B_i$ has the largest probability of heads among coins in the rest of the ordering.
    \item Suppose the unbiased block is $B_{q-2}$, \emph{and} $B_{q-1}$ has more than one coin. Then the argument in the last case works here as well. 
    \item Suppose the unbiased block is $B_{q-2}$, \emph{and} $B_{q-1}$ has exactly one coin. If $B_{q-2}$ uses a coin with probability less than $\frac{1}{2}$, then $B_{q-1}$ is $0$-biased and uses a coin with probability more than $\frac{1}{2}$. We may switch the two coins without changing the cost; then the ordering in $B_{q-1}$ is now $1$-biased.
    \begin{itemize}
        \item If $B_q$ is also $1$-biased, then the two blocks ``merge'' into a single block after switching the two coins. Then $B_{q-2}$ is the ``next-to-last'' block, and is exempt from the claim.
        \item If $B_q$ is unbiased, then by \Cref{claim:optimalordering_onespotfromgreedy}, the optimal ordering can be greedily chosen. The number of blocks remains $q$ after switching the two coins, so the new coin in $B_{q-2}$ is larger than all subsequent coins in the ordering.
    \end{itemize}
    \item Suppose the unbiased block $B_i$ has $i = q$. Then it trivially uses the coin with smallest remaining probability of heads, unless it has more than one coin. If so, this is also true, because  by \Cref{claim:unbiasedblockform}, all of its coins have probability of heads equal to $\frac{1}{2}$.\qedhere
\end{itemize}
\end{proof}

\subsection{A polynomial-time algorithm}
Putting all of this together, \Cref{claim:nearlygreedy_withunbiased} guarantees an optimal ordering can be constructed by applying this local rule at every position except one:
\begin{quote}\itshape
    If the ordering is $1$-biased, select the remaining coin with smallest probability of heads. \newline
    Otherwise, select the remaining coin with \emph{largest} probability of heads.
\end{quote}
This is exactly the rule of the greedy algorithm in \Cref{alg:greedy}! As a result, there is a $O(n^3)$ algorithm to find the optimal ordering, which we state as \Cref{alg:modifiedgreedy}:

\begin{tabular}{|p{6.5in}}{\underline{\textsc{Modified Greedy Algorithm (\Cref{alg:modifiedgreedy})}}}
\alglabel{alg:modifiedgreedy}

Assume coins are in increasing order; i.e. $p_1 \le \dots \le p_{n}$.

For all $j \in [1, \dots, n]$, generate the following orderings:

\quad \textsc{(Greedy)} For the first $(j-1)$ coins:

\quad \quad If the ordering is $1$-biased, choose the remaining coin with smallest probability of heads.

\quad \quad Otherwise, choose the remaining coin with \emph{largest} probability of heads.

\quad \textsc{(Brute Force)} For every coin $x$ in the set of remaining coins, let $c_{j,x}$ be the following ordering:

\quad \quad Choose coin $x$. 

\quad \quad \textsc{(Greedy)} For the remaining $(n-j)$ coins:

\quad \quad \quad  If the ordering is $1$-biased, choose the remaining coin with smallest probability of heads.

\quad \quad \quad Otherwise, choose the remaining coin with largest probability of heads.

Compute the expected cost for all $c_{j,x}$, and return an ordering $c$ that minimizes this cost.
\end{tabular}
\begin{theorem}
    \Cref{alg:modifiedgreedy} exactly solves the Unanimous Vote problem and runs in time $O(n^3)$.
\end{theorem}
\begin{proof}
\Cref{alg:modifiedgreedy} generates \emph{all} orderings that use the greedy algorithm's rule for at least $(n-1)$ positions. By \Cref{claim:nearlygreedy_withunbiased}, it must find an optimal ordering. \Cref{alg:modifiedgreedy} considers $O(n^2)$ orderings, and the cost of any ordering can be computed in $O(n)$ time.
\end{proof}

\subsection{A faster algorithm}
We can construct a faster algorithm to recover the optimal ordering. This is because one can use the output of the greedy algorithm (\Cref{alg:greedy}) to identify the ``non-greedy position'' in the optimal ordering.

\begin{tabular}{|p{6.5in}}{\underline{\textsc{Faster Modified Greedy Algorithm (\Cref{alg:fastermodifiedgreedy})}}}
\alglabel{alg:fastermodifiedgreedy}

Assume coins are in increasing order; i.e. $p_1 \le \dots \le p_{n}$.

Generate the ordering $b$ from \Cref{alg:greedy}. Let $x$ be the position just prior to the final block.

For each $x' \in [x+1, \dots, n]$, generate the ordering $c_{x'}$ by starting from $b$, moving the coin at position $x'$ to position $x+1$, and moving the coin at position $x$ to the end of the ordering. 

Compute the expected cost for $b$ and for all $c_{x'}$, and return an ordering $c$ that minimizes this cost.
\end{tabular}

\begin{theorem}
\label{thm:intermediatespeed}
    \Cref{alg:fastermodifiedgreedy} exactly solves the Unanimous Vote problem and runs in time $O(n^2)$.
\end{theorem}
\begin{proof}
\Cref{alg:fastermodifiedgreedy} generates \emph{all} orderings that use the greedy algorithm's rule at every position except possibly position $x$. It generates $O(n)$ orderings. The cost of any ordering can be computed in $O(n)$ time. 

We now argue that the algorithm is correct. Consider an optimal ordering $a$ generated by \Cref{claim:nearlygreedy_withunbiased}, and let $x$ be the position just prior to the final block. If $x$ was greedily chosen, then we are done. So we assume the coin at position $x$ was not chosen greedily. By \Cref{claim:optimalordering_onespotfromgreedy}, this can only happen if the final block is $0$-biased or $1$-biased.

We assume the final block is $1$-biased; a symmetric argument holds if it is $0$-biased. Then the coin at position $x$ has probability of heads greater than $\frac{1}{2}$, yet there is a coin in the final block which has probability of heads greater than $p_{a(x)}$. Since the final block is sorted in \emph{increasing} probability of heads, $p_{a(n)} > p_{a(x)}$. Moreover, because the final block is $1$-biased, $z_{x'}^1(a) > z_{x'}^0(a)$ for all $x+1 \le x' \le n$.

Now consider the ordering $b$ output by \Cref{alg:greedy}. Before position $x$, $a$ and $b$ are identical. So $z_x^1(b) = z_x^1(a) \le z_x^0(a) = z_x^0(b)$. But at position $x$, the greedy ordering chooses coin $a(n)$; i.e. the coin with largest remaining probability of heads. So then for all $x+1 \le x' \le n$, 
\begin{align*}
    z_{x'}^1(b) = \frac{p_{a(n)}}{p_{a(x)}} \cdot z_{x'}^1(a) 
    > 
     \frac{p_{a(n)}}{p_{a(x)}} \cdot z_{x'}^0(a) 
     > 
  \frac{\pbar_{a(n)}}{\pbar_{a(x)}} \cdot z_{x'}^0(a) 
  =
  z_{x'}^0(b)\,.
\end{align*}
This implies the greedy ordering is $1$-biased for all positions $x' \in [x+1, \dots, n]$, but not before. Thus, for \emph{both} orderings $a$ and $b$, position $x$ is just prior to the final block. 

As a result, \Cref{claim:nearlygreedy_withunbiased} implies that we can find an optimal ordering if we apply the greedy algorithm's rule to all positions except possibly position $x$. Since \Cref{alg:fastermodifiedgreedy} searches over all such orderings, it will find an optimal ordering.
\end{proof}

In fact, the cost of the all orderings from \Cref{alg:fastermodifiedgreedy} may be simultaneously computed in $O(n)$ time by storing partial sums. This allows us to achieve the runtime claimed in the introduction.
\begin{theorem}[\Cref{thm:mainresult_intro}, restated]
    \Cref{alg:fastermodifiedgreedy} exactly solves the Unanimous Vote problem and runs in time $O(n \log n)$. If the input probabilities are already sorted, the algorithm runs in time $O(n)$. 
\end{theorem}
\begin{proof}
We know that \Cref{alg:fastermodifiedgreedy} is correct by \Cref{thm:intermediatespeed}. It remains to show that the cost of all generated orderings can be simultaneously computed in time $O(n)$.

In \Cref{alg:fastermodifiedgreedy}, we create $c_n$ starting from $b$ and then swapping positions $x$ and $n$. But for any $x' \in [x+1, n-1]$, we may create $c_{x'}$ starting from $c_{x'+1}$ and then swapping positions $x$ and $x'$. So we set $b \defeq c_{n+1}$ (and $b(n+1) \defeq b(x)$), and look at the difference for all $x' \in [x+1,n]$: 
\begin{align*}
    cost(c_{x'}) - cost(c_{x'+1}) = 
    &\left(\prod_{i=1}^{x-1} p_{b(i)} \right)(p_{b(x')}-p_{b(x'+1)})(p_{b(x+1)} + p_{b(x+1)} p_{b(x+2)} + \dots + \prod_{i=x+1}^{x'-1} p_{b(i)})
    \\
    + &\left(\prod_{i=1}^{x-1} \pbar_{b(i)} \right)(\pbar_{b(x')}-\pbar_{b(x'+1)})(\pbar_{b(x+1)} + \pbar_{b(x+1)} \pbar_{b(x+2)} + \dots + \prod_{i=x+1}^{x'-1} \pbar_{b(i)})\,.
\end{align*}
We may factor out $(p_{b(x')}-p_{b(x'+1)})$ from the above expression:
\begin{align*}
    \frac{cost(c_{x'}) - cost(c_{x'+1})}{p_{b(x')}-p_{b(x'+1)}}
    &=
 \left(\prod_{i=1}^{x-1} p_{b(i)} \right) \sum_{j=x+1}^{x'-1} \prod_{i=x+1}^j p_{b(i)} -     \left(\prod_{i=1}^{x-1} \pbar_{b(i)} \right) \sum_{j=x+1}^{x'-1} \prod_{i=x+1}^j \pbar_{b(i)}
    \nonumber \\
    &= \frac{1}{p_{b(x)}} \sum_{j=x+1}^{x'-1} \prod_{i=1}^j p_{b(i)} - \frac{1}{\pbar_{b(x)}} \sum_{j=x+1}^{x'-1} \prod_{i=1}^j \pbar_{b(i)}
    \nonumber \\
    &= \sum_{j=x+2}^{x'} \left(\frac{z_j^1(b) }{p_{b(x)}} - \frac{ z_j^0(b)}{\pbar_{b(x)}} \right)\,.
\end{align*}
We now explain how to compute the cost of all orderings:
\begin{enumerate}
    \item Compute the cost of the greedy ordering $b \defeq c_{n+1}$.
    \item For all $j \in [x+2, \dots, n]$:
    \begin{itemize}
        \item Compute  $z_j^1(b)$ and $z_j^0(b)$ using dynamic programming (e.g., $z_{j+1}^1(b) = p_{b(j)} z_j^1(b)$). 
        \item Compute $y_j \defeq  \frac{z_j^1(b) }{p_{b(x)}} - \frac{ z_j^0(b)}{\pbar_{b(x)}}$.
    \end{itemize} 
    \item Let $\Delta_{x+1} \defeq 1$. For each $x' \in [x+2, \dots, n]$, compute the partial sum $\Delta_{x'} \defeq \sum_{j=x+2}^{x'} y_j$.
    \item  For each $x' \in [n, \dots, x+1]$, compute  $cost(c_{x'}) = cost(c_{x'+1}) + (p_{b(x')} - p_{b(x'+1)}) \cdot \Delta_{x'}$.
\end{enumerate}
Each of these steps takes $O(n)$ time. We can then identify the minimum-cost ordering in $O(n)$ time.\qedhere
\end{proof}
The algorithm of Gkenosis et al.\ \cite{gkenosis2018stochastic} for the \emph{adaptive} variant of the Unanimous Vote problem, described in Section~\ref{sec:additional_related}, also
computes the cost of $n$ orderings and chooses the ordering with lowest cost. We remark that a very similar idea can be used to improve the runtime of their algorithm from $O(n^2)$ to $O(n \log n)$, or $O(n)$ if the input probabilities are already sorted.

\section{Comparing the greedy ordering with the optimal ordering}

\label{sub:greedyvsoptimal}
The greedy ordering constructed by \Cref{alg:greedy} is \emph{not} always optimal.
Here is an explicit example:
\begin{claim}
\label{claim:greedy_not_optimal}
    The greedy ordering for the instance $(p_1, p_2, p_3, p_4) = (0.49, 0.99, 0.99, 1)$ is not optimal.  
\end{claim}
\begin{proof}
    Using  \Cref{alg:greedy}, the greedy ordering for this instance is $(1,0.49,0.99,0.99)$, which has cost
    \begin{align*}
        2 + 1 \cdot 0.49 (1 + 0.99) + 0 \cdot 0.51 (1 + 0.01) = 2.9751\,.
    \end{align*}
    Meanwhile, the ordering $[0.49, 0.99, 0.99, 1]$ has cost
    \begin{align*}
        2 + 0.49 (0.99 + 0.99^2) + 0.51 (0.01 + 0.01^2) = 2.9705\,.\tag*{\qedhere}
    \end{align*}
\end{proof}
In fact, there are instances where \textit{every} algorithm using the natural greedy rule is suboptimal. In particular, consider algorithms that may choose any coin that, if flipped, has the largest probability of terminating the algorithm. At unbiased positions, including the first position, such an algorithm is free to choose any remaining coin. It can be checked that any such algorithm is suboptimal on the instance $(p_1, p_2, p_3, p_4, p_5, p_6) = (0.02,0.24,0.7,0.8,0.9,0.993)$.

For any $\delta > 0$, we can construct an instance where the greedy algorithm (\Cref{alg:greedy}) returns an ordering that on average flips $(1-\delta)$ more coins than the optimal ordering:
\begin{claim}[Lower bound of \Cref{claim:greedy_vs_optimal_intro}]
\label{claim:greedy_minus_opt_atleast_1}
    For all $0 < \delta < 1$, there exists a length-$\lceil \frac{4}{\delta^2}\rceil$ instance where \Cref{alg:greedy} returns an ordering with cost greater than $1-\delta$ plus the cost of an optimal ordering.
\end{claim}
\begin{proof}
Consider one coin with probability of heads $1$ and $n$ coins with probability of heads $1-\epsilon$. There are $(n+1)$ coins in this instance.
Let $a$ be the optimal ordering and $b$ be the ordering returned by \Cref{alg:greedy}.    

\Cref{alg:greedy} will first select the coin with probability of heads $1$, and then use all other coins. This has cost
\begin{align*}
    cost(b) = 2 + \sum_{j=2}^n \left( \prod_{i=1}^j p_{b(i)} + \prod_{i=1}^j \pbar_{b(i)} \right)
    =
    2 + \sum_{j=1}^{n-1} (1- \epsilon)^j \,.
\end{align*}
Meanwhile, consider the ordering that first uses all coins with probability  of heads $1-\epsilon$. Then
\begin{align*}
    cost(a) = 2 + \sum_{j=2}^n \left( (1- \epsilon)^j + (\epsilon)^j \right)\,.
\end{align*}
Then the difference in cost is at least
\begin{align*}
    cost(b) - cost(a) &= (1- \epsilon) - (1- \epsilon)^n - \sum_{j=2}^n  \epsilon^j 
\ge 1 - \frac{\epsilon}{1-\epsilon} - (1 - \epsilon)^n\,.
\end{align*}
We set $\epsilon \defeq \frac{\ln n}{n}$; then $\frac{\epsilon}{1-\epsilon} \le 2 \epsilon$ and $(1-\epsilon)^n \le e^{-\ln n} = \frac{1}{n}$ for all $n \ge 1$. So the difference in cost is at least $1 - 2\epsilon - \frac{1}{n} = 1 - \frac{2 \ln n + 1}{n}$. So for any desired $\delta$, we set $n$ large enough so that $\delta > \frac{2 \ln n + 1}{n}$, and this instance has a difference in cost at least $1 - \delta$. One can verify that $\delta > \frac{2 \ln \lceil 4/\delta^2\rceil + 1}{\lceil 4/\delta^2 \rceil}$ for all $0 < \delta  < 1$.
\end{proof}

We show that the cost of \Cref{alg:greedy} is always at most $1$ more than the optimal solution.
Thus the instances in \Cref{claim:greedy_minus_opt_atleast_1} are the asymptotically worst possible for \Cref{alg:greedy}.
\begin{claim}[Upper bound of \Cref{claim:greedy_vs_optimal_intro}]
\label{claim:greedy_minus_opt_atmost_1}
    For all instances,  \Cref{alg:greedy}  returns an ordering with cost at most $1$ plus the cost of an optimal ordering.
\end{claim}
\begin{proof}
Let $a$ be the optimal ordering guaranteed by \Cref{claim:nearlygreedy_withunbiased} and $b$ be the ordering returned by \Cref{alg:greedy}.  We need to prove that $cost(b) \leq cost(a)+1$.  Let $k$ be the position in $a$ just before the last block.  
By \Cref{claim:nearlygreedy_withunbiased}, $a$ and $b$ agree on all coins in positions $< k$.  

Suppose the last block of $a$ is unbiased. Then by \Cref{claim:optimalordering_onespotfromgreedy}, $a = b$ and so $cost(b)=cost(a)$. There are two remaining cases:

{\it \noindent Case 1:} The last block of $a$ is $1$-biased.

Consider the coin in position $k$ of $a$; it is the only coin that might not have been greedily chosen. Since the last block is $1$-biased, this coin must have probability of heads greater than $\frac{1}{2}$.

We split into two subcases: $p_{a(k)} \ge p_{a(n)}$, and $p_{a(k)} < p_{a(n)}$. 

In the first subcase, $p_{a(k)} \ge p_{a(n)}$, and it follows from \Cref{cor:sorted_within_block} that $p_{a(k)} \ge p_{a(x)}$ for all $k < x \le n$. The second-to-last block must be $0$-biased or unbiased.  Either way, $a = b$, so $cost(b) = cost(a)$.

In the second subcase, $p_{a(k)} < p_{a(n)}$, 
we generate two intermediate orderings:
    \begin{enumerate}
    \item Let $a'$ be the length-$(n+1)$ ordering produced by starting from $a$, and then inserting a coin with probability $p_{a(n)}$ at position $k$.
    \item Let $a''$ be the length-$(n+1)$ ordering produced by starting from $a'$, and then sorting all coins in position $\ge k+1$ in increasing probability of heads.
    \end{enumerate}

Since we added one coin to create $a'$, $cost(a') \le cost(a) + 1$. 
Since $p_{a(n)} > p_{a(k)} > \frac{1}{2}$, the last block of $a'$ begins at position $k+1$. 
By \Cref{claim:swap_order_by_one}, sorting the last block of an ordering can only reduce its cost, so $cost(a'') \le cost(a')$.  At position $k$, ordering $a''$ is either $0$-biased or unbiased.  It follows that the
first $n$ coins of $a''$ exactly match $b$, so $cost(b) \leq cost(a'') \leq cost(a') \leq cost(a)+1$.  
\\

{\it Case 2:} The last block of $a$ is $0$-biased.

The second-to-last block can be either $1$-biased or unbiased. If it is $1$-biased, then the claim follows from a symmetric version of the argument used in Case 1.  

Suppose that the second-to-last block is unbiased. 
Let $a'$ be the ordering starting from $a$, but switching the coins at positions $k$ and $k+1$. Since $a$ is unbiased at position $k$, $cost(a') = cost(a)$, and so $a'$ is also optimal. Then $a'$ is unbiased at position $k$, either $0$-biased, unbiased, or $1$-biased at position $k+1$, and $0$-biased from position $k+2$ through position $n$.
We consider each subcase:
\begin{itemize}
    \item If $a'$ is $0$-biased at position $k+1$, then the last block of $a'$ starts at this position. By \Cref{cor:sorted_within_block}, this block is sorted in decreasing probability of heads. Then $b = a'$, and so $cost(b) = cost(a') = cost(a)$.
    \item If $a'$ is unbiased at position $k+1$, then the coin at position $k$ has probability of heads $\frac{1}{2}$. Since $a'$ is $0$-biased at position $k+2$, the coin at position $k+1$ has probability of heads less than $\frac{1}{2}$.  If $k + 1 = n$, then $b = a'$, and so $cost(b) = cost(a') = cost(a)$. Otherwise, we repeat \textit{Case 2} on the ordering $a'$. 
    \item If $a'$ is $1$-biased at position $k+1$, then the claim follows from a symmetric version of the argument used in Case 1.\qedhere
\end{itemize}
\end{proof}

The proof of \Cref{claim:greedy_minus_opt_atmost_1} relies on the convention used by \Cref{alg:greedy} for choosing the coins in unbiased positions.
It is natural to consider an algorithm that also generates a greedy ordering using the opposite convention (smallest $p_i$ rather than largest $p_i$) and outputs the best of the two orderings.
This algorithm is optimal unless $p_1 \le \frac{1}{2} \le p_n$. 
If $p_1 \le \frac{1}{2} \le p_n$, there is at least a $50\%$ chance of terminating in two flips. One can modify the analysis in \Cref{claim:greedy_minus_opt_atmost_1} to show that this algorithm finds an ordering with cost at most $\frac{1}{2}$ more than the cost of the optimal ordering. 
We may observe that this difference is asymptotically tight by adding a coin with probability of heads $0.5 - \epsilon$ to the instance in \Cref{claim:greedy_minus_opt_atleast_1}.

\section{Adaptivity gap}
\label{sec:adaptivity_gap}
The \emph{adaptivity gap} is the worst-case ratio of the cost of the optimal non-adaptive ordering and the cost of the optimal adaptive strategy. We first show this is at least $1.2 - O(\frac{1}{2^n})$, where $n$ is the number of coins.
\begin{claim}[Lower bound of \Cref{thm:adaptivitygap_intro}]
\label{claim:adaptivitygap_lowerbound}
    The adaptivity gap is at least $1.2 - O\left(\frac{1}{2^{n}} \right)$.
\end{claim}
\begin{proof}
Recall that the cost of a non-adaptive ordering $a$ is
$$
cost(a) = 2 + \sum_{j \ge 3} \left( \prod_{i=1}^{j-1} p_{a(i)} + \prod_{i=1}^{j-1} \overline{p}_{a(i)}) \right)\,.
$$
Consider one coin with $0$ probability of heads and $n-1$ coins with $\frac{1}{2}$ probability of heads.
Any non-adaptive ordering is determined by the position of the coin with $0$ probability of heads. Suppose it is at position $k$. Then the cost is thus
$$
2 + \sum_{j =3}^{k} (0.5^{j-1} + 0.5^{j-1}) + \sum_{j=k+1}^{n} (0.5^{j-2}  + 0)
=
2 + \sum_{j=3}^k 0.5^{j-2} + \sum_{j=k+1}^n 0.5^{j-2}
= 2 + \sum_{j=3}^n 0.5^{j-2}\,.
$$
This sum is independent of $k$; thus, all non-adaptive orderings have the same cost of
$$
2 +  \frac{0.5 - 0.5^{n-1}}{0.5} = 3 - \frac{1}{2^{n-2}}\,.
$$
Consider the adaptive strategy that flips a coin with $\frac{1}{2}$ probability of heads, then chooses all other coins optimally. The strategy terminates unless all observed coins are tails, or we run out of coins to flip. The cost of this strategy is
$$
2 + \sum_{j=3}^n 0.5^{j-1} = 2 + \frac{0.5 - 0.5^{n}}{0.5} = 2.5 - \frac{1}{2^{n-1}}\,.
$$
So the adaptivity gap is at least
\begin{align*}
\frac{3 - 0.5^{n-2}}{2.5 - 0.5^{n-1}} = \frac{3}{2.5} - \frac{0.5^{n-2} - \frac{3}{2.5} \cdot 0.5^{n-1}}{2.5 - 0.5^{n-1}} = 1.2 - \frac{0.8 \cdot 0.5^{n-1}}{2.5 - 0.5^{n-1}} =  1.2 - O\left(\frac{1}{2^n}\right)\,.\tag*{\qedhere}
\end{align*}
\end{proof}

Our proof of the upper bound uses a structure in any ordering we call \emph{crossover point}.
Each coin before or at the crossover point has a good chance of terminating the sequence, but this is not true after this point.
\begin{definition}[Crossover point]
\label{defn:crossoverpt}
    Fix an ordering $a$. The \emph{crossover point} of $a$ is the largest position $t$ where for all $2 \le x \le t$, the probability of terminating after flipping the coin at position $a(x)$, assuming termination has not yet occurred, is at least $\frac{1}{2}$. If there is no such $t$, we say $t = 1$.
\end{definition}
The crossover point $t$ of any ordering $a$ naturally induces a partition of $[a(1), \dots, a(n)]$ with a first \emph{region} $[a(1), \dots, a(t)]$ and a (possibly empty) second region $[a(t+1), \dots, a(n)]$. In the first region,  the sequence is likely to terminate quickly.
By \Cref{claim:nonfinal_probatleasthalf}, the crossover point of an optimal ordering guaranteed by \Cref{claim:unbiasedblockform} occurs in the final block.

{\ifarxiv To illustrate our technique, we now prove\fi} that the adaptivity gap is at most $1.5$.
\begin{claim}[{\ifarxiv Weak upper bound\fi} of \Cref{thm:adaptivitygap_intro}]
\label{claim:adaptivitygap_1.5}
    The adaptivity gap is at most $1.5$.
\end{claim}
\begin{proof}
    We try to construct the following two orderings:
\begin{itemize}
    \item The first coin has the smallest probability of heads. For every subsequent position, if the ordering is $0$-biased, use the coin with \emph{smallest} probability of heads at \emph{least} $\frac{1}{2}$; else, use the coin with \emph{smallest} probability of heads.
    \item The first coin has the \emph{largest} probability of heads. For every subsequent position, if the ordering is $1$-biased, use the coin with \emph{largest} probability of heads at \emph{most} $\frac{1}{2}$; else, use the coin with \emph{largest} probability of heads.
\end{itemize}
We explain why one of these orderings can be successfully created. Suppose in the first ordering, we run out of coins with probability at least $\frac{1}{2}$. Since the ordering is $0$-biased at this point, and all other coins have probability of heads less than $\frac{1}{2}$, we have $\prod_{i=1}^n p_i < \prod_{i=1}^n \pbar_i$. However, if in the second ordering, we run out of coins with probability \emph{at most} $\frac{1}{2}$, then $\prod_{i=1}^n p_i > \prod_{i=1}^n \pbar_i$. Only one of these events can occur.

These orderings are constructed to guarantee the following property. Suppose the crossover point is $t$. Then for all $s \ge t$, the smallest (or the largest) $s$ coins are in the first $s$ positions. This makes the last few coins in the ordering look similar to the adaptive strategy.

Choose an ordering $a$ that was successfully created. By definition, the probability of terminating the sequence is at least $\frac{1}{2}$ for any coin in position at most $t$. The cost of this algorithm is thus at most
\begin{align*}
    cost(a) = 2 + \sum_{j=3}^n (z^1_j(a) + z_j^0(a)) \le 2 + \left(\frac{1}{2} + \dots + \frac{1}{2^{t-1}} \right) + \sum_{j = t+2}^n  (z^1_j(a) + z_j^0(a)) \,.
\end{align*}
For the rest of this proof, we assume that we have chosen the first ordering; the argument if we choose the second ordering is symmetric. 
By our assumption, all coins in the second region (if it exists) have probability of heads greater than $\frac{1}{2}$.  So then $z_{j+1}^0(a) = \pbar_{j} \cdot z_{j}^0(a) \le \frac{1}{2} z_{j}^0(a)$ for all $j > t$. So $\sum_{j = t+2}^n z_j^0(a) \le z_{t+1}^0(a) \le \frac{1}{2^{t-1}}$. So the algorithm has cost at most
\begin{align*}
cost(a) \le 1 + \left(1 + \frac{1}{2} + \dots + \frac{1}{2^{t-1}} \right) + \frac{1}{2^{t-1}} + \sum_{j=t+2}^n z_j^1(a) = 3 + \sum_{j=t+2}^n z_j^1(a)\,.
\end{align*}
Meanwhile, the cost of any adaptive strategy for the Unanimous Vote problem is lower-bounded by the cost of the best strategy to see at least one tails while flipping at least two coins, which is $2 + \sum_{j=2}^{n-1} \prod_{i=1}^j p_i$ (i.e. flip the coins in increasing order of $p_i$). So the adaptivity gap is at most
\begin{align*}
    \frac{3 + \sum_{j=t+2}^n z_j^1(a)}{2 + \sum_{j=2}^{n-1} \prod_{i=1}^j p_i} \le \max\left(\frac{3}{2}, \frac{\sum_{j=t+2}^n z_j^1(a)}{ \sum_{j=2}^{n-1} \prod_{i=1}^j p_i}\right)\,.
\end{align*}
By construction, for all $j \ge t$, we have $z_{j+1}^1(a) = \prod_{i=1}^{j} p_i$. So the second ratio equals 
\begin{align*}
    \frac{\sum_{j=t+1}^{n-1} \prod_{i=1}^j p_i}{\sum_{j=2}^{n-1} \prod_{i=1}^j p_i}\,,
\end{align*}
which is at most $1$ since the crossover point $t \ge 1$. So the adaptivity gap is at most $\max(\frac{3}{2}, 1) = \frac{3}{2}$.
\end{proof}
{\ifarxiv
We can generalize the ideas in \Cref{claim:adaptivitygap_1.5} to sharpen the upper bound to $1.2+o(1)$. Our proof of this result is complicated and requires the entirety of \Cref{apx:adaptivity_gap_upperbound}. 
We report the result as \Cref{thm:adaptivitygap_body}.
\fi}
\section{Open questions}
We conclude by highlighting several open questions related to the Unanimous Vote problem. 

First, as discussed in \Cref{sec:additional_related}, SBFE problems are often studied with \emph{costs}, and it is natural to consider the problem in that setting, i.e., where it costs $c_i$ to flip coin $i$ and the goal is to minimize the expected cost of coins flipped before termination. It is unclear if our algorithm would generalize to this setting.

It is also interesting to consider the unit-cost setting for other symmetric Boolean functions beyond the Unanimous Vote (i.e., not-all-equal) function. For example, recent work provides a PTAS for
the unit-cost, non-adaptive SBFE problem for the $k$-of-$n$ function~\cite{nielsen2025nonadaptive}.  We conjecture that there is a polynomial-time exact algorithm. 
More generally, 
It is an open question whether there is a polynomial-time algorithm for the unit-cost SBFE problem for arbitrary symmetric functions, in either the adaptive or the non-adaptive setting. Some related work in the adaptive setting appeared in the information theory literature, 
motivated by problems involving distributed sensors
\cite{AcharyaJafarpourOrlitsky:2011,DasJafarpourOrlitsky:2012,KowshikKumar:2013}.

Finally, the interchange arguments used in our analysis suggest that there may be a simple local-search algorithm for the Unanimous Vote problem based on interchanges. It would be interesting to find such an algorithm that converges to an optimal solution in polynomial time.
 We note that the algorithm would need to interchange non-adjacent elements of the ordering; under adjacent interchanges only, there can be locally optimal solutions that are not globally optimal.

\section*{Acknowledgements}
Thanks to R. Teal Witter for collaborating on early stages of this project.
K.M acknowledges
support from AFOSR (FA9550-21-1-0008). This material is based
upon work partially supported by the National Science Foundation Graduate Research Fellowship under
Grant No. 2140001.
C.M. was partially supported by NSF Award No. 2045590.  L.H. was partially supported by NSF Award No. 1909335.

\clearpage
\newpage

\bibliography{_references}
\bibliographystyle{_alpha-betta-url}
\clearpage
\newpage
\appendix
\crefalias{section}{appendix}
\crefalias{subsection}{appendix}

\section{Improved adaptivity gap upper bound}
\label{apx:adaptivity_gap_upperbound}
Here we show the adaptivity gap is at most $1.2 + O(\frac{1}{2^n})$, where $n$ is the number of coins.
To prove this upper bound, we split the set of all instances into a large number of cases based on a structure in the optimal ordering we call \emph{crossover point} (\Cref{defn:crossoverpt}). 
We provide some intuition in \Cref{sec:adaptivity_gap_intuition} and \Cref{fig:heuristicplot}, and outline the proof strategy in \Cref{sub:proofoutline} and \Cref{fig:adaptivitygap_outline}.
We analyze each case in \Cref{sec:deferredlemmas}. 
Each case requires an upper bound of an explicit multivariable expression, which we show in \Cref{sec:deferredmathfacts}. 
 
The optimal \emph{adaptive} strategy is easy to describe~\cite{gkenosis2018stochastic}.
Suppose the first coin comes up heads.
Then the optimal strategy flips all remaining coins in increasing probability of heads. If instead the first coin comes up \emph{tails}, the optimal strategy flips all remaining coins in decreasing probability of heads. One can use brute force to identify the choice of first coin: generate the $n$ strategies corresponding to each possible choice of the first coin, compute their expected costs, and output the strategy with minimum expected cost.

\subsection{Intuition}
\label{sec:adaptivity_gap_intuition}

Fix some constant $k$. Consider the $k! \cdot \binom{n}{k}$ orderings that fix the first $k$ coins arbitrarily,
and then for all subsequent coins, greedily choose the coin that would maximize the probability of the sequence terminating.
To upper bound the adaptivity gap, we show that there is some $k$ where the minimum cost among the $k! \cdot \binom{n}{k}$ orderings is \textit{always} at most $1.2 + o(1)$ times the cost of the best adaptive strategy. 

It turns out that $k=5$ is enough; here we provide a heuristic calculation. Consider the family of instances with $p_1 = 0$ and $p_n = \dots = p_{n-k} = 1-x$, for some $x \le \frac{1}{2}$. Note that this family of instances includes the extremal instance in \Cref{claim:adaptivitygap_lowerbound}.
For each $k \ge 2$, we can upper bound the adaptivity gap as 
\begin{align*}
    \frac{2 + 2x}{2 + \frac{x^2}{1-x}}\,,
    \frac{2 + x + 2x^2}{2 + \frac{x^2}{1-x}}\,,
        \frac{2 + x + x^2+2x^3}{2 + \frac{x^2}{1-x}}\,,
    \frac{2 + x + x^2+x^3+2x^4}{2 + \frac{x^2}{1-x}}\,, \dots
\end{align*}
The numerators are increasingly good Taylor approximations to the value $2 + \frac{x}{1-x}$. All expressions equal $1.2$ at $x = \frac{1}{2}$. The maximum value of these expressions over $0 \le x \le \frac{1}{2}$ converges to $1.2$ as $k \to \infty$, in fact reaching this value when the numerator has at least $k=5$ terms. See \Cref{fig:heuristicplot} or this \href{https://www.desmos.com/calculator/crkh0i10i5}{online plot}.

\begin{figure}[h!]
    \centering
\includegraphics[width=\linewidth]{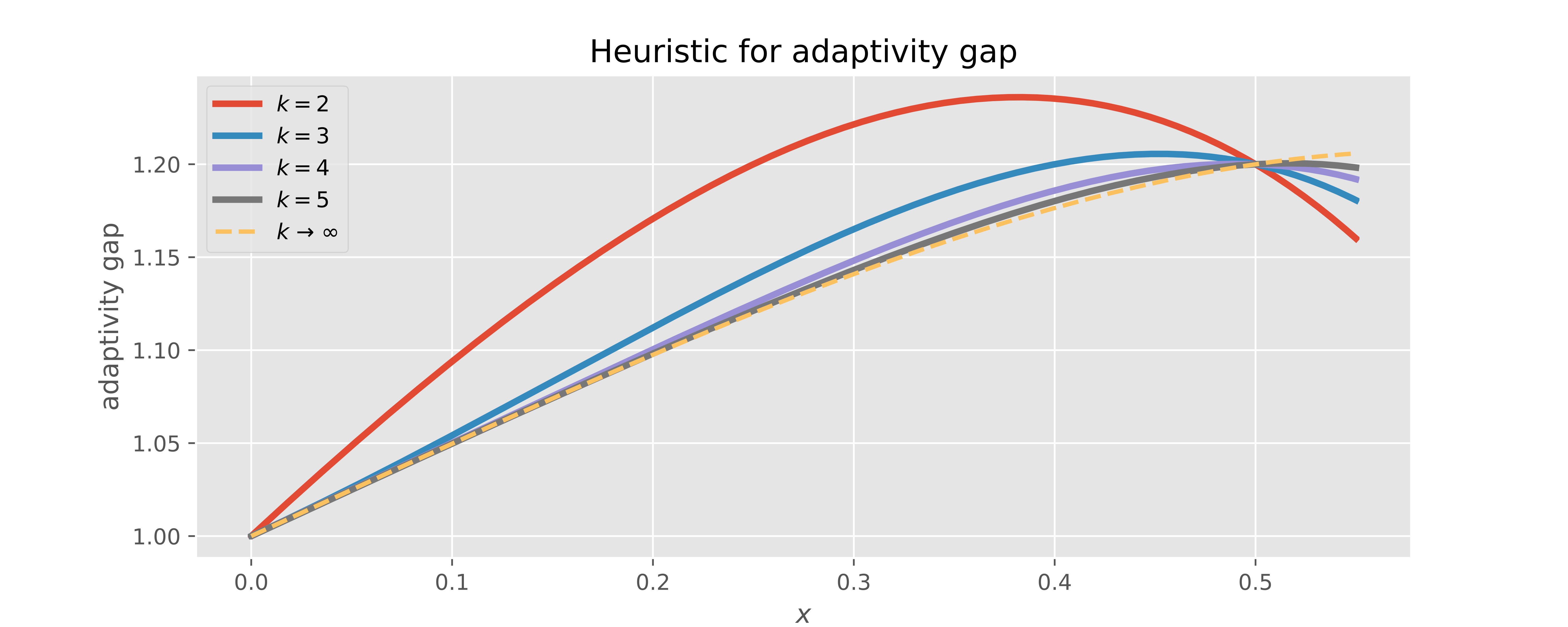}
    \caption{\small Heuristic for the adaptivity gap. The maximum value over $0 \le x \le \frac{1}{2}$ is at least $1.2$ for all $k$, and at most $1.2$ when $k \ge 5$. At $k \to \infty$, the heuristic recovers the adaptivity gap of the instance $\{0, 1-x,1-x,1-x,\dots\}$.
    }
    \label{fig:heuristicplot}
\end{figure}

\subsection{Proof outline}
\label{sub:proofoutline}
We now outline the proof of the upper bound, following \Cref{fig:adaptivitygap_outline}.
We split the set of instances into cases depending on the values $(p_1, p_n)$, the probability of the coin $\ks$ chosen by the optimal adaptive strategy, and the crossover points $t$ of orderings we define below.
Since this proof is quite long, we move some technical parts to \Cref{sec:deferredlemmas} and \Cref{sec:deferredmathfacts}.

\begin{figure}[ht]
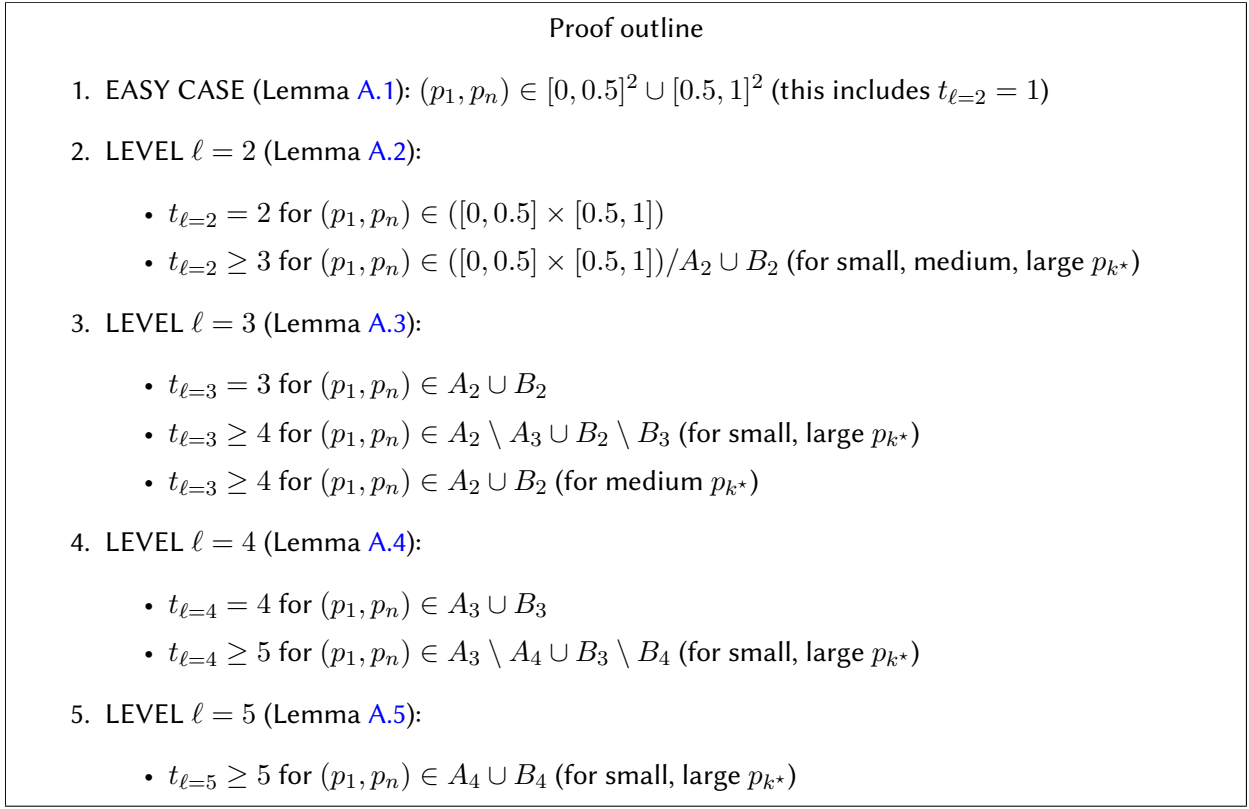

\begin{mdframed}
    {\sffamily
\centering Proof outline
    \begin{enumerate}
        \item \textsf{EASY CASE} (\Cref{lemma:adaptivitygap_1.2_onesideofhalf}): $(p_1, p_n) \in [0, 0.5]^2 \cup [0.5, 1]^2$ (this includes $t_{\ell =2 } = 1$)
        \item \textsf{LEVEL} $\ell = 2$ (\Cref{lemma:adaptivity_ell_2}):
        \begin{itemize}
            \item $t_{\ell = 2} = 2$ for $(p_1, p_n) \in ([0, 0.5] \times [0.5,1])$
        \item $t_{\ell = 2} \ge 3$ for $(p_1, p_n) \in ([0, 0.5] \times [0.5,1]) / A_2 \cup B_2$ (for small, medium, large $p_{\ks}$)
        \end{itemize}
        \item \textsf{LEVEL} $\ell = 3$ (\Cref{lemma:adaptivity_ell_3}):
        \begin{itemize}
            \item $t_{\ell = 3} = 3$ for $(p_1,p_n) \in A_2 \cup B_2$
                \item $t_{\ell = 3} \ge 4$ for $(p_1,p_n) \in A_2\setminus A_3 \cup B_2\setminus B_3$ (for small, large $p_{\ks}$) \item  $t_{\ell = 3} \ge 4$ for $(p_1,p_n) \in A_2 \cup B_2$ (for medium $p_{\ks}$)
        \end{itemize}
        \item \textsf{LEVEL} $\ell = 4$ (\Cref{lemma:adaptivity_ell_4}):
        \begin{itemize}
          \item $t_{\ell = 4} = 4$ for $(p_1, p_n) \in A_3 \cup B_3$
        \item $t_{\ell = 4} \ge 5$ for $(p_1, p_n) \in A_3\setminus A_4 \cup B_3\setminus B_4$ (for small, large $p_{\ks}$)
        \end{itemize}
        \item \textsf{LEVEL} $\ell = 5$ (\Cref{lemma:adaptivity_ell_5}):
        \begin{itemize}
        \item $t_{\ell = 5} \ge 5$ for $(p_1, p_n) \in A_4 \cup B_4$ (for small, large $p_{\ks}$)
        \end{itemize}
    \end{enumerate}
        }
\end{mdframed}
    \caption{ \small
Schematic for the proof that the adaptivity gap is at most $1.2 + o(1)$ (\Cref{thm:adaptivitygap_body}). We consider all choices of $(p_1, p_n, p_{\ks})$ and crossover points $t$ of various orderings. The sets $A_2,B_2,A_3,B_3,A_4,B_4$, and the ranges of $p_{\ks}$ are defined in the proof outline (\Cref{sub:proofoutline}). Each level handles a set of instances the previous level could not handle. Note that if $t_{\ell} \ge \ell + 1$, then $t_{\ell+1} \ge \ell + 1$.}
    \label{fig:adaptivitygap_outline}
\end{figure}

We first claim that if all coins have probability of heads less than $\frac{1}{2}$ or more than $\frac{1}{2}$, the adaptivity gap is at most $1.2$. We defer the proof of this lemma (and all future lemmas) to \Cref{sec:deferredlemmas}.
\begin{restatable}{lemma}{adaptivityeasycase}
\label{lemma:adaptivitygap_1.2_onesideofhalf}
    Suppose the coins all have probability of heads at least $\frac{1}{2}$, or all have probability of heads at most $\frac{1}{2}$. Then the adaptivity gap is at most $1.2$. 
\end{restatable}
From here, we assume $p_1 \le \frac{1}{2} \le p_n$. We break the analysis into \emph{levels} $\ell$, for $\ell \in \{2,3,4,5\}$. 
At each level $\ell$, we construct two orderings (as in \Cref{claim:adaptivitygap_1.5}).
In each ordering, the first $\ell$ coins are chosen ``greedily'' as in \Cref{alg:greedy}: The first two coins have the smallest and largest probability of heads, and for each following coin, we choose the coin that maximizes the probability of termination. After these coins, each ordering uses one of the below strategies:
\begin{itemize}
    \item For every subsequent position, if the ordering is $0$-biased, use the coin with \emph{smallest} probability of heads at \emph{least} $\frac{1}{2}$; otherwise, use the coin with \emph{smallest} probability of heads.
    \item For every subsequent position, if the ordering is $1$-biased, use the coin with \emph{largest} probability of heads at \emph{most} $\frac{1}{2}$; otherwise, use the coin with \emph{largest} probability of heads.
\end{itemize}
As in \Cref{claim:adaptivitygap_1.5}, at least one of these strategies successfully creates an ordering. For this analysis, we assume that we use the ordering $a$ that chooses the coins with smallest probability of heads; the argument in the other case is symmetric.
Let's say the crossover point of this ordering is $t$,\footnote{Sometimes we refer to this as $t_{\ell}$, but we drop the subscript when the level $\ell$ is clear from context.} and assume that $t \ge \ell$. Then the cost of this ordering is
\begin{align*}
    cost(a) \le 2 + \sum_{j=3}^{\ell} (z_j^1(a) + z_j^0(a)) +  \sum_{j=\ell+1}^{t+1} (z_j^1(a) + z_j^0(a)) + \sum_{j=t+2}^{n} (z_j^1(a) + z_j^0(a))\,.
\end{align*}
As in \Cref{claim:adaptivitygap_1.5}, since all coins after position $t$ have probability at least $\frac{1}{2}$, we have $z_{j+1}^0(a) \le \frac{1}{2} z_j^0(a)$ for all $j > t$, and so $\sum_{j=t+2}^n z_j^0(a) \le z_{t+1}^0(a)$. 
Meanwhile, the first region has probability of termination at least $\frac{1}{2}$, so the cost is at most
\begin{align*}
    cost(a) &\le 2 + \sum_{j=3}^{\ell} \left(z_j^1(a) + z_j^0(a)\right) +  (z_{\ell+1}^1(a) + z_{\ell+1}^0(a))(1 + \frac{1}{2} + \dots + \frac{1}{2^{t-\ell}}) + z_{t+1}^0(a) + \sum_{j=t+2}^{n} z_j^1(a)
    \\
    &\le 2 + \sum_{j=3}^{\ell} \left(z_j^1(a) + z_j^0(a)\right) +  2(z_{\ell+1}^1(a) + z_{\ell+1}^0(a)) - z_{t+1}^1(a) +  \sum_{j=t+2}^{n} z_j^1(a)\,.
\end{align*}
Meanwhile, we write down the cost of the adaptive strategy $\mathbf{a}_{\ks}$. Let $S = [1,2,\dots,n] \setminus \{\ks\}$ be the ordered list of numbers from $1$ to $n$ without $\ks$.  Then 
\begin{align*}
    cost(\mathbf{a}_{\ks}) &= 2 
    + (\pbar_{\ks}\pbar_{S[n-1]} + \pbar_{\ks}\pbar_{S[n-1]} \pbar_{S[n-2]} + \dots)
    +
    (p_{\ks}p_{S[1]} + p_{\ks}p_{S[1]} p_{S[2]} + \dots)
    \\
    &= 2 + \pbar_{\ks} \sum_{j=1}^{n-2} \prod_{i=1}^j \pbar_{S[n-i]} +  p_{\ks} \sum_{j=1}^{n-2} \prod_{i=1}^j p_{S[i]}\,.
\end{align*}
So assuming $t \ge \ell$, the adaptivity gap is at most
\begin{align}
\label{eqn:adaptivitygap_arbitrary_ell}
    \frac{
     2 + \sum_{j=3}^{\ell} \left(z_j^1(a) + z_j^0(a)\right) +  2(z_{\ell+1}^1(a) + z_{\ell+1}^0(a)) - z_{t+1}^1(a) +  \sum_{j=t+2}^{n} z_j^1(a)
    }
    {
    2 + \pbar_{\ks} \sum_{j=1}^{n-2} \prod_{i=1}^j \pbar_{S[n-i]} +  p_{\ks} \sum_{j=1}^{n-2} \prod_{i=1}^j p_{S[i]}
    }\,.
\end{align}
For each level $\ell$, we will analyze this expression using a ``rescaling trick''. For arbitrary positive $\alpha,\beta,\gamma,\delta$, we have $\frac{\alpha+\beta}{\gamma+\delta} \le \max(\frac{\alpha + \beta \cdot x}{\gamma + \delta \cdot x}, \frac{\beta}{\delta})$, for any $0 \le x \le 1$. We use this fact to rescale the last summation of the numerator and denominator in the expression above. Precisely, we identify a term $\beta$ in the numerator and $\delta$ in the denominator, and show that $\frac{\beta}{\delta} \le r$. Then we rescale the terms to $\beta \cdot x$ and $\delta \cdot x$ for some $0 \le x \le 1$ in the above expression. The adaptivity gap is then at most the maximum of the new expression and $r$. In general, we want to rescale small enough so that the numerator can simplify, but not too small for the denominator to simplify.

We describe how we can analyze instances with orderings at progressively larger $\ell$. 
As in \Cref{claim:adaptivitygap_1.5}, if $\prod_{i=1}^n p_i \ge \prod_{i=1}^n \pbar_i$, the first ordering (using \emph{smallest} probability of heads) will successfully create an ordering for all $\ell$; otherwise, the second ordering (using \emph{largest} probability of heads) will successfully create an ordering for all $\ell$. Now, if an ordering created with a level-$\ell$ ordering has crossover point $t_{\ell} > \ell$, then there is at least one more coin we can greedily assign. So there must be a level-$(\ell+1)$ ordering, created with the same ordering, with crossover point $t_{\ell+1} > \ell$. Thus, if we show that an instance obeys $t_{\ell} > \ell$, we may analyze the adaptivity gap from the level-$(\ell+1)$ ordering, with crossover point $t_{\ell+1} \ge \ell + 1$, created via the same ordering. And to start the analysis, since $p_1 \le \frac{1}{2} \le p_n$, the crossover point of a level $\ell=2$ ordering is at least $2$. 

Let $A_2 \defeq [0,0.2] \times [0.5,0.76]$ and $B_2 \defeq [0.24,0.5] \times [0.8,1]$. Then this ordering gives adaptivity gap at most $1.2 + o(1)$ for a large swath of instances:
\begin{restatable}{lemma}{adaptivityellequalstwo}
\label{lemma:adaptivity_ell_2}
    The level-$2$ ordering has adaptivity gap of at most $1.2 + O(\frac{1}{2^n})$ when $(p_1, p_n) \in ([0, 0.5] \times [0.5,1]) / A_2 \cup B_2$ holds, and when both $t_{\ell=2} = 2$ and $(p_1, p_n) \in ([0, 0.5] \times [0.5,1])$ hold.
\end{restatable}
We now move to level $\ell = 3$ to handle some cases that the previous level could not handle.  Note that we can safely assume $t_{\ell = 3} \ge 3$, since the previous level handled the case $t_{\ell=2} = 2$.  

Let $A_3 \defeq  [0,0.2] \times [0.5,0.65]$ and  $B_3 \defeq  [0.35,0.5] \times [0.8,1]$. 
Then the level $\ell=3$ ordering gives adaptivity gap at most $1.2 + o(1)$ for another swath of instances:
\begin{restatable}{lemma}{adaptivityellequalsthree}
\label{lemma:adaptivity_ell_3}
    The level-$3$ ordering has adaptivity gap of at most $1.2 + O(\frac{1}{2^n})$ when $(p_1, p_n) \in A_2\setminus A_3 \cup B_2\setminus B_3$ holds, and when both $t_{\ell=3} = 3$ and $(p_1, p_n) \in A_2 \cup B_2$ hold, and when both $\frac{1}{2} \le p_{\ks} \le \frac{p_n}{1.2}$ and $(p_1, p_n) \in A_2 \cup B_2$ hold.
\end{restatable}

We now move to level $\ell = 4$ to handle some cases that the previous levels could not handle. Note that we can assume $t_{\ell = 4} \ge 4$, since the previous level handled the case $t_{\ell=3} = 3$. 
Moreover, we can assume $p_{\ks} \le \frac{1}{2}$ or $p_n \ge 1.2 p_{\ks}$.

Let  $A_4 \defeq  [0,0.2] \times [0.5,0.58]$ and $B_4 \defeq  [0.42,0.5] \times [0.8,1]$. Then the level $\ell=4$ ordering gives adaptivity gap of at most $1.2 + o(1)$ for yet another swath of instances:
\begin{restatable}{lemma}{adaptivityellequalsfour}
\label{lemma:adaptivity_ell_4}
    The level-$4$ ordering has adaptivity gap of at most $1.2 + O(\frac{1}{2^n})$ when both $t_{\ell=4} = 4$ and $(p_1, p_n) \in A_3 \cup B_3$ hold, and when both $p_{\ks} \in [p_1,\frac{1}{2}] \cup [\frac{p_n}{1.2},p_n]$ and $(p_1, p_n) \in A_3\setminus A_4 \cup B_3\setminus B_4$ hold.
\end{restatable}

Finally, we use level $\ell = 5$ to handle the remaining cases that all previous levels could not handle. We assume $t_{\ell = 5} \ge 5$, since the previous level handled the case $t_{\ell=4} = 4$. We can assume $p_{\ks} \le \frac{1}{2}$ or $p_n \ge 1.2 p_{\ks}$.

\begin{restatable}{lemma}{adaptivityellequalsfive}
\label{lemma:adaptivity_ell_5}
    The level-$5$ ordering has adaptivity gap of at most $1.2 + O(\frac{1}{2^n})$ when both $p_{\ks} \in [p_1,\frac{1}{2}] \cup [\frac{p_n}{1.2},p_n]$ and $(p_1, p_n) \in A_4 \cup B_4$ hold.
\end{restatable}

Combining all of this, we may state the result:
\begin{theorem}[Upper bound of \Cref{thm:adaptivitygap_intro}]
\label{thm:adaptivitygap_body}
    The adaptivity gap is at most $1.2 + O(\frac{1}{2^n})$.
\end{theorem}
\begin{proof}
The theorem follows by combining \Cref{lemma:adaptivity_ell_2,lemma:adaptivity_ell_3,lemma:adaptivity_ell_4,lemma:adaptivity_ell_5}, which cover all instances with $p_1 \le \frac{1}{2} \le p_n$, and \Cref{lemma:adaptivitygap_1.2_onesideofhalf}, which covers the remaining instances.
\end{proof}


\subsection{Deferred lemmas}
\label{sec:deferredlemmas}
As we use \Cref{eqn:adaptivitygap_arbitrary_ell} so often in the appendix, we restate it here for convenience:
\begin{align*}
    \frac{
     2 + \sum_{j=3}^{\ell} \left(z_j^1(a) + z_j^0(a)\right) +  2(z_{\ell+1}^1(a) + z_{\ell+1}^0(a)) - z_{t+1}^1(a) +  \sum_{j=t+2}^{n} z_j^1(a)
    }
    {
    2 + \pbar_{\ks} \sum_{j=1}^{n-2} \prod_{i=1}^j \pbar_{S[n-i]} +  p_{\ks} \sum_{j=1}^{n-2} \prod_{i=1}^j p_{S[i]}
    }\,.
\end{align*}
Again, for this analysis, we assume that we use the ordering $a$ that chooses the coins with smallest probability of heads; the argument in the other case is symmetric. Each of these proofs has roughly the same structure and uses elementary techniques, but the number of variables increases each time.

\subsubsection{Easy case}
\adaptivityeasycase*
\begin{proof}
  \emph{Any} strategy to flip the coins has cost at least $2 + \sum_{j=2}^{n-1} \prod_{i=1}^j p_i + \sum_{j=2}^{n-1} \prod_{i=1}^j \pbar_{n+1-i}$. This is because we always need 2 coins; the minimum probability of getting two heads is $p_1 p_2$; the minimum probability of getting two tails is $\pbar_n \pbar_{n-1}$; and so on. So this lower-bounds the cost of the adaptive strategy.

Suppose all coins have probability of heads at least $\frac{1}{2}$. (A symmetric argument works if all coins have probability of heads at \emph{most} $\frac{1}{2}$.) Let $a$ be the ordering that chooses $a(i) \defeq i$. Then the adaptivity gap is at most 
\begin{align*}
    \frac{2 +  \sum_{j=2}^{n-1} \prod_{i=1}^j p_i + \sum_{j=2}^{n-1} \prod_{i=1}^j \pbar_i}
    {2 + \sum_{j=2}^{n-1} \prod_{i=1}^j p_i + \sum_{j=2}^{n-1} \prod_{i=1}^j \pbar_{n+1-i}}\,.
\end{align*}
Let's inspect each of these terms. The first summation is on the numerator and denominator, and it has value at least $\frac{1}{2} - \frac{1}{2^{n-1}}$. The second summation in the numerator is at most $\frac{1}{2} - \frac{1}{2^{n-1}}$. So the adaptivity gap is at most 
\begin{align*}
    \frac{2.5 - \frac{1}{2^{n-1}}+ x}{2.5 - \frac{1}{2^{n-1}}}
\end{align*}
for some $x \le \frac{1}{2} - \frac{1}{2^{n-1}}$. So this ratio is maximized at the largest $x$, i.e. $\frac{3 - 0.5^{n-2}}{2.5 - 0.5^{n-1}} = \frac{3}{2.5} - \frac{0.5^{n-2} - \frac{3}{2.5} \cdot 0.5^{n-1}}{2.5 - 0.5^{n-1}}$. So the adaptivity gap is at most $1.2$ for all positive integer $n$.  
\end{proof}

\subsubsection{Level \texorpdfstring{$\ell=2$}{2}}

\adaptivityellequalstwo*
\begin{proof}
    We start from \Cref{eqn:adaptivitygap_arbitrary_ell}.
    We use the rescaling trick for the summation $\sum_{j=t+2}^n z_j^1(a)$ in the numerator and $p_{\ks} \sum_{j=1}^{n-2} \prod_{i=1}^j p_{S[i]}$ in the denominator. In this level, for all $s \ge t$, the smallest $s-1$ coins \emph{and} the largest coin are in the first $s$ positions. So $\sum_{j=t+2}^n z_j^1(a) = \sum_{j = t}^{n-2} p_n \prod_{i=1}^j p_i$.

We first analyze what happens when $t = 2$. Since the ordering chooses the coins with smallest probability of heads, $p_2 \ge \frac{1}{2}$. We lower-bound the denominator sum of \Cref{eqn:adaptivitygap_arbitrary_ell} as $\sum_{j=2}^{n-1} \prod_{i=1}^j p_i$. By construction, the $j^{\text{th}}$ term in the numerator sum is at most the $j^{\text{th}}$ term in this sum. We use the rescaling trick to send $p_i \mapsto \frac{1}{2}$ for all $2 \le i < n$. Since $p_2 \ge \frac{1}{2}$, this does not increase the value of either term. In fact, the new numerator sum is $p_n p_1(\frac{1}{2} + \frac{1}{4} + \dots) = p_1 p_n  - O(\frac{1}{2^n})$, and the new denominator sum is $p_1(\frac{1}{2} + \frac{1}{4} + \dots) = p_1 - O(\frac{1}{2^n})$. So then the adaptivity gap is at most the maximum of $1$, and 
\begin{align*}
    &\frac{
    2 + 2(z^1_3(a) + z^0_3(a)) - z^1_{t+1}(a) + p_1 p_n
    }{
     2 + \pbar_{\ks} \sum_{j=1}^{n-2} \prod_{i=1}^j \pbar_{S[n-i]} +  p_1
    } + O(\frac{1}{2^n})
    \le \frac{ 2 + 2(p_1 p_n + \pbar_1 \pbar_n) 
    }{
     2 + \pbar_{\ks} \frac{\pbar_n}{p_n} +  p_1
    } + O(\frac{1}{2^n})\,.
\end{align*}
Here, we are $1$-biased at position $3$, which occurs exactly when $p_1 p_n > \pbar_1 \pbar_n$; i.e. when $p_1 > \pbar_n$. We optimize the expression above in the range $\pbar_n < p_1 \le 0.5 \le p_n \le 1$. By \Cref{claim:adaptivity_gap_t_is_2}, the expression is at most $1.2$, and so the adaptivity gap when $t = 2$ is at most $1.2  + O(\frac{1}{2^n})$.

We now assume $t \ge 3$, and split into three cases:
\begin{enumerate}
    \item When $p_{n} \le 1.2 p_{\ks}$, we lower-bound the denominator sum as $p_{\ks} \sum_{j=1}^{n-2} \prod_{i=1}^j p_i$. Then the $j^{\text{th}}$ term in the numerator sum is at most $1.2$ times the $j^{\text{th}}$ term in the denominator sum. We use the rescaling trick to send $p_i \mapsto \frac{1}{2}$ for all $t \le i < n$; since $p_t \ge \frac{1}{2}$, this does not increase the value of either term. In fact, the new numerator sum is $(\frac{1}{2} + \frac{1}{4} + \dots) p_n \prod_{i=1}^{t-1} p_i$, and the new denominator sum is $p_{\ks} \sum_{j=1}^{t-1} \prod_{i=1}^j p_i + (\frac{1}{2} + \frac{1}{4} + \dots) p_{\ks} \prod_{i=1}^{t-1} p_i$. So then the adaptivity gap is at most the maximum of $1.2$, and 
\begin{align*}
    &\frac{
    2 + 2(z_3^1(a) + z_3^0(a)) - z_{t+1}^1(a) + p_n \prod_{i=1}^{t-1} p_i
    }{
     2 + \pbar_{\ks} \sum_{j=1}^{n-2} \prod_{i=1}^j \pbar_{S[n-i]} +  p_{\ks} \sum_{j=1}^{t-1} \prod_{i=1}^j p_i + p_{\ks} \prod_{i=1}^{t-1} p_i
    } + O(\frac{1}{2^n})
    \\
&\le 
\frac{ 2 + 2(p_1 p_n + \pbar_1 \pbar_n) 
    }{
     2 + \pbar_{\ks} \frac{\pbar_n}{p_n} +  p_{\ks} \frac{p_1}{\pbar_1}
    } + O(\frac{1}{2^n})\,,
\end{align*}
where the inequality follows because $z^1_{t+1}(a) = p_n \prod_{i=1}^{t-1} p_i$, and $p_1 \le \frac{1}{2}$.

By \Cref{claim:2coin_1.2}, this expression is at most $1.2$ except when $0 \le p_1 \le 0.2$ and $0.5 \le p_n \le 0.76$, or when $0.24 \le p_1 \le 0.5$ and $0.8 \le p_n \le 1$. So except for these regions of $(p_1, p_n)$, the adaptivity gap is at most $1.2  + O(\frac{1}{2^n})$.

\item Otherwise, it is at least true that the $j^{\text{th}}$ term in the numerator sum is at most the $(j-1)^{\text{th}}$ term in the denominator sum. We first upper bound the numerator sum as $p_n p_{\ks} \sum_{j=t-1}^{n-3} \prod_{i=1}^j p_{S[i]}$. Then we use the rescaling trick to send $p_{S[i]} \mapsto \frac{1}{2}$ for all $i \ge t$ and $p_{S[i]} \to p_i$ otherwise; since $p_t \ge \frac{1}{2}$ and $p_{S[i]} \ge p_i$, this does not increase the value of either term. 
In fact, the new numerator sum is $p_n p_{\ks} (1 +\frac{1}{2} + \frac{1}{4} + \dots) \prod_{i=1}^{t-1} p_i$, and the new denominator sum is $p_{\ks} (1 + \frac{1}{2} + \frac{1}{4} + \dots) \prod_{i=1}^{t-1} p_i$. So then the adaptivity gap is at most the maximum of $1$, and 
\begin{align*}
    &\quad \ \frac{
 2 + 2(z_3^1(a) + z_3^0(a)) - z_{t+1}^1(a) + 2 p_n p_{\ks} \prod_{i=1}^{t-1} p_i
    }{
     2 + \pbar_{\ks} \sum_{j=1}^{n-2} \prod_{i=1}^j \pbar_{S[n-i]} +  p_{\ks} \sum_{j=1}^{t-2} \prod_{i=1}^j p_i + 2 p_{\ks} \prod_{i=1}^{t-1} p_i
    } + O(\frac{1}{2^n})
    \\
    &=
    \frac{
 2 + 2(z_3^1(a) + z_3^0(a)) + 2 p_n (p_{\ks}-\frac{1}{2}) \prod_{i=1}^{t-1} p_i
    }{
     2 + \pbar_{\ks} \sum_{j=1}^{n-2} \prod_{i=1}^j \pbar_{S[n-i]} +  p_{\ks} \sum_{j=1}^{t-2} \prod_{i=1}^j p_i + 2 p_{\ks} \prod_{i=1}^{t-1} p_i
    }  + O(\frac{1}{2^n})\,.
\end{align*}
When $p_{\ks} \le \frac{1}{2}$, the last numerator term is negative, and so the ratio can be upper-bounded as in the previous case (i.e. when $p_n \le 1.2 p_{\ks}$). 

\item The final case is if $\frac{1}{2} < p_{\ks} < \frac{p_n}{1.2}$. We apply the rescaling trick to the last numerator term and last denominator term. Since the numerator term is at most the denominator term, we send $2 \prod_{i=1}^{t-1} p_i \mapsto \frac{p_1^{t-1}}{\pbar_1}$; since $p_1 \le \frac{1}{2}$, this does not increase the value of either term. So the adaptivity gap is at most the maximum of $1$ and 
\begin{align*}
    &\quad \ \frac{
 2 + 2(z_3^1(a) + z_3^0(a)) + p_n (p_{\ks}-\frac{1}{2}) \frac{p_1^{t-1}}{\pbar_1}
    }{
     2 + \pbar_{\ks} \sum_{j=1}^{n-2} \prod_{i=1}^j \pbar_{S[n-i]} +  p_{\ks} \sum_{j=1}^{t-2} \prod_{i=1}^j p_i +  p_{\ks} \frac{p_1^{t-1}}{\pbar_1}  }  + O(\frac{1}{2^n})
     \\
    &\le 
      \frac{
 2 + 2(p_1 p_n + \pbar_1 \pbar_n) + p_n (p_{\ks}-\frac{1}{2}) \frac{p_1^{t-1}}{\pbar_1}
    }{
   2 + \pbar_{\ks} \frac{\pbar_n}{p_n} +  p_{\ks} \frac{p_1}{\pbar_1} }  + O(\frac{1}{2^n})\,.
\end{align*}
We optimize this expression in the region $\frac{1}{2} \le p_{\ks} \le \frac{1}{1.2} p_n$, since the other cases are handled.
By \Cref{claim:2coin_pk_midregion}, this expression is at most $1.2$ when $t \ge 3$, except in the regions specified before ($0 \le p_1 \le 0.2$ and $0.5 \le p_n \le 0.76$, or $0.24 \le p_1 \le 0.5$ and $0.8 \le p_n \le 1$). So except for these regions of $(p_1, p_n)$, the adaptivity gap is at most $1.2 + O(\frac{1}{2^n})$.\qedhere
\end{enumerate}
\end{proof}

\subsubsection{Level \texorpdfstring{$\ell=3$}{3}}
\adaptivityellequalsthree*
\begin{proof}
    We start from \Cref{eqn:adaptivitygap_arbitrary_ell}.  There are two branches: either the ordering is $0$-biased at position $3$ (i.e. when $0 \le p_1 \le 0.2$ and $0.5 \le p_n \le 0.76$), or the ordering is $1$-biased at position $3$ (i.e. $0.24 \le p_1 \le 0.5$ and $0.8 \le p_n \le 1$).
    
    We again use the rescaling trick for the summation $\sum_{j=t+2}^n z_j^1(a)$ in the numerator and $p_{\ks} \sum_{j=1}^{n-2} \prod_{i=1}^j p_{S[i]}$ in the denominator. However, the analysis depends on the choice of third coin: $n-1$ if it is $0$-biased at position $3$, and $2$ if it is $1$-biased at position $3$. In the first branch,  $\sum_{j=t+2}^n z_j^1(a) = \sum_{j=t-1}^{n-3} p_n p_{n-1} \prod_{i=1}^j p_i$, while in the second branch, $\sum_{j=t+2}^n z_j^1(a) = \sum_{j=t}^{n-2} p_n \prod_{i=1}^j p_i$ as before. 
    
    We first analyze what happens when $t = 3$. We lower-bound the denominator sum as  $\sum_{j=2}^{n-1} \prod_{i=1}^j p_i$ in the first branch and $p_{\ks} p_{S[1]} + \sum_{j=3}^{n-1} \prod_{i=1}^j p_i$ in the second branch. By construction, the $j^{\text{th}}$ term in the numerator sum is at most the $j^{\text{th}}$ term in the denominator sum. We use the rescaling trick to send $p_i \mapsto \frac{1}{2}$: in the first branch, for all $2 \le i < n - 1$, and in the second branch, for all $3 \le i < n$. The new numerator sums are  $p_n p_{n-1} p_1 (\frac{1}{2} + \frac{1}{4} + \dots) = p_n p_{n-1} p_1 - O(\frac{1}{2^n})$ and $p_n p_1 p_2 (\frac{1}{2} + \frac{1}{4} + \dots) = p_n p_1 p_2 - O(\frac{1}{2^n})$, while the new denominator sums are  $ p_1 (\frac{1}{2} + \frac{1}{4} + \dots) = p_1 - O(\frac{1}{2^n})$ and $p_{\ks} p_{S[1]} + p_1 p_2 (\frac{1}{2} + \frac{1}{4} + \dots)$. In the first branch, the adaptivity gap is at most the maximum of $1$, and
\begin{align*}
        &\quad \ \frac{
    2 + (z_3^1(a) + z_3^0(a)) + 2(z_4^1(a) + z_4^0(a)) - z_4^1(a) + p_n p_{n-1} p_1
    }{
      2 + \pbar_{\ks} \sum_{j=1}^{n-2} \prod_{i=1}^j \pbar_{S[n-i]} +  p_1
    } + O(\frac{1}{2^n})
    \\
    &\le
        \frac{
    2 + (p_1 p_n + \pbar_1 \pbar_n) + 2(p_1 p_n p_{n-1} + \pbar_1 \pbar_n \pbar_{n-1})}{
      2 + \pbar_{\ks} \frac{\pbar_{S[n-1]}}{p_{n-1}} + p_1
    } +  O(\frac{1}{2^n})\,.
\end{align*}
Similarly, in the second branch, the adaptivity gap is at most the maximum of $1$, and
\begin{align*}
    &\quad \ \frac{
    2 + (z_3^1(a) + z_3^0(a)) + 2(z_4^1(a) + z_4^0(a)) - z_4^1(a) + p_n p_1 p_2
    }{
      2 + \pbar_{\ks} \sum_{j=1}^{n-2} \prod_{i=1}^j \pbar_{S[n-i]} +  p_{\ks} p_{S[1]} + p_1 p_2
    } +  O(\frac{1}{2^n})
    \\
    &\le
        \frac{
    2 + (p_1 p_n + \pbar_1 \pbar_n) + 2(p_1 p_n p_2 + \pbar_1 \pbar_n \pbar_2)}{
      2 + \pbar_{\ks} \frac{\pbar_n}{p_n} + p_{1} p_{\max(2,\ks)} + p_1 p_2
    } +  O(\frac{1}{2^n})\,.
\end{align*}
We optimize each multivariable expression.
\begin{itemize}
    \item In the first branch, $\pbar_{\ks} \pbar_{S[n-1]} = \pbar_{n}\pbar_{\min(\ks,n-1)} \ge \pbar_n \pbar_{n-1}$. Also, since $t = 3$, we have $p_{n-1} \ge 0.5$. By our assumption of using the ordering with \emph{smallest} probability of heads, we must be $1$-biased at position $4$. By \Cref{claim:t_equals_3_branch1}, this expression is at most $1.2$ in the range $0 \le p_1 \le 0.2$ and $0.5 \le p_{n-1} \le p_n \le 0.76$, subject to this $1$-biased condition ($p_1 p_n p_{n-1} > \pbar_1 \pbar_n \pbar_{n-1}$).
    \item In the second branch, the denominator only decreases when changing $\ks$ from $1$ to $2$. Since $p_1 \ge 0.24 > \frac{0.2}{0.8} \ge \frac{\pbar_n}{p_n}$, the denominator increases when $\ks$ increases above $2$. Also, since $t = 3$, we have $p_{2} \le 0.5$. So the denominator is at least $2 + \pbar_2 \frac{\pbar_n}{p_n} + 2 p_1 p_2$. By \Cref{claim:t_equals_3_branch2}, the second expression is at most $1.2$ in the range $0.24 \le p_1 \le p_2 \le 0.5$ and $0.8 \le p_n \le 1$.
\end{itemize}
 Altogether, the adaptivity gap when $t = 3$ is at most $1.2 +  O(\frac{1}{2^n})$.

 We now assume $t \ge 4$, and split into three cases:
 \begin{enumerate}
     \item We start when $p_{\ks}$ is large. In the first branch, suppose $p_n p_{n-1} \le 1.2 p_{\ks}$; in the second branch, suppose $p_n \le 1.2 p_{\ks}$. We lower-bound the denominator sum as $p_{\ks} p_{S[1]} \sum_{j=1}^{n-2} \prod_{i=2}^j p_i$. Then the $j^{\text{th}}$ term in the numerator sum is at most $1.2$ times the $j^{\text{th}}$ term in the denominator sum. We use the rescaling trick to send $p_i \mapsto \frac{1}{2}$: in the first branch, for all $t-1 \le i < n-1$ (since $p_{t-1} \ge \frac{1}{2}$), and in the second branch, for all $t \le i < n$ (since $p_{t} \ge \frac{1}{2}$).

The new numerator sums are $(\frac{1}{2} + \frac{1}{4} + \dots) p_n p_{n-1} \prod_{i=1}^{t-2} p_i$ and $(\frac{1}{2} + \frac{1}{4} + \dots) p_n \prod_{i=1}^{t-1} p_i$; the new denominator sums are $p_{\ks} p_{S[1]} \sum_{j=1}^{t-2} \prod_{i=2}^j p_i + (\frac{1}{2} + \frac{1}{4} + \dots) p_{\ks} p_{S[1]} \prod_{i=2}^{t-2} p_i$ and $p_{\ks} p_{S[1]} \sum_{j=1}^{t-1} \prod_{i=2}^j p_i + (\frac{1}{2} + \frac{1}{4} + \dots) p_{\ks} p_{S[1]} \prod_{i=2}^{t-1} p_i$. In the first branch, the adaptivity gap is at most the maximum of $1.2$, and 
\begin{align*}
&\quad \ \frac{
     2 + \left(z_3^1(a) + z_3^0(a)\right) +  2(z_{4}^1(a) + z_{4}^0(a)) - z_{t+1}^1(a) + p_n p_{n-1} \prod_{i=1}^{t-2} p_i
    }
    {
    2 + \pbar_{\ks} \sum_{j=1}^{n-2} \prod_{i=1}^j \pbar_{S[n-i]} +  p_{\ks} p_{S[1]} \sum_{j=1}^{t-2} \prod_{i=2}^j p_i + p_{\ks} p_{S[1]} \prod_{i=2}^{t-2} p_i
    } +  O(\frac{1}{2^n})
    \\
    &\le 
    \frac{
     2 + \left(z_3^1(a) + z_3^0(a)\right) +  2(z_{4}^1(a) + z_{4}^0(a))
    }
    {
    2 + \pbar_{\ks} \frac{\pbar_{S[n-1]}}{p_{n-1}} +  p_{\ks} \frac{p_1}{\pbar_1}} +  O(\frac{1}{2^n})
    \,.
\end{align*}
In the second branch, the adaptivity gap is at most the maximum of $1.2$, and
\begin{align*}
  &\quad \ \frac{
      2 + \left(z_3^1(a) + z_3^0(a)\right) +  2(z_{4}^1(a) + z_{4}^0(a)) - z_{t+1}^1(a) + p_n \prod_{i=1}^{t-1} p_i
    }
    {
    2 + \pbar_{\ks} \sum_{j=1}^{n-2} \prod_{i=1}^j \pbar_{S[n-i]} +  p_{\ks} p_{S[1]} \sum_{j=1}^{t-1} \prod_{i=2}^j p_i + p_{\ks}  p_{S[1]}  \prod_{i=2}^{t-1} p_i
    } +  O(\frac{1}{2^n})
    \\
    &\le 
    \frac{
     2 + \left(z_3^1(a) + z_3^0(a)\right) +  2(z_{4}^1(a) + z_{4}^0(a))
    }
    {
    2 + \pbar_{\ks} \frac{\pbar_{n}}{p_{n}} +  p_{\ks} \frac{p_{S[1]}}{\pbar_2}} +  O(\frac{1}{2^n})
    \,.
\end{align*}
Notice that these expressions are equivalent under the transformation $p_i \mapsto \pbar_{n+1-i}$. So we may optimize one of them. Let's optimize the first expression in the range $0 \le p_1 \le 0.2$ and $0.5 \le p_n \le 0.76$ (and $t \ge 4$). The denominator of the expression has four variables but the numerator  has only three. To maximize this expression, we minimize the denominator over the extra variable. \begin{itemize}
    \item Suppose $1 \le \ks \le n-1$. Then the denominator is at least
    \begin{align*}
    2+\pbar_{\ks} \frac{\pbar_{n}}{p_{n-1}}
    +p_{\ks} \frac{p_{1}}{\pbar_{1}} \ge 2+\pbar_{n-1} \frac{\pbar_{n}}{p_{n-1}}
    +p_{n-1} \frac{p_{1}}{\pbar_{1}}
    \,,
    \end{align*}
    where the last inequality follows because $\frac{p_1}{\pbar_1} \le \frac{\pbar_n}{p_n}\le \frac{\pbar_n}{p_{n-1}}$ in the region $0 \le p_1 \le 0.2$, $0.5 \le p_n \le 0.76$.
    \item Suppose $\ks = n$. Then the denominator is at least
    \begin{align*}
       2+\pbar_{n} \frac{\pbar_{n-1}}{p_{n-1}}
    +p_{n} \frac{p_{1}}{\pbar_{1}}
    =
    2+\pbar_{n-1} \frac{\pbar_{n}}{p_{n-1}}
    +p_{n-1} \frac{p_{1}}{\pbar_{1}}
    \,.
    \end{align*}
\end{itemize}
So the expression is at most 
\begin{align*}
 \frac{
    2+\left(z_3^1(a) + z_3^0(a)\right) + 2(z_4^1(a) + z_4^0(a))
    }
    {
    2+\pbar_{n-1} \frac{\pbar_{n}}{p_{n-1}}
    +p_{n-1} \frac{p_{1}}{\pbar_{1}}
    } 
    =
    \frac{
    2+p_1p_n + \pbar_1 \pbar_n + 2 (p_1 p_n p_{n-1} 
    + \pbar_1 \pbar_n \pbar_{n-1})
    }
    {
    2+\pbar_{n-1} \frac{\pbar_{n}}{p_{n-1}}
    +p_{n-1} \frac{p_{1}}{\pbar_{1}}
    }
\end{align*}
By \Cref{claim:3coin_1.2}, when $0 \le p_1 \le 0.2$ and $0.65 \le p_n \le 0.76$, this expression is at most $1.2$. So the adaptivity gap here is at most $1.2 +  O(\frac{1}{2^n})$ except in the $(p_1, p_n)$ region $0 \le p_1 \le 0.2$ and $0.5 \le p_n \le 0.65$, or its complement $0.35 \le p_1 \le 0.5$ and $0.8 \le p_n \le 1$.

\item Otherwise, it is at least true that the $j^{\text{th}}$ term in the numerator sum is at most the $(j-1)^{\text{th}}$ term in the denominator sum. This works for all terms because the denominator sum starts from $1 \le t-2$. We first upper bound the numerator sum as $p_n p_{n-1} p_{\ks} \sum_{j=t-2}^{n-4} \prod_{i=1}^j p_{S[i]}$ in the first branch and $p_n p_{\ks} \sum_{j=t-1}^{n-3} \prod_{i=1}^j p_{S[i]}$ in the second branch. Then we use the rescaling trick in the first branch (second branch) to send $p_{S[i]} \mapsto \frac{1}{2}$ for all $i \ge t-1$ ($i \ge t$) and $p_{S[i]} \to p_i$ for $i < t-1$ ($2 \le i < t$); since $p_{t-1} \ge \frac{1}{2}$ ($p_t \ge \frac{1}{2}$) and $p_{S[i]} \ge p_i$, this does not increase the value of either term. 

The new numerator sums are $p_n p_{n-1} p_{\ks} (1 +\frac{1}{2} + \frac{1}{4} + \dots) \prod_{i=1}^{t-2} p_i$ and $p_n p_{\ks} p_{S[1]} (1 +\frac{1}{2} + \frac{1}{4} + \dots) \prod_{i=2}^{t-1} p_i$; the new denominator sums are $p_{\ks} \sum_{j=1}^{t-3} \prod_{i=1}^j p_i + p_{\ks} (1 + \frac{1}{2} + \frac{1}{4} + \dots) \prod_{i=1}^{t-2} p_i$ and $p_{\ks} p_{S[1]} \sum_{j=1}^{t-2} \prod_{i=2}^j p_i + p_{\ks} p_{S[1]}  (1 + \frac{1}{2} + \frac{1}{4} + \dots) \prod_{i=2}^{t-1} p_i$. In the first branch, the adaptivity gap is at most the maximum of $1$, and 
\begin{align*}
&\quad \ \frac{
     2 + \left(z_3^1(a) + z_3^0(a)\right) +  2(z_{4}^1(a) + z_{4}^0(a)) - z_{t+1}^1(a) + 2 p_n p_{n-1} p_{\ks} \prod_{i=1}^{t-2} p_i
    }
    {
    2 + \pbar_{\ks} \sum_{j=1}^{n-2} \prod_{i=1}^j \pbar_{S[n-i]} +  p_{\ks} \sum_{j=1}^{t-3} \prod_{i=1}^j p_i + 2 p_{\ks} \prod_{i=1}^{t-2} p_i
    } +  O(\frac{1}{2^n})
    \\
    &\le 
    \frac{
     2 + \left(z_3^1(a) + z_3^0(a)\right) +  2(z_{4}^1(a) + z_{4}^0(a)) + 2 p_n p_{n-1} (p_{\ks}  - \frac{1}{2}) \prod_{i=1}^{t-2} p_i
    }
    {
    2 + \pbar_{\ks} \frac{\pbar_{S[n-1]}}{p_{n-1}}
    + p_{\ks} \sum_{j=1}^{t-3} \prod_{i=1}^j p_i + 2 p_{\ks} \prod_{i=1}^{t-2} p_i
   } +  O(\frac{1}{2^n})
    \,.
\end{align*}
Suppose the previous case does not hold; i.e. $p_n p_{n-1} > 1.2 p_{\ks}$. Then $p_{\ks} < \frac{1}{1.2}p_n p_{n-1} \le \frac{0.76^2}{1.2} < 0.49$. Then the last numerator term is negative, and so the ratio can be upper-bounded as in the case when $p_n p_{n-1} \le 1.2 p_{\ks}$.

In the second branch, the adaptivity gap is at most the maximum of $1$, and 
\begin{align*}
&\quad \ \frac{
     2 + \left(z_3^1(a) + z_3^0(a)\right) +  2(z_{4}^1(a) + z_{4}^0(a)) - z_{t+1}^1(a) + 2 p_n p_{\ks} p_{S[1]} \prod_{i=2}^{t-1} p_i
    }
    {
    2 + \pbar_{\ks} \sum_{j=1}^{n-2} \prod_{i=1}^j \pbar_{S[n-i]} +  p_{\ks} p_{S[1]} \sum_{j=1}^{t-2} \prod_{i=2}^j p_i + 2 p_{\ks}  p_{S[1]} \prod_{i=2}^{t-1} p_i
    } +  O(\frac{1}{2^n})
    \\
    &\le 
    \frac{
     2 + \left(z_3^1(a) + z_3^0(a)\right) +  2(z_{4}^1(a) + z_{4}^0(a)) + 2 p_{n} (p_{\ks}  - \frac{1}{2}) p_{S[1]} \prod_{i=2}^{t-1} p_i
    }
    {
    2 + \pbar_{\ks} \frac{\pbar_{n}}{p_{n}}
    + p_{\ks} p_{S[1]} \sum_{j=1}^{t-2} \prod_{i=2}^j p_i + 2 p_{\ks} p_{S[1]} \prod_{i=2}^{t-1} p_i
   } +  O(\frac{1}{2^n})
    \,.
\end{align*}
When $p_{\ks} \le \frac{1}{2}$, the last numerator term is negative, and so again the ratio can be upper-bounded as in the previous case (i.e. when $p_n \le 1.2 p_{\ks}$).

\item The case $\frac{1}{2} \le p_{\ks} \le \frac{p_n}{1.2}$ only occurs in the second branch. Here, we apply the rescaling trick to the last numerator term and last denominator term. Since the numerator term is at most the denominator term, we send $2 \prod_{i=2}^{t-1} p_i \mapsto \frac{p_2^{t-2}}{\pbar_2}$; since $p_2 \le \frac{1}{2}$ in this branch, this does not increase the value of either term. So the adaptivity gap is at most the maximum of $1$ and 
\begin{align*}
&\quad \ 
    \frac{
     2 + \left(z_3^1(a) + z_3^0(a)\right) +  2(z_{4}^1(a) + z_{4}^0(a)) + p_n p_{S[1]} (p_{\ks}  - \frac{1}{2}) \frac{p_2^{t-2}}{\pbar_2}
    }
    {
    2 + \pbar_{\ks} \frac{\pbar_{n}}{p_{n}} 
    + p_{\ks} p_{S[1]} \sum_{j=1}^{t-2} \prod_{i=2}^j p_i + p_{\ks}p_{S[1]}  \frac{p_2^{t-2}}{\pbar_2}
    } +  O(\frac{1}{2^n})
    \\
    &\le \frac{
     2 + \left(p_1 p_n + \pbar_1 \pbar_n \right) +  2(p_1 p_n p_{2} + \pbar_1 \pbar_n \pbar_{2}) + p_n p_{S[1]} (p_{\ks}  - \frac{1}{2}) \frac{p_2^{t-2}}{\pbar_2}
    }
    {
    2 + \pbar_{\ks} \frac{\pbar_{n}}{p_{n}} 
    + p_{\ks} \frac{p_{S[1]}}{\pbar_2} 
    } +  O(\frac{1}{2^n})
    \,.
\end{align*}
In the region $\frac{1}{2} < p_{\ks} \le \frac{p_n}{1.2}$, it must be that $\ks \ne 1$ (since $p_1 \le \frac{1}{2}$), so $p_{S[1]} = p_1$. By \Cref{claim:3coin_pk_midregion}, the multivariable expression is at most $1.2$ for $t \ge 4$ in the range $0.24 \le p_1 \le p_2 \le 0.5$ and $0.8 \le p_n \le 1$ and $\frac{1}{2} < p_{\ks} \le \frac{p_n}{1.2}$. So the adaptivity gap is at most $1.2 +  O(\frac{1}{2^n})$ here.\qedhere
 \end{enumerate}
\end{proof}

\subsubsection{Level \texorpdfstring{$\ell=4$}{4}}

\adaptivityellequalsfour*
\begin{proof}
We start from \Cref{eqn:adaptivitygap_arbitrary_ell}.
There are two ranges to handle: $0 \le p_1 \le 0.2$ and $0.5 \le p_n \le 0.65$, or $0.35 \le p_1 \le 0.5$ and $0.8 \le p_n \le 1$. We use the rescaling trick for the summation $\sum_{j=t+2}^n z_j^1(a)$ in the numerator and $p_{\ks} \sum_{j=1}^{n-2} \prod_{i=1}^j p_{S[i]}$ in the denominator. However, the analysis depends on the choice of third and fourth coin. There are only two branches to consider:
\begin{itemize}
    \item If $0 \le p_1 \le 0.2$ and $0.5 \le p_n \le 0.65$, the ordering is $0$-biased at position $4$, since 
    \begin{align*}
         \frac{z^1_4(a)}{z^0_4(a)} \le  \frac{p_1 p_n p_{n-1}}{\pbar_1 \pbar_n \pbar_{n-1}} \le \frac{0.2 \cdot 0.65^2}{0.8 \cdot 0.35^2} 
    =
    \frac{0.65^2}{0.7^2} < 1\,.
    \end{align*}
    So, the first four coins are $(1, n,n-1,n-2)$, and numerator sum $\sum_{j=t+2}^n z_j^1(a)$ which equals $\sum_{j=t-2}^{n-4} p_n p_{n-1} p_{n-2} \sum_{i=1}^j p_i$.
    \item By symmetry, if $0.35 \le p_1 \le 0.5$ and $0.8 \le p_n \le 1$, the ordering is $1$-biased at position $4$, and so the first four coins are $(1, n, 2,3)$. The numerator sum in this branch is $\sum_{j=t+2}^n z_j^1(a) = \sum_{j=t}^{n-2} p_n \sum_{i=1}^j p_i$.
\end{itemize}

We analyze when $t = 4$, when $p_{\ks}$ is large enough, and when $p_{\ks} \le \frac{1}{2}$.
Each branch has a symmetric analysis. We do change the denominator lower bound depending on the branch:
\begin{align*}
    \pbar_{\ks} \sum_{j=1}^{n-2} \prod_{i=1}^j \pbar_{S[n-i]} 
\ge \pbar_{\ks} \pbar_{S[n-1]} + \pbar_{S[n-1]}\frac{\pbar_{S[n-2]}}{p_{n-2}} 
\ge \pbar_{\ks} \frac{\pbar_{n}}{p_{n}}
\end{align*}
\begin{enumerate}
    \item We first analyze what happens when $t = 4$. 
\begin{itemize}
    \item In the first branch, the numerator sum is $\sum_{j=t-2}^{n-4} p_n p_{n-1} p_{n-2} \sum_{i=1}^j p_i$. We lower-bound the denominator sum as $\sum_{j=2}^{n-1} \prod_{i=1}^j p_i$. The $j^{\text{th}}$ term of the numerator sum is at most the $j^{\text{th}}$ term of the denominator sum. We use the rescaling trick to send all $p_i \mapsto \frac{1}{2}$ for all $2 \le i < n-2$. The new sums are (up to $O(\frac{1}{2^n})$ factors) $p_n p_{n-1} p_{n-2} p_1$ and $p_1$. The new numerator sum is exactly $z_5^1(a)$, and cancels out. So the adaptivity gap is at most the maximum of $1$, and  
$$
\frac{
2 + \sum_{j=3}^4 \left( z_j^1(a) + z_j^0(a) \right) + 2(z_5^1(a) + z_5^0(a))
}{
2 + \pbar_{\ks} \pbar_{S[n-1]} + \pbar_{\ks} \pbar_{S[n-1]}\frac{\pbar_{S[n-2]}}{p_{n-2}} + p_1
} + O(\frac{1}{2^n})\,.
$$
\item In the other branch, the numerator sum is $\sum_{j=t}^{n-2} p_n \sum_{i=1}^j p_i$. We lower-bound the denominator sum as $p_{\ks} p_{S[1]} + p_{\ks} p_{S[1]} p_{S[2]} + \sum_{j=4}^{n-1} \prod_{i=1}^j p_i$. The $j^{\text{th}}$ term of the numerator sum is at most the $j^{\text{th}}$ term of the denominator sum. We use the rescaling trick to send all $p_i \mapsto \frac{1}{2}$ for all $4 \le i < n$. The new sums are (up to $O(\frac{1}{2^n})$ factors) $p_n p_1 p_2 p_3$ and $p_{\ks} p_{S[1]} + p_{\ks} p_{S[1]} p_{S[2]} + p_1 p_2 p_3$. The new numerator sum is exactly $z_5^1(a)$, and cancels out. So the adaptivity gap is at most the maximum of $1$, and 
$$
\frac{
2 + \sum_{j=3}^4 \left( z_j^1(a) + z_j^0(a) \right) + 2(z_5^1(a) + z_5^0(a))
}{
2 + \pbar_{\ks} \frac{\pbar_{n}}{p_{n}} + p_1 p_{\max(2,\ks)} + p_1 p_2 p_{\max(3,\ks)} + p_1 p_2 p_3
} + O(\frac{1}{2^n})\,.
$$
\end{itemize}
We optimize each multivariable expression.
\begin{itemize}
    \item For the first branch, the term $\pbar_{\ks}\pbar_{S[n-1]} = \pbar_n \pbar_{\min(\ks,n-1)}$ and $\pbar_{\ks}\pbar_{S[n-1]}\pbar_{S[n-2]} = \pbar_n\pbar_{n-1} \pbar_{\min(\ks,n-2)}$.  As a result, the choice of $\ks$ that minimizes the denominator is $\ks = n-1$. So the adaptivity gap is at most
\begin{align*}
    \frac{
    2 + p_1 p_n(1 + p_{n-1}) + \pbar_1 \pbar_n(1 + \pbar_{n-1}) + 2(p_1 p_n p_{n-1} p_{n-2} + \pbar_1 \pbar_n \pbar_{n-1} \pbar_{n-2})
    }{
    2 + \pbar_n \pbar_{n-1} + \pbar_n \pbar_{n-1} \frac{\pbar_{n-2}}{p_{n-2}} + p_1
    }\,.
\end{align*}
By our choice of ordering algorithm (\emph{smallest} probability of heads), we are $1$-biased at position $5$. This means that $p_1 p_n p_{n-1} p_{n-2} > \pbar_1 \pbar_n \pbar_{n-1} \pbar_{n-2}$. As a result, this means $\frac{p_1}{\pbar_1} > \frac{\pbar_n \pbar_{n-1} \pbar_{n-2}}{p_n p_{n-1} p_{n-2}} > \frac{0.35^3}{0.65^3} > 0.15$, and so $p_1 > \frac{0.15}{1.15} > 0.13$. So we can restrict to the parameter range $p_1 \in [0.13, 0.2]$. In this branch, since $t \ge 4$, we have $p_{n-2} \ge 0.5$.
By \Cref{claim:adaptivitygap_t_is_4_branch1}, the multivariable expression is at most $1.2$ when  $p_1 \in [0.13,0.2]$ and $0.5 \le p_{n-2} \le p_{n-1} \le p_n \le 0.65$. 
\item For the second branch, since $p_1 \ge 0.35 > \frac{0.2}{0.8} \ge \frac{\pbar_n}{p_n}$, the denominator increases when increasing $\ks$ above $2$. So the adaptivity gap is at its largest when $\ks = 2$. In this branch, since $t \ge 4$, we have $p_{3} \le 0.5$. By \Cref{claim:adaptivitygap_t_is_4_branch2}, the multivariable expression is at most $1.2$ when $p_n \in [0.8,1]$ and $0.35 \le p_1 \le p_2 \le p_3 \le 0.5$.
\end{itemize}
Altogether, the adaptivity gap when $t = 4$ is at most $1.2 + O(\frac{1}{2^n})$.

\item Otherwise, we assume $t \ge 5$. We consider when $p_{\ks}$ is large enough. 
This means $p_n p_{n-1} \le 1.2 p_{\ks}$ (in the first branch) or $p_n \le 1.2 p_{\ks}$ (in the second branch). Actually, this will not leave out any instances, since in the first branch, if $p_{\ks} > \frac{1}{2}$, then $p_n p_{n-1} \le 0.65^2 \le 1.2 \cdot 0.5 \le 1.2 p_{\ks}$.

\begin{itemize}
    \item In the first branch,  the numerator sum is $\sum_{j=t-2}^{n-4} p_n p_{n-1} p_{n-2} \prod_{i=1}^j p_i$. We lower-bound the denominator sum as $p_{\ks} \sum_{j=1}^{n-2} \prod_{i=1}^j p_i$.  The $j^{\text{th}}$ term in the numerator sum is at most $1.2$ times the $j^{\text{th}}$ term in the denominator sum. We use the rescaling trick to send $p_i \mapsto \frac{1}{2}$ for all $2 \le i < n-2$. The new sums are (up to $O(\frac{1}{2^n})$ factors) $p_n p_{n-1} p_{n-2} \prod_{i=1}^{t-3} p_i$ and $p_{\ks} \sum_{j=1}^{t-3} \prod_{i=1}^j p_i + p_{\ks} \prod_{i=1}^{t-3} p_i$. The new numerator sum is exactly $z_{t+1}^1(a)$, and cancels out. And the denominator sum is at least $p_{\ks} \frac{p_1}{\pbar_1}$ since $p_1 \le \frac{1}{2}$. So the adaptivity gap is at most the maximum of $1.2$, and  
\begin{align*}
\frac{
2 + \sum_{j=3}^4 \left( z_j^1(a) + z_j^0(a) \right) + 2(z_5^1(a) + z_5^0(a))
}{
2 + \pbar_{\ks} \pbar_{S[n-1]} + \pbar_{\ks} \pbar_{S[n-1]}\frac{\pbar_{S[n-2]}}{p_{n-2}} + p_{\ks} \frac{p_1}{\pbar_1}
} + O(\frac{1}{2^n})\,.
\end{align*}
\item In the other branch, the numerator sum is $\sum_{j=t}^{n-2} p_n \prod_{i=1}^j p_i$. We lower-bound the denominator sum as $p_{\ks} p_{S[1]} + p_{\ks} p_{S[1]} p_{S[2]} \sum_{j=2}^{n-2} \prod_{i=3}^j p_i$.  The $j^{\text{th}}$ term in the numerator sum is at most $1.2$ times the $j^{\text{th}}$ term in the denominator sum. We use the rescaling trick to send $p_i \mapsto \frac{1}{2}$ for all $4 \le i < n$. The new sums are (up to $O(\frac{1}{2^n})$ factors) $p_n \prod_{i=1}^{t-1} p_i$ and $p_{\ks} p_{S[1]} + p_{\ks} p_{S[1]} p_{S[2]} \sum_{j=3}^{t-1} \prod_{i=3}^j p_i + p_{\ks} p_{S[1]} p_{S[2]} \prod_{i=3}^{t-1} p_i$. The new numerator sum is exactly $z_{t+1}^1(a)$, and cancels out. And the denominator sum is at least $p_{\ks} p_{S[1]} + p_{\ks} p_{S[1]} \frac{p_{S[2]}}{p_3}$ since $p_3 \le \frac{1}{2}$ in this branch. So the adaptivity gap is at most the maximum of $1.2$, and 
\begin{align*}
\frac{
2 + \sum_{j=3}^4 \left( z_j^1(a) + z_j^0(a) \right) + 2(z_5^1(a) + z_5^0(a))
}{
2 + \pbar_{\ks} \frac{\pbar_n}{p_n} + p_{\ks} p_{S[1]} +  p_{\ks} p_{S[1]}  \frac{p_{S[2]}}{\pbar_3}
} + O(\frac{1}{2^n})\,.
\end{align*}
\end{itemize}
Note that these expressions are equivalent under the transformation $p_i \mapsto \pbar_{n+1-i}$. So we may optimize one of them. Let's optimize the first expression in the range $0 \le p_1 \le 0.2$ and $0.5 \le p_n \le 0.65$ (and $t \ge 5$).
The denominator has five variables but the numerator has only four. To maximize this expression, we minimize the denominator over the extra variable.
\begin{itemize}
        \item Suppose $1 \le \ks \le n-2$. Then the denominator is at least 
    \begin{align*}
        &2 + \pbar_{\ks} \pbar_n ( 1 + \frac{\pbar_{n-1}}{p_{n-2}}) + p_{\ks}\frac{p_1}{\pbar_1}
        \\
        &\ge
        2 + \pbar_{\ks} \pbar_n ( 1 + \frac{\pbar_{n-1}}{p_{n-2}}) + p_{\ks}\frac{p_1}{\pbar_1}
        \\
        &\ge
           2 + \pbar_{n-2} \pbar_n ( 1 + \frac{\pbar_{n-1}}{p_{n-2}}) + p_{n-2}\frac{p_1}{\pbar_1}     
        \,,
    \end{align*}
    where the last inequality follows because $\frac{p_1}{\pbar_1} \le \pbar_n$ in the region $0 \le p_1 \le 0.2$, $0.5 \le p_n \le 0.65$.
    \item Suppose $n-1 \le \ks \le n$. Then the denominator is at least
    \begin{align*}
        2 + \pbar_{\ks} \pbar_{S[n-1]} ( 1 + \frac{\pbar_{n-2}}{p_{n-2}}) + p_{\ks}\frac{p_1}{\pbar_1} \ge 2 + \pbar_n \pbar_{n-1} ( 1 + \frac{\pbar_{n-2}}{p_{n-2}}) + p_{n-1}\frac{p_1}{\pbar_1}\,.
    \end{align*}
    This bound is smaller than in the previous case, since $\frac{p_1}{\pbar_1} \le \pbar_n$ in the region $0 \le p_1 \le 0.2$, $0.5 \le p_n \le 0.65$.
\end{itemize}
So, the denominator is always at least $2 + \pbar_n \pbar_{n-1} ( 1 + \frac{\pbar_{n-2}}{p_{n-2}}) + p_{n-1}\frac{p_1}{\pbar_1} = 2 + \pbar_{n-1} \frac{\pbar_n}{p_{n-2}} + p_{n-1} \frac{p_1}{\pbar_1}$. This means the expression is at most
\begin{align*}
    \frac
    {
    2
    + p_1 p_n(1 + p_{n-1}) + \pbar_1 \pbar_n(1 + \pbar_{n-1}) + 2(p_1 p_n p_{n-1}p_{n-2} + \pbar_1 \pbar_n \pbar_{n-1}\pbar_{n-2})
    }
    {
    2 + \pbar_{n-1} \frac{ \pbar_{n}}{p_{n-2}} + p_{n-1}\frac{p_1}{\pbar_1}
    }\,.
\end{align*}
Since $t \ge 5$, we know that $p_{n-1} \ge p_{n-2} \ge \frac{1}{2}$. By \Cref{claim:4coin_1.2}, when $0 \le p_1 \le 0.2$ and $\frac{1}{2} \le p_{n-2} \le p_{n-1} \le 0.65$ and $0.58 \le p_n \le 0.65$, this is at most $1.2$. So the adaptivity gap is at most $1.2 + O(\frac{1}{2^n})$, except in the $(p_1, p_n)$ region $0 \le p_1 \le 0.2$ and $0.5 \le p_n \le 0.58$, or its complement $0.42 \le p_1 \le 0.5$ and $0.8 \le p_n \le 1$.

\item Otherwise, we know that $p_{\ks} \le \frac{1}{2}$. For each branch, the $j^{\text{th}}$ term in the numerator sum is at most the $(j-1)^{\text{th}}$ term in the denominator sum.
\begin{itemize}
    \item In the first branch, we first upper-bound the numerator sum as $p_n p_{n-1} p_{n-2} p_{\ks} \sum_{j=t-3}^{n-5} p_{S[i]}$. Then we use the rescaling trick to send $p_{S[i]} \mapsto \frac{1}{2}$ for all $i \ge t - 2$ and $p_{S[i]}$ for all $i < t- 2$. The new sums are (up to $O(\frac{1}{2^n})$ factors) $2 p_n p_{n-1} p_{n-2} p_{\ks}  \prod_{i=1}^{t-3} p_i$ in the numerator and $p_{\ks} \sum_{j=1}^{t-4} \prod_{i=1}^j p_i + 2 p_{\ks} \prod_{i=1}^{t-3} p_i$ in the denominator. Since $p_{\ks} \le \frac{1}{2}$, the new numerator sum is at most $z_{t+1}^1(a)$, and cancels out. And the denominator sum is at least $p_{\ks} \frac{p_1}{\pbar_1}$ since $p_1 \le \frac{1}{2}$. So the adaptivity gap is at most the maximum of $1$, and
\begin{align*}
\frac{
2 + \sum_{j=3}^4 \left( z_j^1(a) + z_j^0(a) \right) + 2(z_5^1(a) + z_5^0(a))
}{
2 + \pbar_{\ks} \pbar_{S[n-1]} + \pbar_{\ks} \pbar_{S[n-1]}\frac{\pbar_{S[n-2]}}{p_{n-2}} + p_{\ks} \frac{p_1}{\pbar_1}
} + O(\frac{1}{2^n})\,.
\end{align*}
\item In the second branch, we first upper-bound the numerator sum as $p_n p_{\ks} \sum_{j=t-1}^{n-3} p_{S[i]}$. Then we use the rescaling trick to send $p_{S[i]} \mapsto \frac{1}{2}$ for all $i \ge t$ and $p_{S[i]}$ for all $3 \le i < t$.  The new sums are (up to $O(\frac{1}{2^n})$ factors) $2 p_n p_{\ks} p_{S[1]} p_{S[2]} \prod_{i=3}^{t-1} p_i$ in the numerator and $p_{\ks} p_{S[1]} + p_{\ks} p_{S[1]} p_{S[2]} \sum_{j=2}^{t-2} \prod_{i=3}^j p_i + 2 p_{\ks} p_{S[1]} p_{S[2]} \prod_{i=3}^{t-1} p_i$ in the denominator. Since $p_{\ks} \le \frac{1}{2}$, the new numerator sum is at most $z_{t+1}^1(a)$, and cancels out. And the denominator sum is at least $p_{\ks} p_{S[1]} + p_{\ks} p_{S[1]} \frac{p_{S[2]}}{\pbar_3}$ since $p_3 \le \frac{1}{2}$. So the adaptivity gap is at most the maximum of $1$, and
\begin{align*}
\frac{
2 + \sum_{j=3}^4 \left( z_j^1(a) + z_j^0(a) \right) + 2(z_5^1(a) + z_5^0(a))
}{
2 + \pbar_{\ks} \frac{\pbar_n}{p_n} + p_{\ks} p_{S[1]} +  p_{\ks} p_{S[1]}  \frac{p_{S[2]}}{\pbar_3}
} + O(\frac{1}{2^n})\,.
\end{align*}
\end{itemize}
These are exactly the same expressions as in the previous case (i.e. when $p_{\ks}$ is large enough). So the adaptivity gap is at most $1.2 + O(\frac{1}{2^n})$ except in the same regions specified in the previous case.\qedhere
\end{enumerate}
\end{proof}

\subsubsection{Level \texorpdfstring{$\ell=5$}{5}}
\adaptivityellequalsfive*

\begin{proof}
    We start from \Cref{eqn:adaptivitygap_arbitrary_ell}. There are two ranges to handle:  $0 \le p_1 \le 0.2$ and $0.5 \le p_n \le 0.58$, or $0.42 \le p_1 \le 0.5$ and $0.8 \le p_n \le 1$.

We use the rescaling trick for the summation $\sum_{j=t+2}^n z_j^1(a)$ in the numerator and $p_{\ks} \sum_{j=1}^{n-2} \prod_{i=1}^j p_{S[i]}$ in the denominator. The analysis depends on the choice of third, fourth \emph{and} fifth coin. But there are again only two branches to consider:
\begin{itemize}
    \item If $0 \le p_1 \le 0.2$ and $0.5 \le p_n \le 0.58$, the ordering is $0$-biased at position $5$, since 
    \begin{align*}
         \frac{z^1_5(a)}{z^0_5(a)} \le     
         \frac{p_1 p_n p_{n-1} p_{n-2}}{\pbar_1 \pbar_n \pbar_{n-1} \pbar_{n-2}} \le \frac{0.2 \cdot 0.58^3}{0.8 \cdot 0.42^3}
    <
    \frac{0.58^3}{(1.5\cdot 0.42)^3}
    =
    \frac{0.58^3}{0.63^3} < 1\,.
    \end{align*}
    So, the first five coins are $(1, n,n-1,n-2,n-3)$. For this branch, the numerator sum $\sum_{j=t+2}^n z_j^1(a)$ is equal to $\sum_{j=t-3}^{n-5} p_n p_{n-1} p_{n-2} p_{n-3} \sum_{i=1}^j p_i$.
    \item By symmetry, if $0.42 \le p_1 \le 0.5$ and $0.8 \le p_n \le 1$, the ordering is $1$-biased at position $4$, and so the first five coins are $(1, n, 2,3,4)$. The numerator sum in this branch is $\sum_{j=t+2}^n z_j^1(a) = \sum_{j=t}^{n-2} p_n \sum_{i=1}^j p_i$.
\end{itemize}
Recall that the remaining cases only have $p_{\ks}$ large enough or  $p_{\ks} \le \frac{1}{2}$. Each branch has a symmetric analysis, which is very similar to the previous levels. We do change the denominator lower bound depending on the branch:
\begin{align*}
    \pbar_{\ks} \sum_{j=1}^{n-2} \prod_{i=1}^j \pbar_{S[n-i]} 
\ge \pbar_{\ks} \pbar_{S[n-1]}(1 + \pbar_{S[n-2]}) + \pbar_{\ks} \pbar_{S[n-1]} \pbar_{S[n-2]} \frac{\pbar_{S[n-3]}}{p_{n-3}}
\ge \pbar_{\ks} \frac{\pbar_{n}}{p_{n}}\,.
\end{align*}
\begin{enumerate}
    \item We first consider when $p_{\ks}$ is large enough. As in the proof of \Cref{lemma:adaptivity_ell_4}, this means $p_n p_{n-1} \le 1.2 p_{\ks}$ (in the first branch) or $p_n \le 1.2 p_{\ks}$ (in the second branch). Actually, this will not leave out any instances, since in the first branch, if $p_{\ks} > \frac{1}{2}$, then $p_n p_{n-1} \le 0.65^2 \le 1.2 \cdot 0.5 \le 1.2 p_{\ks}$.
\begin{itemize}
    \item In the first branch,  the numerator sum is $\sum_{j=t-3}^{n-5} p_n p_{n-1} p_{n-2} p_{n-3} \prod_{i=1}^j p_i$. We lower-bound the denominator sum as $p_{\ks} \sum_{j=1}^{n-2} \prod_{i=1}^j p_i$.  The $j^{\text{th}}$ term in the numerator sum is at most $1.2$ times the $j^{\text{th}}$ term in the denominator sum. We use the rescaling trick to send $p_i \mapsto \frac{1}{2}$ for all $2 \le i < n-3$. The new sums are (up to $O(\frac{1}{2^n})$ factors) $p_n p_{n-1} p_{n-2} p_{n-3} \prod_{i=1}^{t-4} p_i$ and $p_{\ks} \sum_{j=1}^{t-4} \prod_{i=1}^j p_i + p_{\ks} \prod_{i=1}^{t-4} p_i$. The new numerator sum is exactly $z_{t+1}^1(a)$, and cancels out. And the denominator sum is at least $p_{\ks} \frac{p_1}{\pbar_1}$ since $p_1 \le \frac{1}{2}$. So the adaptivity gap is at most the maximum of $1.2$, and  
\begin{align*}
\frac{
2 + \sum_{j=3}^5 \left( z_j^1(a) + z_j^0(a) \right) + 2(z_6^1(a) + z_6^0(a))
}{
2 + \pbar_{\ks} \pbar_{S[n-1]}(1 + \pbar_{S[n-2]}) + \pbar_{\ks} \pbar_{S[n-1]} \pbar_{S[n-2]} \frac{\pbar_{S[n-3]}}{p_{n-3}} + p_{\ks} \frac{p_1}{\pbar_1}
} + O(\frac{1}{2^n})\,.
\end{align*}
\item In the other branch, the numerator sum is $\sum_{j=t}^{n-2} p_n \prod_{i=1}^j p_i$. We lower-bound the denominator sum as $p_{\ks} p_{S[1]} (1 + p_{S[2]}) + p_{\ks} p_{S[1]} p_{S[2]} p_{S[3]} \sum_{j=2}^{n-2} \prod_{i=4}^j p_i$.  The $j^{\text{th}}$ term in the numerator sum is at most $1.2$ times the $j^{\text{th}}$ term in the denominator sum. We use the rescaling trick to send $p_i \mapsto \frac{1}{2}$ for all $5 \le i < n$. The new sums are (up to $O(\frac{1}{2^n})$ factors) $p_n \prod_{i=1}^{t-1} p_i$ and $p_{\ks} p_{S[1]}(1 + p_{S[2]}) + p_{\ks} p_{S[1]} p_{S[2]} p_{S[3]} \sum_{j=4}^{t-1} \prod_{i=3}^j p_i + p_{\ks} p_{S[1]} p_{S[2]} p_{S[3]} \prod_{i=4}^{t-1} p_i$. The new numerator sum is exactly $z_{t+1}^1(a)$, and cancels out. And the denominator sum is at least $p_{\ks} p_{S[1]}(1 + p_{S[2]}) + p_{\ks} p_{S[1]} p_{S[2]} \frac{p_{S[3]}}{p_4}$ since $p_4 \le \frac{1}{2}$ in this branch. So the adaptivity gap is at most the maximum of $1.2$, and 
\begin{align*}
\frac{
2 + \sum_{j=3}^5 \left( z_j^1(a) + z_j^0(a) \right) + 2(z_6^1(a) + z_6^0(a))
}{
2 + \pbar_{\ks} \frac{\pbar_n}{p_n} + p_{\ks} p_{S[1]}(1 + p_{S[2]}) +   p_{\ks} p_{S[1]} p_{S[2]}\frac{p_{S[3]}}{\pbar_4}
} + O(\frac{1}{2^n})\,.
\end{align*}
\end{itemize}
Note that these expressions are equivalent under the transformation $p_i \mapsto \pbar_{n+1-i}$. So we may optimize one of them. Let's optimize the first expression in the range $0 \le p_1 \le 0.2$ and $0.5 \le p_n \le 0.58$ (and $t \ge 5$). Here, the denominator has six variables but the numerator has only five. To maximize this expression, we minimize the denominator over the extra variable.
\begin{itemize}
    \item Suppose $1 \le \ks \le n-3$. Then the denominator is at least 
 \begin{align*}
    2
    +
    \pbar_{\ks}
    \pbar_n
    (
    1+\pbar_{n-1}
    (
    1+\frac{\pbar_{n-2}}{p_{n-3}}
    ))
    +
    p_{\ks}
    \frac{p_1}{\pbar_1}
    \ge
       2
    +
    \pbar_{n-3}
    \pbar_n
    (
    1+\pbar_{n-1}
    (
    1+\frac{\pbar_{n-2}}{p_{n-3}}
    ))
    +
    p_{n-3}
    \frac{p_1}{\pbar_1}
    \end{align*}
    where the inequality follows because $\frac{p_1}{\pbar_1}\le \pbar_n$ in the region $0 \le p_1 \le 0.2$, $0.5 \le p_n \le 0.58$.
    \item Suppose $\ks = n-2$. Then the denominator is at least
     \begin{align*}
    2
    +
    \pbar_{n-2}
    \pbar_n
    (
    1+\pbar_{n-1}
    (
    1+\frac{\pbar_{n-3}}{p_{n-3}}
    ))
    +
    p_{n-2}
    \frac{p_1}{\pbar_1}
    \ge
        2
    +
    \pbar_{n-2}
    \pbar_n
    (
    1+\pbar_{n-1}
    (
    1+\frac{\pbar_{n-3}}{p_{n-3}}
    ))
    +
    p_{n-2}
    \frac{p_1}{\pbar_1}
    \end{align*}
    This is smaller than previous case, since $\frac{p_1}{\pbar_1} \le \pbar_n$ in region $0 \le p_1 \le 0.2$, $0.5 \le p_n \le 0.58$.
    \item Suppose $n-1 \le \ks \le n$. Then the denominator is at least
\begin{align*}
    2
    +
    \pbar_{\ks}
    \pbar_{S[n-1]}
    (
    1+\pbar_{n-2}
    (
    1+\frac{\pbar_{n-3}}{p_{n-3}}
    ))
    +
    p_{\ks}
    \frac{p_1}{\pbar_1}
    \ge
        2
    +
    \pbar_{n-1}
    \pbar_n
    (
    1+\pbar_{n-2}
    (
    1+\frac{\pbar_{n-3}}{p_{n-3}}
    ))
    +
    p_{n-1}
    \frac{p_1}{\pbar_1}
    \end{align*}
        This is smaller than previous case, since $\frac{p_1}{\pbar_1} \le \pbar_n$ in the region $0 \le p_1 \le 0.2$, $0.5 \le p_n \le 0.58$.
\end{itemize}
So, the denominator is at least $2
    +
    \pbar_{n-1}
    \pbar_n
    (
    1+\pbar_{n-2}
    (
    1+\frac{\pbar_{n-3}}{p_{n-3}}
    ))
    +
    p_{n-1}\frac{p_1}{\pbar_1}
    =
    2 + \pbar_{n-1}\pbar_n (1 + \frac{\pbar_{n-2}}{p_{n-3}}) + p_{n-1} \frac{p_1}{\pbar_1}$.
So the expression is at most
\begin{align*}
    \frac{
    2
    +p_1 p_n(1 + p_{n-1}(1 + p_{n-2}(1+2p_{n-3})))
    +\pbar_1 \pbar_n(1 + \pbar_{n-1}(1 + \pbar_{n-2}(1+2\pbar_{n-3})))
    }
    {2
    +
    \pbar_n
    \pbar_{n-1}
    (
    1+\frac{\pbar_{n-2}}{p_{n-3}})
    +
    p_{n-1}
    \frac{p_1}{\pbar_1}
    }\,.
\end{align*}
Since $t \ge 5$, we have $p_{n-3} \ge 0.5$. We may also assume $p_n > 0.5$, since otherwise all coins have probability of heads at most $\frac{1}{2}$ and the adaptivity gap is at most $1.2$ in this case. So by \Cref{claim:5coin_1.2}, when $0 \le p_1 \le 0.2$ and $0.5 \le p_{n-3} \le p_{n-2} \le p_{n-1} \le p_n \le 0.58$ and $p_n > 0.5$, the expression is at most $1.2$, and so the adaptivity gap is at most $1.2 + O(\frac{1}{2^n})$ here.

\item Otherwise, we know that $p_{\ks} \le \frac{1}{2}$. For each branch, the $j^{\text{th}}$ term in the numerator sum is at most the $(j-1)^{\text{th}}$ term in the denominator sum.
\begin{itemize}
  \item In the first branch, we first upper-bound the numerator sum as $p_n p_{n-1} p_{n-2} p_{n-3} p_{\ks} \sum_{j=t-4}^{n-6} p_{S[i]}$. Then we use the rescaling trick to send $p_{S[i]} \mapsto \frac{1}{2}$ for all $i \ge t - 3$ and $p_{S[i]}$ for all $i < t- 3$. The new sums are (up to $O(\frac{1}{2^n})$ factors) $2 p_n p_{n-1} p_{n-2} p_{n-3} p_{\ks}  \prod_{i=1}^{t-4} p_i$ in the numerator and $p_{\ks} \sum_{j=1}^{t-5} \prod_{i=1}^j p_i + 2 p_{\ks} \prod_{i=1}^{t-4} p_i$ in the denominator. Since $p_{\ks} \le \frac{1}{2}$, the new numerator sum is at most $z_{t+1}^1(a)$, and cancels out. And the denominator sum is at least $p_{\ks} \frac{p_1}{\pbar_1}$ since $p_1 \le \frac{1}{2}$. So the adaptivity gap is at most the maximum of $1$, and
\begin{align*}
\frac{
2 + \sum_{j=3}^5 \left( z_j^1(a) + z_j^0(a) \right) + 2(z_6^1(a) + z_6^0(a))
}{
2 + \pbar_{\ks} \pbar_{S[n-1]}(1 + \pbar_{S[n-2]}) + \pbar_{\ks} \pbar_{S[n-1]} \pbar_{S[n-2]} \frac{\pbar_{S[n-3]}}{p_{n-3}} + p_{\ks} \frac{p_1}{\pbar_1}
} + O(\frac{1}{2^n})\,.
\end{align*}
\item In the second branch, we first upper-bound the numerator sum as $p_n p_{\ks} \sum_{j=t-1}^{n-3} p_{S[i]}$. Then we use the rescaling trick to send $p_{S[i]} \mapsto \frac{1}{2}$ for all $i \ge t$ and $p_{S[i]}$ for all $4 \le i < t$.  The new sums are (up to $O(\frac{1}{2^n})$ factors) $2 p_n p_{\ks} p_{S[1]} p_{S[2]} p_{S[3]} \prod_{i=4}^{t-1} p_i$ in the numerator and $p_{\ks} p_{S[1]}(1 + p_{S[2]}) + p_{\ks} p_{S[1]} p_{S[2]} p_{S[3]} \sum_{j=3}^{t-2} \prod_{i=4}^j p_i + 2 p_{\ks} p_{S[1]} p_{S[2]} p_{S[3]} \prod_{i=4}^{t-1} p_i$ in the denominator. Since $p_{\ks} \le \frac{1}{2}$, the new numerator sum is at most $z_{t+1}^1(a)$, and cancels out. And the denominator sum is at least $p_{\ks} p_{S[1]}(1 + p_{S[2]}) + p_{\ks} p_{S[1]} p_{S[2]} \frac{p_{S[3]}}{\pbar_4}$ since $p_4 \le \frac{1}{2}$ in this branch. So the adaptivity gap is at most the maximum of $1$, and
\begin{align*}
\frac{
2 + \sum_{j=3}^5 \left( z_j^1(a) + z_j^0(a) \right) + 2(z_6^1(a) + z_6^0(a))
}{
2 + \pbar_{\ks} \frac{\pbar_n}{p_n} + p_{\ks} p_{S[1]} (1 + p_{S[2]}) +  p_{\ks} p_{S[1]}  p_{S[2]} \frac{p_{S[3]}}{\pbar_4}
} + O(\frac{1}{2^n})\,.
\end{align*}
\end{itemize}
These are exactly the same expressions as in the previous case (i.e. when $p_{\ks}$ is large enough), and so by \Cref{claim:5coin_1.2}, the adaptivity gap is at most $1.2 + O(\frac{1}{2^n})$ here as well.\qedhere
\end{enumerate}
\end{proof}


\subsection{Deferred math facts}
\label{sec:deferredmathfacts}

\subsubsection{Level \texorpdfstring{$\ell=2$}{2}}

\begin{claim}
\label{claim:adaptivity_gap_t_is_2}
    Consider variables $v,u$ such that $1-u \le v \le 0.5 \le u \le 1$. Then the function
    \begin{align*}
        f(v,u) \defeq \frac{2 + 2\left(uv + (1-u)(1-v)\right)}{2 + \frac{(1-u)^2}{u} + v}
    \end{align*}
    has maximum value equal to $1.2$.
\end{claim}
\begin{proof}
    Any expression $\frac{ax+b}{cx+d}$ is piecewise monotonic in $(-\infty, -\frac{d}{c})$ and $(-\frac{d}{c}, \infty)$. So $f$ is monotonic in $v$ for all values in $(0, \infty)$. As a result, $f$ is maximized at an boundary point of $v$ (either $v = 1-u$ or $v = 0.5$).

    When $v = 0.5$, the function simplifies to 
    \begin{align*}
        f(0.5, u) = \frac{3}{2.5 + \frac{(1-u)^2}{u}}\,.
    \end{align*}
    This expression is maximized at the largest value of $u$; i.e. $u = 1$. So $f(0.5, u) \le f(0.5, 1) = \frac{3}{2.5} = 1.2$. 

    When $v = 1-u$, the function simplifies to 
    \begin{align*}
        f(1-u,u) = \frac{2 + 4u(1-u)}{2 + \frac{(1-u)^2}{u} + 1-u}
        =
        \frac{2 + 4u(1-u)}{\frac{1 + u}{u}}\,.
    \end{align*}
    We want to see when $f(1-u,u) \stackrel{?}\le c$ for $c \defeq 1.2$. This is true if 
    \begin{align*}
        2u + 4u^2(1-u) \stackrel{?}\le c(1 + u)\,.
    \end{align*}
    Rearranging terms, this is true exactly when
\begin{align*}
4 u^3 - 4u^2 + (c-2)u + c \stackrel{?}\ge 0\,.
\end{align*}
The extremal points of this polynomial occur when the derivative is equal to $0$; i.e. when $12u^{2}-8u+c-2 = 0$. This occurs when $c = 1.2$ at $u \approx -0.08$ and $u \approx 0.75$. Since the cubic coefficient is positive, it is decreasing only between these points. So if the inequality is true at the location of the second extremal point, it is true for all $0 \le u \le 1$. By inspection, the cubic evaluated at its second extremal point is at least $0.37375 > 0$. So the inequality is true.

Thus, for all $1-u \le v \le 0.5 \le u \le 1$, we have $f(v,u) \le 1.2$.
\end{proof}

\begin{restatable}{claim}{twocoinonepointtwo}
\label{claim:2coin_1.2}
    Consider variables $v,x,u$ such that $0 \le v \le 0.5 \le u \le 1$, $v \le x \le u$, and $v + u \le 1$. Then the function
 \begin{align*}
        f(v,x,u) \defeq \frac{2 + 2\left(uv + (1-u)(1-v)\right)}{2 + \frac{xuv + (1-x)(1-u)(1-v)}{u(1-v)}}
    \end{align*}
    is at most $1.2$, unless $0 \le v \le 0.2$ and $0.5 \le u \le 0.76$.
\end{restatable}
\begin{proof}
    Since $f \ge 0$, the maximum of $f$ is the minimum of $\frac{1}{f}$. Since $\frac{1}{f}$ is linear in $x$, it is minimized at one of the boundary points $x = u$ or $x = v$.

    We assume $x = u$; at the end of the proof we make comments about what happens when $x = v$. When $x = u$, 
\begin{align*}
    f(v,u,u) = \frac{2 + 2\left(uv + (1-u)(1-v)\right)}{2 + \frac{u^2 v + (1-u)^2 (1-v)}{u(1-v)}} = \frac{2u (1-v) (2 - u -v + 2uv)}{1 - v + u^2}\,.
\end{align*}
We would like to prove that $f(v,u,u) \stackrel{?}\le c$ for $c \defeq 1.2$. This is true if 
\begin{align*}
        2u(1-v)(2-u-v+2uv) \stackrel{?}\le c(1-v+u^2)\,.
    \end{align*}
    Rearranging terms by powers of $u$, we get 
    \begin{align*}
        u^2(c - 2(1-v)(2v-1)) + u(2(1-v)(v-2)) + c(1-v) \stackrel{?}\ge 0\,.
    \end{align*}
    The coefficient on $u^2$ is positive for all $0 \le v \le 0.5$. Then the inequality is satisfied if the quadratic has no real roots. This happens when the discriminant is negative, i.e.
    \begin{align*}
        (2(1-v)(v-2))^2 - 4(c - 2(1-v)(2v-1))c(1-v)   \stackrel{?}< 0\,.
    \end{align*}
    Since $v < 1$, we can divide out by one factor of $(1-v)$:
    \begin{align*}
        4(1-v)(v-2)^2 - 4c(c - 2(1-v)(2v-1)) \stackrel{?}< 0\,.
    \end{align*}
     Since $c = 1.2$,  this can be exactly factored as
    \begin{align*}
    0 &\stackrel{?}> -4v^3 + (20-16c)v^2 + (24c-32)v +(-4c^2-8c+16)
    \\
    &=-4v^3 + 0.8 v^2 -3.2v +0.64 
        \\
        &= -4(v-0.2)(v^2+0.8)\,.
    \end{align*}
    So when $v > 0.2$, the function is at most $1.2$.

Alternatively, rearranging terms by powers of $v$, we get 
    \begin{align*}
        v^2(4u^2-2u) + v(-6u^2+6u-c) + ((2+c)u^2-4u+c) \stackrel{?}\ge 0\,.
    \end{align*}
    Since $u \ge 0.5$ by assumption, the coefficient on $v^2$ is positive. Then the inequality is satisfied if the quadratic has no real roots; i.e. when the discriminant is negative:
    \begin{align*}
        (-6u^2+6u-c)^2 - 4(4u^2-2u)((2+c)u^2-4u+c) \stackrel{?}< 0\,.
    \end{align*}
    Since $c = 1.2$, this can be written exactly as 
    \begin{align*}
        0 &\stackrel{?}> - 15.2 u^4  +17.6 u^3 - 0.8u^2 - 4.8u + 1.44
        \\
        &= -15.2 (u-0.5)^4 - 12.8(u-0.5)^3 + 2.8 (u-0.5)^2 + 0.09
        \,.
    \end{align*}
    By Descartes' rule of signs, the second equation has one real positive root where $u > 0.5$, which occurs below $0.76$ (since the value at $u = 0.76$ is negative, but the value at $u =0.5$ is positive). Since the leading coefficient is negative, the function is at most $1.2$ whenever $u > 0.76$.

We finally remark about the $x = v$ boundary case. Note that $f(v,x,u) = f(1-u,1-x,1-v)$, which flips the boundary case. Thus, when $x = v$, the function is at most $1.2$ except when $0 \le 1-u \le 0.2$ and $0.5 \le 1 - v \le 0.76$. But in this region, $u+v \ge 0.8 + 0.24 = 1.16$ which does not satisfy $u+v \le 1$.
\end{proof}

\begin{claim}
\label{claim:2coin_pk_midregion}
    Consider variables $v,x,u$ such that $0 \le v \le 0.5 \le u \le 1$, $\frac{1}{2} \le x \le \frac{u}{1.2}$. Then the function
 \begin{align*}
        f(v,x,u) \defeq \frac{2 + 2\left(uv + (1-u)(1-v)\right) + u(x-\frac{1}{2})\frac{v^2}{1-v}}{2 + \frac{xuv + (1-x)(1-u)(1-v)}{u(1-v)}}
    \end{align*}
    is at most $1.2$, unless $0 \le v \le 0.2$ and $0.5 \le u \le 0.76$, or $0.24 \le v \le 0.5$ and $0.8 \le u \le 1$.
\end{claim}
\begin{proof}
Both the numerator and denominator are linear in $x$. Since any expression
$\frac{ax+b}{cx+d}$ is affine to $\frac{1}{cx+d}$, this is a piecewise monotone function in $x$ on either side of the pole $x = \frac{-d}{c}$.
Since the denominator is positive for all $0 \le x \le 1$, the function is monotone in the region we are interested in. So it is maximized at a boundary point of $x$, either $x = \frac{1}{2}$ or $x = \frac{u}{1.2}$. If $x = \frac{1}{2}$, the last term in the numerator cancels out and so the claim is true by \Cref{claim:2coin_1.2}.

We assume $x = \frac{u}{1.2} = a u$ where $a \defeq \frac{1}{1.2} = \frac{5}{6}$. Then
\begin{align*}
    f(v,au,u) &= \frac{2 + 2\left(uv + (1-u)(1-v)\right) + u(au-0.5)\frac{v^2}{1-v}}{2 + \frac{au^2 v + (1-u)(1-au) (1-v)}{u(1-v)}} 
    \\
    &= \frac{4 u - 2 u^2 - 6 u v + 6 u^2 v + 2 u v^2 - 4.5 u^2 v^2 + a u^3 v^2}{a u^2 + a u v - a u - u v + u - v + 1}\,.
\end{align*}
We would like to prove that $f(v,u,u) \stackrel{?}\le c$ for $c \defeq 1.2$. This is true if 
\begin{align*}
    4 u - 2 u^2 - 6 u v + 6 u^2 v + 2 u v^2 - 4.5 u^2 v^2 + a u^3 v^2 \stackrel{?}\le c(a u^2 + a u v - a u - u v + u - v + 1)\,.
\end{align*}
Rearranging terms by powers of $v$, we get 
\begin{align*}
    v^2 (-au^3 + 4.5 u^2 - 2u) 
    + v (cau - cu - c - 6u^2 + 6u)
    + (cau^2 - cau + cu + c + 2u^2 - 4u)
    \stackrel{?}\ge 0\,.
\end{align*}
This is true when the minimum value achieved by the quadratic in $v$ is at least $0$; i.e.
\begin{align*}
    (cau^2 - cau + cu + c + 2u^2 - 4u)  \stackrel{?}\ge \frac{(cau - cu - c - 6u^2 + 6u)^2}{4(-au^3 + 4.5 u^2 - 2u)}\,.
\end{align*}
In other words, when
\begin{align*}
    4(cau^2 - cau + cu + c + 2u^2 - 4u)(-au^3 + 4.5 u^2 - 2u) -   (cau - cu - c - 6u^2 + 6u)^2 \stackrel{?}\ge 0\,.
\end{align*}
By inspection, at $a = \frac{5}{6}$ and $c = 1.2$, this univariate polynomial is at least $0$ in the range $u \in [\frac{2}{3},1]$. 

In the range $u \in [\frac{1}{2}, \frac{2}{3}]$, the minimum value of the quadratic in $v$ occurs at $v \le 0$, since the $v^2$ coefficient is positive and the $v$ coefficient is non-negative. So we may evaluate the quadratic at $v = 0.1$; if it non-negative here, then it is non-negative for all $v \ge 0.1$. This would be true when
\begin{align*}
     0.01 (-au^3 + 4.5 u^2 - 2u) 
    + 0.1 (cau - cu - c - 6u^2 + 6u)
    + (cau^2 - cau + cu + c + 2u^2 - 4u)
    \stackrel{?}\ge 0\,.
\end{align*}
By inspection, at $a = \frac{5}{6}$ and $c = 1.2$, this univariate polynomial is at least $0$ for all values of $u$. So then $f(v,au,u) \le 1.2$ except in the region $0 \le v \le 0.1$ and $\frac{1}{2} \le u \le \frac{2}{3}$.
\end{proof}

\subsubsection{Level \texorpdfstring{$\ell=3$}{3}}

\begin{claim}
\label{claim:t_equals_3_branch1}
    Consider variables $v,x,u$, where $0 \le v \le 0.2$, and $0.5 \le x \le u \le 0.76$, and $vxu > (1-v)(1-x)(1-u)$. Then the function
    \begin{align*}
        f(v,x,u) \defeq \frac{
2 + uv(1+2x) + (1-u)(1-v)(1+2(1-x))
        }
        {
        2+(1-x)\frac{1-u}{x}+v
        }\,
    \end{align*}
    is at most $1.2$.
\end{claim}
\begin{proof}
Both the numerator and denominator are linear in $u$. Since any expression
$\frac{au+b}{cu+d}$ is affine to $\frac{1}{cu+d}$, this is a piecewise monotone function in $u$ on either side of the pole $u = \frac{-d}{c}$.
Since the denominator is positive for all $0 \le u \le 1$, the function is monotone in the region we are interested in. So it is maximized at a boundary point of $u$, either $u = 0.76$ or $\frac{u}{1-u} = \frac{(1-v)(1-x)}{vx}$.

Let's start by assuming $u = 0.76$. Then the remaining expression has numerator and denominator linear in $v$. Since the denominator is positive for all $0 \le v \le 1$, the function is monotone in $v$, and so maximized at a boundary point in $v$; i.e. $v = 0$ or $v = 0.2$. We try both. 

We know that $f(v,x,u) \stackrel{?}\le c$ for $c \defeq 1.2$ if
\begin{align*}
   2x + uvx(1 + 2x) + x(1-u)(1-v)(3-2x) \stackrel{?} \le c ( 2x + (1-x)(1-u) + vx )\,.
\end{align*}
So, $f(0,x,0.76) \stackrel{?}\le c$ for $c \defeq 1.2$ if
\begin{align*}
     2x + 0.24x(3-2x) \stackrel{?} \le c ( 2x + 0.24(1-x) )\,.
\end{align*}
This is equivalent to $-0.288 + 0.608 x - 0.48 x^2 \stackrel{?} \le 0$, which is true by inspection.

Similarly, $f(0.2,x,0.76) \stackrel{?}\le c$ for $c \defeq 1.2$ if
\begin{align*}
     2x + 0.2 \cdot 0.76x(1 + 2x)  + 0.8\cdot 0.24x(3-2x) \stackrel{?} \le c ( 2x + 0.24(1-x) + 0.2x)\,.
\end{align*}
This is equivalent to $-0.288 + 0.376 x - 0.08 x^2 \stackrel{?} \le 0$. By inspection, this quadratic has roots at $x \approx 0.9$ and $x \approx 3.7$; since the quadratic coefficient is negative, the inequality holds for all $x \in [0.5, 0.76]$.

Now we assume  $\frac{u}{1-u} = \frac{(1-v)(1-x)}{vx}$. This can be rewritten as $u = \frac{(1-v)(1-x)}{(1-v)(1-x)+vx}$. At this value of $u$, $f(v,x,u) \le c$ if 
\begin{align*}
    \frac{x}{2 v x - v - x + 1} \cdot (-2 c v^2 x + c v^2 - 2 c v x + 2 c x - 2 c + 4 v^2 x^2 - 4 v^2 x - v^2 - 4 v x^2 + 8 v x - v - 2 x + 2)\stackrel{?}\le 0\,.
\end{align*}
By inspection, for $0 \le v \le 0.2$ and $0.5 \le x \le 0.76$, the fraction is positive. So we must decide if the expression in the parenthetical is at most $0$. It can be regrouped as
\begin{align*}
    x^2 \cdot (4v^2 - 4v)
+    x \cdot (-2cv^2 - 2cv + 2c - 4v^2 + 8v - 2)
+ cv^2 - 2c - v^2 - v + 2
\stackrel{?}\le 0\,.
\end{align*}
When $v = 0$, this is a linear function in $x$; i.e. $(2c-2)x - 2c + 2 \stackrel{?} \le 0$, which is true for all $0 \le x \le 1$. Otherwise, this is a quadratic in $x$ with negative quadratic coefficient. Consider the derivative of the expression at $x = 0.76$:
\begin{align*}
 2 \cdot 0.76 \cdot (4v^2 - 4v) +
 (-2cv^2 - 2cv + 2c - 4v^2 + 8v - 2)
\end{align*}
By inspection, at $c = 1.2$, this is positive for all $0 \le v \le 0.2$. So the maximum value of the quadratic in $x$ (if it exists) occurs at $x \ge 0.76$. So, if the equation is true at $x = 0.76$, then it is true for all $x \in [0.5, 0.76]$. This occurs when
\begin{align*}
       0.76^2 \cdot (4v^2 - 4v)
+    0.76 \cdot (-2cv^2 - 2cv + 2c - 4v^2 + 8v - 2)
+ cv^2 - 2c - v^2 - v + 2
\stackrel{?}\le 0\,.
\end{align*}
By inspection, at $c = 1.2$, this is true for all $v$.
\end{proof}

\begin{claim}
\label{claim:t_equals_3_branch2}
    Consider variables $v,x,u$, where $0.24 \le v \le x \le 0.5$, and $0.8 \le u \le 1$. Then the function
    \begin{align*}
        f(v,x,u) \defeq \frac{
2 + uv(1+2x) + (1-u)(1-v)(1+2(1-x))
        }
        {
        2+(1-x)\frac{1-u}{u}+2xv
        }\,
    \end{align*}
    is at most $1.2$.
\end{claim}
\begin{proof}
We have that $f(v,x,u) \stackrel{?}\le c$ exactly when
\begin{align*}
    2u + u^2 v(1 + 2x) + u(1-u)(1-v)(3-2x) \stackrel{?}\le 
    2cu + c(1-x)(1-u) + 2cuxv\,.
\end{align*}
Rearranging terms by $u$:
\begin{align}
\label{eqn:t_3_branch2_quadratic_in_u}
    u^2 (v(1+2x) - (1-v)(3-2x))
    + u(2 + (1-v)(3-2x) - 2c - 2cxv + c(1-x))
    - c(1-x)
    \stackrel{?} \le 0\,.
\end{align}
Both the numerator and denominator of $f$ are linear in $v$. Since any expression
$\frac{av+b}{cv+d}$ is affine to $\frac{1}{cv+d}$, this is a piecewise monotone function in $v$ on either side of the pole $v = \frac{-d}{c}$.
Since the denominator is positive for all $0 \le v \le 1$, the function is monotone in the region we are interested in. So it is maximized at a boundary point of $u$, either $v = 0.24$ or $v = x$.

Let's first assume $v = 0.24$. Then the numerator and denominator are linear in $x$. For the same reason, the function is maximized at a boundary point of $x$, either $x = 0.24$ or $x = 0.5$. The $x = 0.24$ case will be covered by the $v = x$ case, so we consider $x = 0.5$. We have $f(0.24, 0.5, u) \stackrel{?} \le c$ exactly when $- 1.04 u^2 + (3.52 - 1.74c) u - 0.5c \stackrel{?} \le 0$.
By inspection, this holds for all $u$ when $c = 1.2$.

Now we assume $v= x$.
When $x = 0.5$, \cref{eqn:t_3_branch2_quadratic_in_u} is linear in $u$; i.e. $ u(3-2c) - 0.5c \stackrel{?}\le 0$. At $c = 1.2$, this holds for all $u \le 1$. Otherwise, \cref{eqn:t_3_branch2_quadratic_in_u} is a quadratic in $u$ with negative coefficient. It has maximum value
\begin{align*}
    -c(1-x) - \frac{(2 + (1-x)(3-2x) - 2c - 2cx^2 + c(1-x))^2}{4(x(1+2x) - (1-x)(3-2x))}\,.
\end{align*}
By inspection, this is non-negative at $c =1.2$ for all $0.24 \le x \le 0.43$. So \cref{eqn:t_3_branch2_quadratic_in_u} holds in this region. 

Consider the derivative \cref{eqn:t_3_branch2_quadratic_in_u} with respect to $u$ at $u = 1$:
\begin{align*}
    2(x(1+2x) - (1-x)(3-2x)) + (2 + (1-x)(3-2x) - 2c - 2cx^2 + c(1-x))\,.
\end{align*}
By inspection, this is positive for $0.43 \le x \le 0.5$. So the maximum of \cref{eqn:t_3_branch2_quadratic_in_u} occurs at value $u \ge 1$. Thus, if \cref{eqn:t_3_branch2_quadratic_in_u} holds at $u = 1$, it holds for all $0.8 \le u \le 1$. Plugging in $u = 1$, we get
\begin{align*}
    (x(1+2x) - (1-x)(3-2x)) + (2 + (1-x)(3-2x) - 2c - 2cx^2 + c(1-x)) - c(1-x) \stackrel{?}\le 0\,.
\end{align*}
By inspection, this is non-negative for all $x \le 0.5$. 
\end{proof}

\begin{restatable}{claim}{threecoinonepointtwo}
\label{claim:3coin_1.2}
    Consider variables $v,x,u$, where $0 \le v \le 0.2$, and $v \le x \le u$, and $0.65 \le u \le 0.76$. Then the function
    \begin{align*}
        f(v,x,u) \defeq \frac{
2 + uv(1+2x) + (1-u)(1-v)(1+2(1-x))
        }
        {
        2+(1-x)\frac{1-u}{x}+x\frac{v}{1-v}
        }\,
    \end{align*}
    is at most $1.2$.
\end{restatable}
\begin{proof}
    The function is at most $c \defeq 1.2$ exactly when
\begin{align*}
c\left(2\bar{v}x
+
\bar{v}\bar{u}(1-x)
+x^2v
\right)
\stackrel{?}\ge
\bar{v}x
\left(
2
+
uv(1+2x)
+
\bar{u}\bar{v}(3-2x)
\right)
\,.
\end{align*}
We factor terms by $x$: 
\begin{align*}
x^2
\left(
cv-2uv\bar{v}+2\bar{u}\bar{v}^2
\right)
+
x
\left(
2c\bar{v}-2\bar{v}-c\bar{v}\bar{u}-uv\bar{v}-3\bar{u}\bar{v}^2
\right)
+
c\bar{v}\bar{u}
\stackrel{?}\ge 0\,.
\end{align*}
We argue that the quadratic coefficient $q \defeq cv-2uv\bar{v}+2\bar{u}\bar{v}^2$ is positive. Since $0 \le v \le 0.2$, we have $v\bar{v} \le \bar{v}^2$; so $q$ is minimized at the largest value of $u$, i.e. $u = 0.76$. So $q \ge 1.2v-1.52v\bar{v}+0.48\bar{v}^2 = 0.48 - 1.28 v + 2 v^2$. This takes value at least\footnote{The quadratic $ax^2 + bx + c$ with $a > 0$ takes minimum value at $x = -\frac{b}{2a}$ and has minimum value $c - \frac{b^2}{4a}$.} $0.48 - \frac{(-1.28)^2}{4 \cdot 2} = 0.2752 > 0$. So $q > 0$.

As a result, the function is at most $1.2$ when the minimum value of the quadratic is at least $0$. This occurs when
\begin{align*}
    c\bar{v}\bar{u} - \frac{\left(
2c\bar{v}-2\bar{v}-c\bar{v}\bar{u}-uv\bar{v}-3\bar{u}\bar{v}^2
\right)^2}{4\left(
cv-2uv\bar{v}+2\bar{u}\bar{v}^2
\right)}
\stackrel{?}\ge 0\,.
\end{align*}
Since the quadratic coefficient $q$ is positive, this occurs exactly when
\begin{align*}
    c\bar{v}\bar{u} \cdot 4\left(
cv-2uv\bar{v}+2\bar{u}\bar{v}^2
\right) \stackrel{?}\ge \left(
2c\bar{v}-2\bar{v}-c\bar{v}\bar{u}-uv\bar{v}-3\bar{u}\bar{v}^2
\right)^2 \,.
\end{align*}
Since $v < 1$, we can divide out one factor of $\bar{v}$:
\begin{align*}
    c\bar{u} \cdot 4\left(
cv-2uv\bar{v}+2\bar{u}\bar{v}^2
\right) \stackrel{?}\ge \bar{v}\left(
2c-2-c\bar{u}-uv-3\bar{u}\bar{v}
\right)^2 \,.
\end{align*}
We factor terms by $u$ and set $c = 1.2$:
\begin{align}
\label{eqn:3coin_uv}
    u^2 (16 v^3 - 49.6 v^2 + 41.64 v - 8.04) 
+ u (-24 v^3 + 70 v^2 - 64.48 v + 12.72) 
+ (9 v^3 - 22.2 v^2 + 23.8 v - 4.84)
\stackrel{?}\ge 0\,.
\end{align}
The quadratic coefficient in \cref{eqn:3coin_uv} has roots at $v \approx 0.27$, $v = 1$, and $v \approx 1.83$, so it is negative for all  $v \le 0.2$. So the quadratic in \cref{eqn:3coin_uv} takes maximum value at 
\begin{align*}
    u = \frac{-1}{2} \cdot \frac{-24 v^3 + 70 v^2 - 64.48 v + 12.72}{16 v^3 - 49.6 v^2 + 41.64 v - 8.04}\,.
\end{align*}
We prove that this location is at least $u \ge 0.775$. Since the denominator (quadratic coefficient in \cref{eqn:3coin_uv}) is negative for all $v \le 0.2$, we have $u \ge 0.775$ when 
\begin{align*}
    -24 v^3 + 70 v^2 - 64.48 v + 12.72 \stackrel{?}\ge 0.775 (-2)(16 v^3 - 49.6 v^2 + 41.64 v - 8.04)\,.
\end{align*}
Rearranging terms, this occurs when $0.258 + 0.062 v - 6.88 v^2 + 0.8 v^3 \ge 0$. The roots of this cubic are $v \approx -0.187216$, $v \approx 0.200616$, and $v \approx 8.5$; so this is true for all $0 \le v \le 0.2$. 

Since \cref{eqn:3coin_uv} is a quadratic with negative quadratic coefficient and with maximum value at $u \ge 0.775$, it is increasing for all $u \le 0.775$. Thus, if \cref{eqn:3coin_uv} is positive at $u = 0.65$, then it is satisfied for all $0.65 \le u \le 0.76$ and so the claim is proved. We set $u = 0.65$:
\begin{align*}
    0.0311 - 0.5191 v + 2.344 v^2 + 0.16 v^3 \stackrel{?}\ge 0\,.
\end{align*}
This has only one root at $v \approx -15$, and thus is satisfied for all $v \ge 0$.
\end{proof}

\begin{claim}
\label{claim:3coin_pk_midregion}
    Consider variables $w,v,x,u$, where $0.24 \le v  \le x \le 0.5$, and $0.8 \le u \le 1$, and $\frac{1}{2} \le w \le \frac{u}{1.2}$. Then the function
    \begin{align*}
        f(w;v,x,u) \defeq \frac{
2 + uv(1+2x) + (1-u)(1-v)(1+2(1-x)) + uv(w - \frac{1}{2}) \frac{x^2}{1-x}
        }
        {
        2+(1-w)\frac{1-u}{u}+w\frac{v}{1-x}
        }\,
    \end{align*}
    is at most $1.2$.
\end{claim}
\begin{proof}
Both the numerator and denominator of $f$ are linear in $w$. Since any expression
$\frac{aw+b}{cw+d}$ is affine to $\frac{1}{cw+d}$, this is a piecewise monotone function in $w$ on either side of the pole $w = \frac{-d}{c}$.
Since the denominator is positive for all $0 \le w \le 1$, the function is monotone in the region we are interested in. So it is maximized at a boundary point of $w$, either $w = \frac{1}{2}$ or $w = \frac{u}{1.2}$. At each boundary point of $w$, the same is true of $v$; so it is maximized at a boundary point (either $v = 0.24$ or $v = x$).

We have $f(w;v,x,u) \le c$ exactly when
\begin{align*}
    u \bar{x}(2 + uv(1 + 2x) + \bar{u}\bar{v}(1+2\bar{x})) + u^2v(w - \frac{1}{2})x^2 \stackrel{?} \le
    2cu\bar{x} + c\bar{w}\bar{u}\bar{x} + cwvu\,.
\end{align*}
By the above, we only need to try $w = \frac{1}{2}$ or $w = \frac{u}{1.2}$, and $v = 0.24$ or $v = x$.
\begin{itemize}
    \item When $v = 0.24$ and $w = \frac{1}{2}$, the equation reads
\begin{align*}
    u(1-x)(2 + 0.24u(1 + 2x) + 0.76(1-u)(1 + 2(1-x))) \stackrel{?} \le 2cu (1-x) + 0.5c(1-u)(1-x)+0.12cu\,.
\end{align*}
Regrouping terms, we get
\begin{align*}
x^2(1.52 u - 2 u^2)
+ x(0.5 c - 5.8 u + 1.5 c u + 4.04 u^2)
    -0.5 c + 4.28 u - 1.62 c u - 2.04 u^2  \stackrel{?}\le 0\,.
\end{align*}
The quadratic coefficient of $x$ is negative for $u \le 0.76$. So the maximum value of this expression has value
\begin{align*}
    -0.5 c + 4.28 u - 1.62 c u - 2.04 u^2 - \frac{(0.5 c - 5.8 u + 1.5 c u + 4.04 u^2)^2}{4 (1.52 u - 2 u^2)}\,.
\end{align*}
By inspection, this value is negative at $c = 1.2$ for all $u \ge 0.8$. So the inequality holds in this case.
\item When $v = 0.24$ and $w = \frac{u}{1.2}$, the equation reads
\begin{align*}
    u\bar{x}(2 + 0.24u(1 + 2x) + 0.76\bar{u}(1 + 2\bar{x})) 
    + 0.76u^2(\frac{u}{1.2} - \frac{1}{2})x^2 \stackrel{?} \le 2cu \bar{x} + c(1-\frac{u}{1.2})\bar{u}\bar{x}+u^2\frac{c}{1.2}\,.
\end{align*}
Regrouping terms, we get
\begin{align*}
x^2(1.52 u - 2.38 u^2 + \frac{19}{30} u^3)
+ x(c- 5.8 u + \frac{1}{6} c u + 4.04 u^2 + \frac{5}{6} c u^2)
-c + 4.28 u - \frac{1}{6} c u - 2.04 u^2 - \frac{5}{3} c u^2 
 \stackrel{?} \le 0\,.
\end{align*}
The quadratic coefficient of $x$ is equal to $0$ at $u^\star \defeq \frac{2.38-\sqrt{2.38^2 - 4 \cdot 1.52 \cdot \frac{19}{30}}}{2 \cdot \frac{19}{30}} \approx 0.8157$. By inspection at $c = 1.2$, the inequality holds here for all $x \ge 0$.

Consider the derivative at the boundary points $x = 0.5$ and $x = 0.24$:
\begin{align*}
&2(0.5)(1.52 u - 2.38 u^2 + \frac{19}{30} u^3) + (c- 5.8 u + \frac{1}{6} c u + 4.04 u^2 + \frac{5}{6} c u^2) 
\\
&2(0.24)(1.52 u - 2.38 u^2 + \frac{19}{30} u^3) + (c- 5.8 u + \frac{1}{6} c u + 4.04 u^2 + \frac{5}{6} c u^2) 
\end{align*}
By inspection at $c = 1.2$, the first is negative for all $u \in [0.8, 0.8249]$ and positive for all $u \in [0.825, 1]$. The second is negative for all $[0.8, 0.8224]$ and positive for all $u \in [0.8226, 1]$.

When both derivatives have the same sign, the expression is monotonic in $x$. So we can evaluate the boundary points $x = 0.24$ and $x = 0.5$. Plugging in each point $x = 0.24$ and $x = 0.5$, the expression holds by inspection at $c = 1.2$ for all $u$.

Otherwise, $u \in [0.8224, 0.825]$. The quadratic coefficient of $x$ is negative, so we calculate the maximum value:
\begin{align*}
    -c + 4.28 u - \frac{1}{6} c u - 2.04 u^2 - \frac{5}{3} c u^2  - \frac{(c- 5.8 u + \frac{1}{6} c u + 4.04 u^2 + \frac{5}{6} c u^2)^2}{4(1.52 u - 2.38 u^2 + \frac{19}{30} u^3)}
\end{align*}
By inspection, this is negative for $u \in [0.82, 0.83]$ and so the expression holds in this region.
\item When $v = x$ and $w = 0.5$, the equation reads
\begin{align*}
    u(1-x)(2 + xu(1 + 2x) + (1-x)(1-u)(1 + 2(1-x))) \stackrel{?} \le 2cu (1-x) + 0.5c(1-u)(1-x)+0.5xcu\,.
\end{align*}
Regrouping terms, we get
\begin{align*}
u^2 (-3 + 9x - 6x^2)
+ u(5 - 1.5c - 10x + cx + 7x^2 - 2x^3)
-0.5 c + 0.5 c x  \stackrel{?}\le 0\,.
\end{align*}
When $x = 0.5$, this expression is linear in $u$, and holds for all $u \le 1$ at $c = 1.2$.

For $x \in [0.24, 0.5)$, the expression is a negative quadratic. We consider the maximum value
\begin{align*}
    -0.5 c + 0.5 c x - \frac{(5 - 1.5c - 10x + cx + 7x^2 - 2x^3)^2}{4 (-3 + 9x - 6x^2)}
\end{align*}
By inspection, at $c = 1.2$, this is negative for all $x \in [0.24, 0.44]$.

For $x \in [0.44, 0.5)$, we consider the derivative at $u = 1$. It takes value
\begin{align*}
    2(-3 + 9x - 6x^2)
+(5 - 1.5c - 10x + cx + 7x^2 - 2x^3)
\end{align*}
This is positive for $x \in [0.41,0.5]$. So the maximum value of the expression when $u \le 1$ occurs at $u = 1$. Plugging in $u =1$, we see that 
\begin{align*}
(-3 + 9x - 6x^2)
+ (5 - 1.5c - 10x + cx + 7x^2 - 2x^3)
-0.5 c + 0.5 c x  \le 0\,.
\end{align*}
for all $x\le 0.5$.
\item When $v = x$ and $w = \frac{u}{1.2}$, the equation reads
\begin{align*}
    u\bar{x}(2 + xu(1 + 2x) + \bar{x}\bar{u}(1 + 2\bar{x})) 
    +u^2x(\frac{u}{1.2}-\frac{1}{2})x^2
    \stackrel{?} \le 2cu \bar{x} + c(1-\frac{u}{1.2})\bar{u}\bar{x}+\frac{u}{1.2}xcu\,.
\end{align*}
Regrouping terms and plugging in $c = 1.2$, we get
\begin{align*}
u^3 (\frac{5}{6} x^3)
+
u^2(-0.5x^3-6x^2+9x-4)
+
u(-2x^3+7x^2-9.8x+4.8)
-1.2 + 1.2 x
 \stackrel{?}\le 0\,.
\end{align*}
This is a cubic in $u$, with positive leading coefficient for $x \ge 0$. Consider the derivative evaluated at $u = 0.8$ and $u = 1$. These values are 
\begin{align*}
    &3\cdot 0.8^2  (\frac{5}{6} x^3)
    +
    2 \cdot 0.8(-0.5x^3-6x^2+9x-4)
    +
    (-2x^3+7x^2-9.8x+4.8)
    \\
    &3\cdot(\frac{5}{6} x^3)
    +
    2 \cdot (-0.5x^3-6x^2+9x-4)
    +
    (-2x^3+7x^2-9.8x+4.8)
\end{align*}
By inspection, both are negative for all $x \in [0, 1]$. So the expression is decreasing in the range $u \in [0.8, 1]$. So the maximum occurs at $u = 0.8$.  This takes value
\begin{align*}
    0.8^3 (\frac{5}{6} x^3)
+
0.8^2(-0.5x^3-6x^2+9x-4)
+
0.8(-2x^3+7x^2-9.8x+4.8)
-1.2 + 1.2 x
\end{align*}
By inspection, this is negative for all $x \in [0.24, 0.5]$.\qedhere
\end{itemize}
\end{proof}

\subsubsection{Level \texorpdfstring{$\ell=4$}{4}}

\begin{claim}
    \label{claim:adaptivitygap_t_is_4_branch1}
    Consider variables $v,s,x,u$, where $0.01 \le v \le 0.2$, and $0.5 \le s \le x \le u \le 0.65$. Then the function
    \begin{align*}
        f(v;s,x,u) \defeq \frac{
2 + uv(1+x+2xs) + (1-u)(1-v)(2-x+2(1-x)(1-s))
        }
        {
        2+\frac{(1-u)(1-x)}{s}+v
        }\,
    \end{align*}
    is at most $1.2$.
\end{claim}
\begin{proof}
This function is at most $c \defeq 1.2$ when
\begin{align*}
    2s + usv(1 + x + 2xs) + s\bar{u}\bar{v}(1 + \bar{x} + 2\bar{x}\bar{s}) \stackrel{?}\le c(2s + \bar{u}\bar{x} + sv)\,.
\end{align*}
    Both the numerator and denominator of $f$ are linear in $v$. Since any expression
$\frac{av+b}{cv+d}$ is affine to $\frac{1}{cv+d}$, this is a piecewise monotone function in $v$ on either side of the pole $v = \frac{-d}{c}$.
Since the denominator is positive for all $0 \le v \le 1$, the function is monotone in the region we are interested in. So it is maximized at a boundary point of $v$, either $v = 0.01$ or $v = 0.2$. At each boundary point of $v$, the same is true of $u$; so it is maximized at a boundary point (either $u = 0.65$ or $u = x$). We try all four cases:
\begin{itemize}
    \item Suppose $u = 0.65$ and $v = 0.01$. Then the function is monotone in $x$, so it is maximized at a boundary point ($x = u = 0.65$ or $x = s$). This gives two expressions:
    \begin{align*}
        f(0.01;s,0.65,0.65) = \frac{
        2.72105 - 0.2341 s
        }{
        2.01 + \frac{0.1225}{s}
        }
    \end{align*}
        \begin{align*}
        f(0.01;s,s,0.65) = \frac{
        3.3925 - 1.726 s + 0.706 s^2
        }{
        1.66 + \frac{0.35}{s}
        }
    \end{align*}
    By inspection, both expressions are at most $1.2$ for all $s \in [0,1]$.
    \item Suppose $u = 0.65$ and $v = 0.2$. Then the function is monotone in $x$, so it is maximized at a boundary point ($x = u = 0.65$ or $x = s$). This gives two expressions:
    \begin{align*}
        f(0.2;s,0.65,0.65) = \frac{
        2.7885 - 0.027 s
        }{
        2.2 + \frac{0.1225}{s}
        }
    \end{align*}
        \begin{align*}
        f(0.2;s,s,0.65) = \frac{
        3.25 - 1.27 s + 0.82 s^2
        }{
        1.85 + \frac{0.35}{s}
        }
    \end{align*}
    By inspection, both expressions are at most $1.2$ for all $s \in [0.5, 0.65]$.
    \item Suppose $u = x$ and $v = 0.01$. Then the function is at most $1.2$ exactly when
    \begin{align*}
        -1.2 + 3.548 s - 1.98 s^2 
        +x (2.4 - 6.92 s + 3.96 s^2)
        + x^2 (-1.2 + 2.98 s - 1.96 s^2)
        \stackrel{?} \le 0\,.
    \end{align*}
    For all $s$, this is a negative quadratic in $x$. Consider the derivative at $x = s$:
    \begin{align*}
        (2.4 - 6.92 s + 3.96 s^2) + 2s(-1.2 + 2.98 s - 1.96 s^2)\,,
    \end{align*}
    which is negative for all $s > 0.4$. So the maximum value of the function when $x \in [s, 0.65]$ is at $x = s$. This gives value
    \begin{align*}
        -1.2 + 3.548 s - 1.98 s^2 
        +s (2.4 - 6.92 s + 3.96 s^2)
        + s^2 (-1.2 + 2.98 s - 1.96 s^2)\,.
    \end{align*}
    By inspection, this is $0$ for all $s$. So the inequality holds here.
    \item Suppose $u = x$ and $v = 0.2$. Then the function is at most $1.2$ exactly when
    \begin{align*}
        -1.2 + 2.56 s - 1.6 s^2 
        + x (2.4 - 5.4 s + 3.2 s^2)
        + x^2 (-1.2 + 2.6 s - 1.2 s^2) \stackrel{?}\le 0\,.
    \end{align*}
    For $s \in [0.5, 0.65]$, this is a negative quadratic in $x$. Consider the derivative at $x = 0.65$:
    \begin{align*}
         (2.4 - 5.4 s + 3.2 s^2)
        + 2(0.65) (-1.2 + 2.6 s - 1.2 s^2)\,,
    \end{align*}
    which is positive for all $s$. So the maximium value of the function when $x \in [s,0.65]$ is at $x = 0.65$. This gives value
    \begin{align*}
        -1.2 + 2.56 s - 1.6 s^2 
        + 0.65 (2.4 - 5.4 s + 3.2 s^2)
        + 0.65^2 (-1.2 + 2.6 s - 1.2 s^2) \,,
    \end{align*}
    which is negative for all $s \le 1$. So the inequality holds here.
\end{itemize}
\end{proof}

\begin{claim}
     \label{claim:adaptivitygap_t_is_4_branch2}
    Consider variables $v,x,w,u$, where $0.8 \le u \le 1$, and $0.35 \le v \le x \le w \le 0.5$. Then the function
    \begin{align*}
        f(v,x,w;u) \defeq \frac{
2 + uv(1+x+2xw) + (1-u)(1-v)(2-x+2(1-x)(1-w))
        }
        {
        2+\frac{(1-u)(1-x)}{u}+vx + 2vxw
        }\,
    \end{align*}
    is at most $1.2$.
\end{claim}
\begin{proof}
This function is at most $c \defeq 1.2$ when
\begin{align*}
    2u + u^2v(1 + x + 2xw) + u\bar{u}\bar{v}(1 + \bar{x} + 2\bar{x}\bar{w}) \stackrel{?}\le c(2u + \bar{u}\bar{x} + uvx + 2uvxw)\,.
\end{align*}

Both the numerator and denominator of $f$ are linear in $w$. Since any expression
$\frac{aw+b}{cw+d}$ is affine to $\frac{1}{cw+d}$, this is a piecewise monotone function in $w$ on either side of the pole $w = \frac{-d}{c}$.
Since the denominator is positive for all $0 \le w \le 1$, the function is monotone in the region we are interested in. So it is maximized at a boundary point of $w$, either $w = x$ or $w = \frac{1}{2}$. At each boundary point of $w$, the same is true of $v$; so it is maximized at a boundary point (either $v = 0.35$ or $v = x$). We try all four cases:
\begin{itemize}
    \item Suppose $v = 0.35$ and $w = 0.5$. Then the function is monotonic in $x$, so it is maximized at a boundary point (either $x = v = 0.35$ or $x = w = 0.5$). We consider each expression:
    \begin{align*}
        f(0.35,0.35,0.5;u) = 
        \frac{3.495-0.9u }{
        1.595 + \frac{0.65}{u}
        }
    \end{align*}
    \begin{align*}
        f(0.35,0.5,0.5;u) = \frac{3.3 - 0.6u}{
        1.85 + \frac{1}{2u}
        }
    \end{align*}
    By inspection, both are at most $c = 1.2$ for all $u \ge 0$.
    \item Suppose $v = x = w$. The function is at most $c = 1.2$ when
    \begin{align*}
        -1.2 + 1.2 x + u (4.8 - 10.2 x + 5.8 x^2 - 4.4 x^3) + u^2 (-4 + 10 x - 6 x^2 + 4 x^3) \stackrel{?} \le 0\,.
    \end{align*}
    At $x = 0.5$, this is the expression $0.6 u - 0.6 \stackrel{?} \le 0$, which is true for all $u \le 1$.

    For $x < 0.5$, the expression is a negative quadratic in $u$. The maximum value is
    \begin{align*}
        -1.2 + 1.2 x  - \frac{(4.8 - 10.2 x + 5.8 x^2 - 4.4 x^3)^2}{4(-4 + 10 x - 6 x^2 + 4 x^3)}\,,
    \end{align*}
    which is negative for $x \in [0.35, 0.44]$. So the inequality holds here.

    Now we handle $x \in [0.44, 0.5)$. Consider the derivative at $u = 1$:
    \begin{align*}
        (4.8 - 10.2 x + 5.8 x^2 - 4.4 x^3) + 2(1) (-4 + 10 x - 6 x^2 + 4 x^3)\,,
    \end{align*}
    which is positive when $x \ge 0.41$. So the maximum value of the function when $u \in [0.8, 1]$ occurs at $u = 1$. This gives value
    \begin{align*}
         -1.2 + 1.2 x + (4.8 - 10.2 x + 5.8 x^2 - 4.4 x^3) +  (-4 + 10 x - 6 x^2 + 4 x^3)\,,
    \end{align*}
    which is at most $0$ for $x \in [0.44, 0.5]$. So the inequality holds here as well.
    \item Suppose $v = 0.35$ and $w = x$. The function is at most $1.2$ when
    \begin{align*}
    x^2( 0.46 u - 0.6 u^2)
    + x(1.2 - 4.87 u  + 3.6 u^2)
        -1.2 + 3.4 u - 2.25 u^2  \stackrel{?} \le 0\,.
    \end{align*}
    For $u \ge 0.8$, this is a negative quadratic in $x$. Consider the derivative at $x = 0.35$:
    \begin{align*}
        2(0.35)(0.46u-0.6u^2) + (1.2-4.87u+3.6u^2)\,.
    \end{align*}
    This is negative for all $u \in [0.8,1]$. So the function is maximized in the region $x \in [0.35, 0.5]$ at $x = 0.35$. Here, $v = x = w$, and we have already handled this case.
    \item Suppose $v = x$ and $w = 0.5$. The function is at most $c = 1.2$ when
    \begin{align*}
    x^2(-0.4u)
    + x(1.2 - 6.2u + 6u^2)
    -1.2 + 3.8u - 3u^2
    \stackrel{?} \le 0\,.
    \end{align*}
    For $u \ge 0$, this is a negative quadratic in $x$. The maximum value is
    \begin{align*}
        -1.2 + 3.8u - 3u^2 - \frac{(1.2 - 6.2u + 6u^2)^2}{4(-0.4u)}\,,
    \end{align*}
    which is negative for all $u \in [0.8,0.9]$. So the inequality holds here.
    
    Now we handle $u \in [0.9, 1]$. Consider the derivative at $x = 0.5$:
    \begin{align*}
        2(0.5)(-0.4u) + (1.2 - 6.2u + 6u^2)\,.
    \end{align*}
    When $u \ge 0.871$, this is positive, so in this case, the function is maximized in the region $x \in [0.35, 0.5]$ at $x = 0.5$. Here, $v = x = w$, and we have already handled this case.\qedhere
\end{itemize}
\end{proof}

\begin{restatable}{claim}{fourcoinonepointtwo}
\label{claim:4coin_1.2}
Consider variables $v,s,x,u$, where $0 \le v \le 0.2$, $0.5 \le s \le x \le u \le 0.65$, and $0.58 \le u \le 0.65$. Then the function
    \begin{align*}
        f(v,s,x,u) \defeq \frac{
2 + uv(1+x+2xs) + (1-u)(1-v)(2-x+2(1-x)(1-s))
        }
        {
        2+(1-x)\frac{1-u}{s}+x\frac{v}{1-v}
        }\,
    \end{align*}
    is at most $1.2$.
\end{restatable}
\begin{proof}
Both the numerator and denominator are linear in $x$. Since any expression
$\frac{ax+b}{cx+d}$ is affine to $\frac{1}{cx+d}$, this is a piecewise monotone function in $x$ on either side of the pole $x = \frac{-d}{c}$.
Since the denominator is positive for all $0 \le x \le 1$, the function is monotone in the region we are interested in. So it is maximized at a boundary point of $x$, either $x = s$ or $x = u$.

We would like to show that the function is at most $c \defeq 1.2$. This is true exactly when
\begin{align}
\label{eqn:4coin_inequality}
    s\bar{v}\left(2 + uv(1+x+2xs) + \bar{u}\bar{v}(1+\bar{x}+2\bar{x}\bar{s}\right)
    \stackrel{?}\leq c \left(
    2s\bar{v}+\bar{x}\bar{u}\bar{v}+xvs
    \right)\,.
\end{align}
Grouping terms by $\bar{v}$, we get
\begin{align*}
    \bar{v}^2
    \left(
s\bar{u}(1+\bar{x}+2\bar{x}\bar{s})
-
su(1+x+2xs)
    \right)
    +\bar{v}
    \left(
    2s+us+usx+2uxs^2
    -2cs
    +cxs
    -c\bar{x}\bar{u}
    \right)
-cxs
\stackrel{?}\le 0\,.
\end{align*}
This is quadratic in $\bar{v}$ even at boundary conditions $x = s$ and $x = u$.
Recall that $u,x,s \ge 0.5$. Since $u \ge 0.58$, the quadratic coefficient is negative. So this inequality is satisfied when the discriminant is negative, i.e.
\begin{align*}
    \left(
    2s+us+usx+2uxs^2
    -2cs
    +cxs
    -c\bar{x}\bar{u}
    \right)^2 - 4 \left(
s\bar{u}(1+\bar{x}+2\bar{x}\bar{s})
-
su(1+x+2xs)
    \right)(-cxs) \stackrel{?}\leq 0\,.
\end{align*}
We try both boundary conditions $x = s$ and $x = u$.
In the first case, the term is sextic in $s$ and quadratic in $u$; in the second, it is quartic in $s$ and quartic in $u$.

Let's first set $x = s$. Grouping terms by $u$, we get 
\begin{align}
\label{eqn:4coin_byu_quadratic}
\left(
u(c+(1-c)s+s^2+2s^3)
+
cs^2+2s-cs-c
\right)^2
+
4cs^3
\left(
\bar{u}(1+\bar{s}+2\bar{s}^2)
-
u(1+s+2s^2)
\right)
\stackrel{?} \leq 0\,.
\end{align}
This is quadratic in $u$, and the quadratic coefficient is a square (so is non-negative). In fact, it is positive since $c > 0$.

We would like to check that the quadratic \cref{eqn:4coin_byu_quadratic} is negative for all $0.58 \le u \le 0.65$. Since it has a positive quadratic coefficient, we can just check that it is negative at the boundary points $u =  0.58$ and $u = 0.65$ are negative to ensure it is negative for the whole region. If we set $u = 0.58$ and $u = 0.65$ (and $c = 1.2$), we get that
\begin{align*}
0.254016-0.689472s-1.32638s^{2}+6.54576s^{3}-8.10872s^{4}+2.5936s^{5}+1.3456s^{6} &\stackrel{?} \le 0\,.
\\
0.1764 - 0.5628 s - 1.1051 s^2 + 4.987 s^3 - 6.3555 s^4 + 1.93 s^5 + 1.69 s^6
&\stackrel{?}\leq 0
\end{align*}
By inspection, both are negative within the region $0.5 \le s \le 0.65$, with no roots in this region.

Now we set $x = u$. Grouping terms in \cref{eqn:4coin_inequality} by $s$, we get a quartic:
\begin{align*}
    0 \stackrel{?} \ge&\ s^4( 4u^4)
    \\
    &+
    s^3(-12cu^3 + 8cu^2 -8cu + 4u^4 + 4u^3 + 8u^2)
    \\
    &+
    s^2
    (c^2u^2 - 4c^2u + 4c^2 - 4cu^4 + 18cu^3 - 38cu^2 + 16cu - 8c + u^4 + 2u^3 + 5u^2 + 4u + 4)
    \\
    &+
    s(-2c^2u^3 + 8c^2u^2 - 10c^2u + 4c^2 - 2cu^4 + 2cu^3 - 2cu^2 + 6cu - 4c)
    \\
    &+
c^2u^4 - 4c^2u^3 + 6c^2u^2 - 4c^2u + c^2 \,.
\end{align*}
We calculate at the discriminant of the quartic as a polynomial of $u$ when $c = 1.2$.
By inspection, the discriminant is negative for all $0.56 \le u \le 0.82$. So the quartic has only two real roots.

Notice that the quartic has a positive quartic coefficient. 
We calculate the value of the quartic at the boundary points $s = 0.5$ and $s = 0.65$. If the quartic is negative at both points, then the two roots $r_1, r_2$ must be outside the region $0.5 \le s \le 0.65$. At these boundary points, the quartic evaluates to
\begin{align*}
    1.96 - 6.2 u + 5.61 u^2 - 1.4 u^3 + 0.04 u^4 &\stackrel{?} \leq 0
    \\
     2.1316 - 5.708 u + 2.8563 u^2 + 1.0429 u^3 + 0.087025 u^4 &\stackrel{?}\leq 0
\end{align*}
By inspection, both are negative in the region $0.58 \le u \le 0.8$. So the claim holds for $0.58 \le u \le 0.65$.
\end{proof}

\subsubsection{Level \texorpdfstring{$\ell=5$}{5}}

\begin{restatable}{claim}{fivecoinonepointtwo}
\label{claim:5coin_1.2}
    Consider variables $v,r,s,t,u$, where $0 \le v \le 0.2$, $0.5 \le r \le s \le t \le u \le 0.58$, and $u > 0.5$. Then the function
    \begin{align*}
        f(v,r,s,t,u) \defeq \frac{
2 + uv(1+t+ts(1+2r)) + (1-u)(1-v)(1+(1-t)+(1-t)(1-s)(1+2(1-r)))
        }
        {
        2+(1-t)(1-u)(1+\frac{1-s}{r})+t\frac{v}{1-v}
        }\,
    \end{align*}
    is at most $1.2$.
\end{restatable}
\begin{proof}
Both the numerator and denominator are linear in $s$ and $u$. Since any expression
$\frac{ax+b}{cx+d}$ is affine to $\frac{1}{cx+d}$, this is a piecewise monotone function in $s$ and in $u$ on either side of the pole $x = \frac{-d}{c}$.
Since the denominator is positive for all $0 \le s,u \le 1$, the function is monotone in the region we are interested in. So, both $s$ and $u$ are maximized at a boundary point. In the proof, we fix $s$ and $u$ at each boundary point, creating four options from $r \le s \le t \le u \le 0.58$. 

We would like to show that the function is at most $c \defeq 1.2$. This is true exactly when
\begin{align*}
    r\bar{v}\bigg(
    2+uv(1+t(1+s(1+2r)))
    +\bar{u}\bar{v}(1+\bar{t}(1+\bar{s}(1+2\bar{r})))
    \bigg)
     \stackrel{?}\leq 
    c\left(
2r\bar{v}+\bar{v}\bar{t}\bar{u}(r+\bar{s})+rtv
    \right)
   \,.
\end{align*}
Grouping terms by $\bar{v}$, we get
\begin{align}
\label{eqn:5coin_main}
0 \stackrel{?}\ge &\ \bar{v}^2
\big(
r\bar{u}(1+\bar{t}(1+\bar{s}(1+2\bar{r})))
-
ru(1+t(1+s(1+2r)))
\big) \nonumber
\\
&+ \bar{v}
\big(
2r
+ru(1+t(1+s(1+2r)))
-2cr
-c\bar{t}\bar{u}(r+\bar{s})
+crt
\big)
-crt\,.
\end{align}
Recall that $r,s,t,u \ge 0.5$, so the quadratic coefficient is at most $0$. Since $u > 0.5$ it is in fact negative.

We will prove that \cref{eqn:5coin_main} is satisfied for all $0 \le \bar{v} \le 1$ in the following way: We show that the value of \cref{eqn:5coin_main} at $\bar{v} = 1$ is at most $0$, and that it is increasing at $\bar{v} = 1$. Since the quadratic coefficient is negative, the root of the quadratic is at $\bar{v} > 1$, and so \cref{eqn:5coin_main} is at most $0$ for the whole region $0 \le \bar{v} \le 1$.

Note that the equation is quadratic in $\bar{v}$ even if we plug in boundary conditions for $s$ and $u$; we consider each case.

Let's first verify \cref{eqn:5coin_main} is at most $0$ at $\bar{v} = 1$. Here, it takes value
\begin{align}
\label{eqn:5coin_case1_byr}
0 \stackrel{?}\ge&\  r\bar{u}(1+\bar{t}(1+\bar{s}(1+2\bar{r})))
+
2(1-c)r
-c\bar{t}\bar{u}(r+\bar{s}) \nonumber
\\
&= r^2(-2\bar{u}\bar{t}\bar{s})
    +r(\bar{u}(1+\bar{t}(1+3\bar{s}))+2(1-c)-c\bar{t}\bar{u})
    -c\bar{t}\bar{u}\bar{s}\,.
\end{align}
\begin{itemize}
    \item We first consider the boundary when $s = t$. Then \cref{eqn:5coin_case1_byr} is quadratic in $r$, with negative quadratic coefficient. We check that its value at the uppermost point of $r$ in the region $0.5 \le r \le s=t$ is at most $0$, and that it is increasing at said point. If so, then \cref{eqn:5coin_case1_byr} is at most $0$ for all $r$ in the region. Plugging in the point $r = t$ at $c = 1.2$, we see
\begin{align*}
0 \stackrel{?} \ge &\    t^2(-2\bar{u}\bar{t}^2)
    +t(\bar{u}(1+\bar{t}(1+3\bar{t}))+2(1-c)-c\bar{t}^2)
    -c\bar{t}^2\bar{u}
    \\
    &=
    -1.2+1.2u + t(5.8-6.2u)+9t^{2}(-1+u)-7t^{3}(-1+u)+2t^{4}(-1+u)
\end{align*}
By inspection, this is at most $0$ for all $0.5 \le t \le 0.7$ both when $u = t$ and $u = 0.58$ (the boundary points of $u$).

Now we check that it is increasing  at $r = t$. The derivative is 
\begin{align*}
    2r(-2\bar{u}\bar{t}\bar{s})
    +(\bar{u}(1+\bar{t}(1+3\bar{s}))+2(1-c)-c\bar{t}\bar{u})\,,
\end{align*}
which at $r = t$ and $c = 1.2$ evaluates to 
\begin{align*}
    3.4 - 11 t^2 (-1 + u) + 4 t^3 (-1 + u) - 3.8 u + t (-9.8 + 9.8 u)\,,
\end{align*}
which by inspection is positive for all $0.5 \le t \le 0.6$ both when $u = t$ and $u = 0.58$ (the boundary points of $u$).
\item Now let's consider the boundary when $s = r$. Then \cref{eqn:5coin_case1_byr} is cubic in $r$, i.e.
\begin{align}
\label{eqn:5coin_case2_byr}
    r^3(2\bar{u}\bar{t})
    +
    r^2(-5\bar{u}\bar{t})
    +
    r(\bar{u}(1+4\bar{t})+2(1-c))
    -c\bar{t}\bar{u}
    \stackrel{?}\leq 0\,.
\end{align}
We first show that \cref{eqn:5coin_case2_byr} is increasing for all $0.5 \le r \le 0.6$. The derivative has value
\begin{align*}
    3r^2(2\bar{u}\bar{t})
    +
    2r(-5\bar{u}\bar{t})
    +
    (\bar{u}(1+4\bar{t})+2(1-c)) \stackrel{?}> 0\,.
\end{align*}
We will show \emph{this} expression (the derivative of \cref{eqn:5coin_case2_byr}) has positive value and is decreasing at $r = 0.6$. Since the quadratic coefficient is positive, this implies \emph{this} expression is positive for $0.5 \le r \le 0.6$. 
\begin{itemize}
    \item Plugging in $r = 0.6$ and $c = 1.2$, this has value $0.76 + t (-0.16 + 0.16 u) - 1.16 u$, which by inspection is positive for $0.5 \le t \le 0.6$ both when $u = t$ and $u = 0.58$ (boundary points of $u$).
    \item The derivative of the expression (e.g. second derivative of \cref{eqn:5coin_case2_byr}) at $r = 0.6$ and $c = 1.2$ is exactly $6\cdot0.6\cdot(2(1-u)(1-t))+2(-5(1-u)(1-t))$, which by inspection is negative for $0.5 \le t \le 0.6$ both when $u = t$ and $u = 0.58$ (the boundary points of $u$).
\end{itemize}
Since \cref{eqn:5coin_case2_byr} is increasing for all $0.5 \le r \le 0.6$, we can inspect the upper-most value of $r$ (which equals $t$). Again, this evaluates to
\begin{align*}
    -1.2 + 1.2 u + t (5.8 - 6.2 u) + 9 t^2 (-1 + u) - 7 t^3 (-1 + u) + 2 t^4 (-1 + u)\,,
\end{align*}
which by inspection is at most $0$ for all $0.5 \le t \le 0.7$ both when $u = t$ and $u = 0.58$ (the boundary points of $u$). So \cref{eqn:5coin_case2_byr} is always satisfied.
\end{itemize}
Now let's verify \cref{eqn:5coin_main} is increasing at $\bar{v} = 1$. Here, the derivative takes value
\begin{align}
\label{eqn:5coin_secondpart}
 0 \stackrel{?} < &\  
2r\bar{u}(1+\bar{t}(1+\bar{s}(1+2\bar{r})))
-
ru(1+t(1+s(1+2r)))
+
2(1-c)r
-c\bar{t}\bar{u}(r+\bar{s})
+crt
\nonumber \\
&= 
r^2
(-2uts-4\bar{u}\bar{t}\bar{s})
+
r(
2\bar{u}(1+\bar{t}(1+3\bar{s}))
-u(1+t(1+3s))
+2(1-c)
+ct
-c\bar{t}\bar{u}
)
-c\bar{t}\bar{u}\bar{s}
\,.
\end{align}
\begin{itemize}
    \item We first consider the boundary when $s = t$. Then  \cref{eqn:5coin_secondpart} is quadratic in $r$, with negative quadratic coefficient. Since we want to test whether it is positive in the region $0.5 \le r \le 0.6$, we can evaluate this at both $r = 0.5$ and $r = 0.6$; if they both are positive, then it is positive in the region. Evaluating $s = t$ and $c = 1.2$, we get
    \begin{align*}
        0 \stackrel{?}<&\ 2 + t^2 (0.8 - 1.8 u) - 2.7 u + t (-1.4 + 1.5 u) \\
        0 \stackrel{?}<&\ 2.4 + t^2 (0.96 - 2.28 u) - 3.24 u + t (-1.68 + 1.8 u)
    \end{align*}
    By inspection, both are positive for $0.5 \le t \le 0.585$, when $u = t$ and $u = 0.58$ (boundary points of $u$).
    \item Now let's consider the boundary when $s = r$. Then \cref{eqn:5coin_secondpart} is cubic in $r$, i.e.
    \begin{align*}
        r^3(-2ut+4\bar{u}\bar{t})
        +
        r^2(
        -10\bar{u}\bar{t}
        -ut
        )
        +
        r(
        2\bar{u}(1+4\bar{t})
        -
        u(1+t)
        +2(1-c)
        +ct
)
        -c\bar{t}\bar{u}
        \stackrel{?}>0\,.
    \end{align*}
    We verify by inspection that the cubic term is positive for all $0.5 \le t,u \le 0.58$. 
    So we can evaluate the cubic at $r = 0.5$, $r = 0.6$, and $r = 1$. If it is negative at $r = 1$ but positive at $r \in \{0.5,0.6\}$, then it must be positive for all $0.5 \le r \le 0.6$.

    Evaluating at $r \in \{0.5,0.6,1\}$ and $c = 1.2$, we get
    \begin{align*}
    0 \stackrel{?}<&\ 1.6 + t (-0.2 - 0.2 u) - 2.3 u
    \\
    0 \stackrel{?}<&\ 1.824 + t (-0.144 - 0.528 u) - 2.664 u
    \\
       0 \stackrel{?}>&\  2.4 + t (0.4 - 3.2 u) - 3.8 u \,.
    \end{align*}
By inspection, the first two expressions are positive for $0.5 \le t \le 0.585$ both when $u = t$ and $u = 0.58$ (the boundary points of $u$). Similarly, the last  expression is negative for $0.5 \le t$ at both boundary points of $u$.   
\end{itemize}
Altogether, the quadratic coefficient of \cref{eqn:5coin_main} is negative, and the root of the quadratic is at $\bar{v} \ge 1$ (since \cref{eqn:5coin_main} is at most $0$ and increasing at $\bar{v} = 1$). So \cref{eqn:5coin_main} is at most $0$ for the whole region $0 \le \bar{v} \le 1$, assuming $0.5 \le r \le s \le t \le u \le 0.58$. This proves the claim.
\end{proof}

\end{document}